\documentclass[11pt]{article}
\usepackage{amsthm,amsgen,amsmath,amstext,amsbsy,amsopn,amssymb,latexsym,mathrsfs}
\usepackage{array}
\usepackage{cite}
\usepackage{multirow}
\usepackage{graphicx}
\usepackage{amssymb,bm,bbm}
\usepackage{titling}
\usepackage{url,bm}
\usepackage{algpseudocode,float,algorithm}
\usepackage[round]{natbib}
\usepackage{bbm}
\usepackage{verbatim}
\usepackage[normalem]{ulem}
\usepackage{enumerate}
\usepackage{subcaption}
\usepackage{esdiff}
\usepackage[dvipsnames,table,dvipsnames*, svgnames*]{xcolor}

\usepackage{hyperref}
\hypersetup{
    colorlinks=true,
    linkcolor=blue,
    filecolor=blue,
    urlcolor=blue,
    citecolor=blue,
}
\usepackage{cleveref}

\allowdisplaybreaks

\newcommand{\argmin}{\mathop{\rm arg\min}}

\theoremstyle{plain}

\newtheorem{proposition}{Proposition}
\newtheorem{lemma}{Lemma}
\newtheorem{definition}{Definition}
\newtheorem{theorem}{Theorem}
\newtheorem{example}{Example}
\newtheorem{remark}{Remark}

\newtheorem{claim}{Claim}
\newtheorem{assumption}{Assumption}

\numberwithin{intassumption}{assumption}

\newcommand{\R}{{\mathbb{R}}}
\newcommand{\E}{{\mathbb{E}}}

\newcommand{\Prob}{{\mathbb{P}}}
\newcommand{\1}{{\mathbbm{1}}}

\newcommand{\Acal}{{\mathcal{A}}}

\newcommand{\Ccal}{{\mathcal{C}}}
\newcommand{\Dcal}{{\mathcal{D}}}

\newcommand{\Fcal}{{\mathcal{F}}}
\newcommand{\Gcal}{{\mathcal{G}}}

\newcommand{\Mcal}{{\mathcal{M}}}

\newcommand{\Ecal}{\mathcal{E}}

\newcommand{\Vcal}{\mathcal{V}}

\newcommand{\Var}{{\rm Var}}
\newcommand{\Pcal}{{\mathcal{P}}}

\newcommand{\indep}{\mathrel{\!\perp\!\!\!\perp}}
\newcommand{\parent}{{\rm Pa}}

\newcommand{\ancestor}{{\rm An}}
\newcommand{\anterior}{{\rm Ant}}

\newcommand{\IMEC}{\Mcal_{\bm I}}
\newcommand{\Aug}{{\rm Aug}}
\newcommand{\MAG}{{\rm MAG}}

\newcommand{\CM}{{\rm CM}}
\newcommand{\TCM}{{\rm TCM}}

\graphicspath{{Figures/}}

\title{Statistical Inference for Causal Discovery under Selection and Latent Variables via Single-Target Interventions}
\author{Xiaotian Hou \and Kwangmoon Park \and Hongzhe Li
\\[1ex]
Department of Biostatistics, Epidemiology and Informatics\\
University of Pennsylvania
}

\date{}

\begin{document}

\maketitle

\begin{abstract}
    Causal discovery from observational and interventional data becomes substantially more challenging in the presence of latent confounding and selection bias, where causal structure is no longer adequately represented by directed acyclic graphs over observed variables. Existing model-free methods often rely on an exponential number of conditional independence tests and provide limited uncertainty quantification in high-dimensional settings. We develop a model-free and constraint-query optimal statistical inference framework for causal discovery under latent variables and selection using single-target interventions. We introduce the system-induced subgraph (SIS) to capture the causal relations among system variables while accounting for context variables. We establish the identifiability of SISs through maximal ancestral graphs (MAGs), and show that interventions on each observed system variable are necessary and sufficient for unique identification of the MAG. Building on these results, we develop statistically valid graph inference procedures with family-wise error control and establish constraint-query optimality up to multiplicative constants, requiring at most $\frac{5}{2}d_X^2$ parallel statistical tests for $d_X$ observable variables. The framework accommodates soft interventions and avoids parametric structural equation assumptions. We illustrate the methods through analysis of Perturb-seq data from interferon-$\beta$-stimulated A549 lung cancer cell lines.
\end{abstract}

\tableofcontents

\section{Introduction}
Single-cell CRISPR screens, such as perturb-seq \citep{dixit2016perturb, replogle2022mapping, jiang2025systematic}, measure gene expression under targeted interventions and offer a promising setting for causal network reconstruction. Combining observational and interventional variation can resolve causal ambiguities that remain in observational data. However, in practice, only a subset of relevant genes and regulatory factors is measured, and the observed cells are often subject to selection mechanisms arising from experimental design, cell survival, and preprocessing. As a result, the data are more appropriately generated from a latent-variable directed acyclic graph (DAG) under selection, for which the induced structure over observed variables need not be a DAG. Ignoring such selections or latent factors can produce spurious adjacencies and misleading orientations.

Three issues arise in this setting: \textit{what causal structure is identifiable, how efficiently it can be recovered, and how uncertainty can be quantified.} Existing work characterizes the interventional Markov equivalence classes (IMECs) under latent variables and, more recently, post-treatment selection \citep{kocaoglu2019characterization,luo2025characterization}. Nevertheless, unique identification and the interventions needed to achieve such an identification require further study. Moreover, constraint-based procedures may require exponentially many sequential conditional independence (CI) tests in the worst case \citep{spirtes2001anytime,kocaoglu2019characterization,luo2025characterization,yang2018characterizing}, limiting their applicability to high-dimensional settings. In these algorithms, because later tests depend on earlier decisions, error can propagate through the procedure, making uncertainty quantification challenging. 

\noindent\textbf{Contributions.} We address these challenges by developing a model-free and constraint-query optimal statistical inference framework for causal discovery under latent variables and selection bias using single-target interventions. Specifically, to capture the causal relations among system variables while accounting for context variables, we introduce a system-induced subgraph (SIS) formulation and characterize its identifiability under single-target interventions through maximal ancestral graphs (MAGs) \citep{richardson2002ancestral,spirtes1997polynomial}. In particular, we show that interventions on every observed system variable are sufficient for identifying the trimmed Combined MAG and necessary in the worst case. To our knowledge, this is the first sufficient and worst-case necessary condition on single-target intervention for unique identification of MAGs under selection and latent variables. 

Building on these results, we propose a two-stage constraint-based inference procedure for the MAG with asymptotic family-wise error rate (FWER) control under sufficient first-stage power. For $d_X$ observed system variables, our algorithm requires at most $\frac{5}{2}d_X^2$ tests that are parallelizable within each stage. A matching lower bound of $\frac{1}{2}d_X(d_X-1)$ demonstrates the constraint-query optimality. To our knowledge, our method is the first model-free algorithms that recover the MAG in $O(d_X^2)$ constraint queries, and we establish the first constraint-query optimality for MAGs under interventions. Furthermore, our method provides the first model-free statistical inference framework for MAGs with FWER control. 

\noindent\textbf{Related Work.} Without selection or latent variables, observational data identify a DAG only up to its Markov equivalence class (MEC), characterized by its skeleton and $v$-structures \citep{verma1991equivalence}. \cite{hauser2012characterization} shows that hard interventions refine the MEC, and \cite{yang2018characterizing} extends the characterization to soft interventions using augmented graphs. \cite{eberhardt2006n} shows that $d_X-1$ single-target interventions are sufficient and necessary to uniquely identify the DAG. 
In the presence of selection and latent variables, the framework of MAGs \citep{spirtes1997polynomial,richardson2002ancestral} provides a principled representation of the MEC and CI relations among observable variables. 
Their interventional extensions characterize the IMEC through MAGs of augmented graphs \citep{kocaoglu2019characterization,luo2025characterization}. However, it remains unclear whether the causal structure is uniquely identifiable and which interventions suffice.

Many algorithms have been proposed for learning the causal structure from data. Constraint-based procedures include the Peter-Clark (PC) algorithm \citep{spirtes2000causation} for DAGs, fast causal inference (FCI) \citep{spirtes1995causal,spirtes2001anytime} under latent variables and selection, and their interventional extensions \citep{kocaoglu2019characterization,luo2025characterization}. Their worst-case searches can require exponentially many sequential CI tests. Under observational data, \cite{chickering2004large} and \cite{mones2026number} show the necessity of exponential computational complexity for learning MEC. With sufficient interventions and no selection or latent variables, DAGs admit an $O(d_X^5)$ algorithm \citep{shiragur2024causal}. However, the optimal computational complexity for IMEC remains unknown.

Score-based algorithms, including greedy equivalence search (GES) \citep{chickering2002optimal} and greedy interventional equivalence search (GIES) \citep{hauser2012characterization}, and structural restrictions such as non-Gaussian linear \citep{shimizu2006linear} or nonlinear additive-noise models \citep{hoyer2008nonlinear}, offer alternative identification and learning strategies. Although these assumptions enable stronger theoretical guarantees, the resulting methods may be sensitive to model misspecification.

A closely related work is \cite{sadeghi2025axiomatization}, which axiomatizes causal relations from per-node single-target hard interventional distributions using the bowless directed mixed graphs (BDMGs) that allow latent variables and cycles. Our work differs in several important aspects. \cite{sadeghi2025axiomatization} focused on axiomatizing the causal structure from hard interventional distributions, whereas our objective is statistical inference from finite observational and interventional samples. We additionally allow selection mechanisms and general soft interventions, which arise naturally in Perturb-seq. Finally, whereas their primary graphical representation is bowless, our use of intervention indicators retains additional information about latent confounding that may coexist with causal relations between observed variables.

Another closely related work is \cite{luo2025characterization}, which characterizes the IMEC using the MAG of a DAG augmented with all intervention indicators and develops the $\Fcal$-FCI algorithm. We instead use a collection of MAGs of single-indicator augmented DAGs, avoiding additional inter-indicator constraints. Whereas $\Fcal$-FCI accommodates partial intervention coverage, our main algorithm exploits per-node single-target interventions to learn anterior sets that determine both conditioning sets and edge orientations, thereby avoiding a combinatorial search over conditioning sets.

\noindent\textbf{Organization.} The remainder of the paper is organized as follows. Section \ref{sec:preliminaries} introduces the data setting and MAG preliminaries. Section \ref{sec:SIS} introduces the system-induced subgraph and studies its identification and interpretation. Sections \ref{sec:latent_selection} and \ref{sec:selection} present inference procedures with and without latent variables, respectively. Section \ref{sec:numerical} reports simulation studies and a real data application. The supplement \citep{hou2026supplement} contains omitted sections about examples, interpretations of the learned graph, CTS tests, an algorithm in the absence of selection, and simulation setup, together with all proofs.

\section{Preliminaries}\label{sec:preliminaries}

\subsection{Notations}

We use uppercase letters (e.g., $X_j$) for random variables and boldface symbols (e.g., $\bm X=(X_1,\ldots,X_d)$) for random vectors. Let $\bm X_{-j}$ denote all variables excluding $X_j$, i.e., $\bm X_{-j}=(X_1,\ldots,X_{j-1},X_{j+1},\ldots,X_d)$, and denote $\bm X_{[K]}=(X_1,\ldots,X_K)$.

For a MAG $\Gcal=(\Vcal,\Ecal)$, write $\Vcal(\Gcal)$ and $\Ecal(\Gcal)$ for its vertices and edges. We say two vertices $X_j$ and $X_k$ are adjacent if there is an edge between them. Let $\parent_{\Gcal}(X),\ancestor_{\Gcal}(X)$, and $\anterior_{\Gcal}(X)$ denote the parent, ancestor, and anterior sets of $X$, respectively, where $X'$ is anterior to $X$ if they are connected by a path consisting of undirected edges and directed edges oriented toward $X$. For a set of vertices $\bm X'$, define $\ancestor_{\Gcal}(\bm X')=\bigcup_{X\in\bm X'}\ancestor_{\Gcal}(X)$ and $\anterior_{\Gcal}(\bm X')=\bigcup_{X\in\bm X'}\anterior_{\Gcal}(X)$. We follow the standard definition of paths, colliders, $d$-separation, and $m$-separation \citep{richardson2002ancestral}. We write $(X_j\indep X_k\mid\bm X')_{\Gcal}$ for $d$-separation in a DAG and $m$-separation in a MAG, and use $X_j\indep X_k\mid \bm X'$ for conditional independence of random variables. Two MAGs $\Gcal$ and $\Gcal'$ are $m$-separation equivalent if they share the same nodes $\Vcal(\Gcal)=\bm X=\Vcal(\Gcal')$ and for any disjoint sets $\bm X_1,\bm X_2,\bm X_3\subseteq\bm X$, $(\bm X_1\indep \bm X_2\mid\bm X_3)_{\Gcal}$ if and only if $(\bm X_1\indep \bm X_2\mid\bm X_3)_{\Gcal'}$. 

The family-wise type I error rate (${\rm FWER}_1$) refers to the probability of making at least one false discovery among a collection of hypothesis tests. Analogously, the family-wise type II error rate (${\rm FWER}_2$) is the probability of making at least one false nondiscovery.

\subsection{Problem Setup}

Let $\bm Z=(\bm C, \bm X, \bm L, \bm S)\in\R^{d_Z}$, where $\bm C\in\R^{d_C}$ denotes context variables, $\bm X\in\R^{d_X}$ contains the system variables of primary interest, $\bm L\in\R^{d_L}$ denotes latent variables, and $\bm S\in\{0,1\}^{d_S}$ denotes $d_S$ selection criteria. Here, only $(\bm C,\bm X)$ are observed for the retained samples. In single-cell gene expression data, $\bm C$ may include batch, library size, and measured cell states, $\bm X$ are gene expression levels, and $\bm L$ may include unmeasured genes, cell-cycle states, or chromatin states.

Under the control regime, such as the control cells in perturb-seq experiments, $\bm Z$ follows a distribution $P_{\bm Z}^{(0)}$ generated according to a DAG $\Gcal=(\bm Z,\Ecal)$,
\[p_{\bm Z}^{(0)}=\prod_{j\in [d_Z]}p_{Z_j\mid\parent_\Gcal(Z_j)}^{(0)},\quad p_{\bm Z}^{(0)}=\frac{dP_{\bm Z}^{(0)}}{d\mu},\]
where $\parent_\Gcal(Z_j)$ denotes the parents of $Z_j$ in $\Gcal$ and $\mu$ is some product dominating measure. 

We model selection as the final stage of data generation and assume $\bm S$ has no children in $\Gcal$. A sample $(\bm C,\bm X)$ is retained if and only if $\bm S=\bm 1$. A selection criterion $S_j$ may depend on both observed variables $\bm X$ and latent variables $\bm L$ in $\parent_\Gcal(\bm S)$, following $P_{S_j|\parent_\Gcal(S_j)}^{(0)}$. For example, filtering cells based on mitochondrial transcript fraction or viability may involve gene expression variables and unmeasured determinants of cell viability. Therefore, the observed control sample $(\bm C,\bm X)$ follows $P_{\bm C,\bm X\mid\bm S=\bm 1}^{(0)}$. Although $\bm S=1$ for every retained sample, the selection mechanisms and their parents $\parent_\Gcal(\bm S)$ are unknown.

In addition to the control regime, we consider $K\le d_X$ single-target soft-intervention regimes, indexed so that regime $k$ has target $X_k$. We allow intervention assignment to depend on context variables $\bm C_1\subseteq\bm C$ and the intervention to modify the mechanisms of other context variables $\bm C_2\subseteq\bm C\setminus\bm C_1$. Here, ``single-target'' refers only to the system variables, and effects on $\bm C_2$ are treated as contextual spillover effects rather than additional intervention targets. For example, sgRNA detection in Perturb-seq can depend on batch and library size \citep{barry2024robust}, whereas perturbations can alter cell morphology \citep{chandrasekaran2024three}. We assume that the assignment and intervention mechanisms, together with $\Gcal$, admit an acyclic generative ordering.

Following the intervention-indicator formulation in \cite{yang2018characterizing,kocaoglu2019characterization,luo2025characterization}, let $\bm I=(I_1,\ldots,I_K)\in\{0,1\}^K$ with $\sum_{k\in[K]}I_k=1$, where $\bm I=e_k$ denotes assignment to interventional regime $k$, and $e_k$ is the $k$th standard basis vector. Its conditional distribution $P_{\bm I\mid\bm C_1}$ accounts for context-dependent assignment. We assume the joint distribution $P_{\bm Z,\bm I}$ of $(\bm Z,\bm I)$ can be factorized as
\[p_{\bm Z,\bm I}=p_{\bm I\mid\bm C_1}\prod_{C\in\bm C_2}p_{C\mid\parent_{\Gcal}(C),\bm I}\prod_{k\in[K]}p_{X_k\mid\parent_{\Gcal}(X_k),I_k}\prod_{Z\in\bm Z\setminus(\bm C_2,\bm X_{[K]})}p_{Z\mid\parent_{\Gcal}(Z)}^{(0)}.\]
The local mechanism of a non-target variable agrees with its control counterpart, 
\[p_{X_j\mid\parent_{\Gcal}(X_j),I_j=0}=p_{X_j\mid\parent_{\Gcal}(X_j)}^{(0)}.\]
Thus, in regime $k$, only the mechanisms of $X_k$ and $\bm C_2$ may differ from their control counterparts. In particular, the selection mechanisms remain unchanged. Consequently, the $k$th interventional distribution $P_{\bm Z}^{(k)}$ equals
\begin{equation}\label{eq:interventional_distribution}
    p_{\bm Z}^{(k)}=p_{\bm Z\mid\bm I=e_k}=\frac{\Prob(\bm I=e_k\mid\bm C_1)}{\Prob(\bm I=e_k)}\prod_{C\in\bm C_2}p_{C\mid\parent_{\Gcal}(C)}^{(k)}p_{X_k\mid\parent_{\Gcal}(X_k)}^{(k)}\prod_{Z\in\bm Z\setminus(\bm C_2, X_k)}p_{Z\mid\parent_{\Gcal}(Z)}^{(0)},
\end{equation}
where $p_{C\mid\parent_{\Gcal}(C)}^{(k)}=p_{C\mid\parent_{\Gcal}(C),\bm I=e_k}$ for $C\in\bm C_2$ and $p_{X_k\mid\parent_{\Gcal}(X_k)}^{(k)}=p_{X_k\mid\parent_{\Gcal}(X_k),I_k=1}$. The leading probability ratio accounts for context-dependent assignment and equals one when assignment is independent of $\bm C_1$. 

For each $k\in\{0,\ldots,K\}$, we observe a dataset
\[\Dcal_k=\{(\bm C_i^{(k)},\bm X_i^{(k)}):i\in[n_k]\}\overset{\rm i.i.d.}{\sim} P_{\bm C,\bm X\mid\bm S=\bm 1}^{(k)}.\]
Thus, $\Dcal_0$ is the control dataset, and $\Dcal_1,\ldots,\Dcal_K$ are the interventional datasets. Our objective is to recover causal relations among $\bm X$, treating the causal structure involving $\bm C$ as nuisance structure.

\subsection{Maximal Ancestral Graph}\label{sec:MAG}

In this section, we review the maximal ancestral graph \citep{richardson2002ancestral}. Throughout the section, for simplicity, we assume there is no context variable $\bm C$ and the observable variables are $\bm X$. All the definitions in this section are from \cite{spirtes1997polynomial} and \cite{zhang2008completeness}.

Although the existence of latent and selection variables $(\bm L, \bm S)$ may be identified from the data, $(\bm L, \bm S)$ themselves are not identifiable due to their unobservability. The MAG is a powerful summarization of the CI relations among the observed variables $\bm X$. To formalize the definition of MAG, we first introduce the inducing path. See Figure \ref{fig:example_MAG} for an example.

\begin{figure}
    \centering
    \begin{subfigure}{0.45\textwidth}
        \centering
        \includegraphics[width=0.4\linewidth]{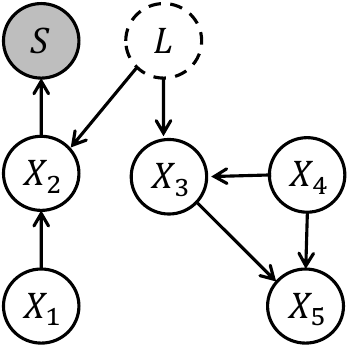}
        \caption{\footnotesize Full DAG $\Gcal$.}
        \label{fig:MAG_inducing}
    \end{subfigure}
    \begin{subfigure}{0.45\textwidth}
        \centering
        \includegraphics[width=0.4\linewidth]{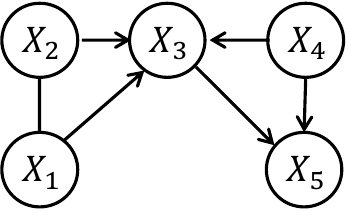}
        \caption{\footnotesize Corresponding $\MAG(\Gcal)$ of $\Gcal$.}
        \label{fig:MAG_unshielded}
    \end{subfigure}
    \caption{\footnotesize Example of a DAG $\Gcal$ (a) and its $\MAG(\Gcal)$ (b). $X_2\leftarrow L\rightarrow X_3$, $X_1\to X_2\leftarrow L\to X_3$, and all edges among $\bm X$ are inducing paths in $\Gcal$. $X_1\to X_3\leftarrow X_4$ is an unshielded collider in $\MAG(\Gcal)$. $X_2\to X_3\leftarrow X_4\to X_5$ is a discriminating path in $\MAG(\Gcal)$.}
    \label{fig:example_MAG}
\end{figure}

\begin{definition}[Inducing Path]\label{def:inducing_path}
    For observable variables $X,Y\in\bm X$, a path $\pi$ in $\Gcal$ between $X,Y$ is called an inducing path relative to $(\bm L,\bm S)$ if every non-endpoint node on $\pi$ is either in $\bm L$ or a collider, and every collider on $\pi$ is an ancestor of either $X,Y$, or $\bm S$.
\end{definition}
Intuitively, there is an inducing path between $X,Y$ if and only if $X,Y$ can not be $d$-separated by observable variables given $\bm S$. The MAG then summarizes all the $d$-separation information among $\bm X$ given $\bm S$ by collecting all the inducing paths.
\begin{definition}[MAG]\label{def:MAG}
    $\MAG(\Gcal)$ is a mixed graph on $\bm X$ that satisfies
    \begin{enumerate}[1)]
        \item for any $X,Y\in\bm X$, there is an edge between them in $\MAG(\Gcal)$ if and only if there is an inducing path relative to $(\bm L,\bm S)$ between $X,Y$ in $\Gcal$,
        \item for any edge in $\MAG(\Gcal)$ between $X,Y$,
        \begin{enumerate}[a)]
            \item $X\rightarrow Y$, if $X\in\ancestor_\Gcal(Y,\bm S)$ and $Y\not\in\ancestor_\Gcal(X,\bm S)$,
            \item $X\leftarrow Y$, if $X\not\in\ancestor_\Gcal(Y,\bm S)$ and $Y\in\ancestor_\Gcal(X,\bm S)$,
            \item $X\leftrightarrow Y$, if $X\not\in\ancestor_\Gcal(Y,\bm S)$ and $Y\not\in\ancestor_\Gcal(X,\bm S)$,
            \item $X - Y$, if $X\in\ancestor_\Gcal(Y,\bm S)$ and $Y\in\ancestor_\Gcal(X,\bm S)$.
        \end{enumerate}
    \end{enumerate}
    
\end{definition}
It is straightforward to verify that when there is no latent and selection variable, $\MAG(\Gcal)$ reduces to $\Gcal$. Since $(\bm L,\bm S)$ is not identifiable due to their unobservability, and MAGs fully characterize the CI relations among $\bm X$ given $\bm S$ in the original graph, the identifiability of $\Gcal$ is usually characterized through MAGs \citep{spirtes1997polynomial,kocaoglu2019characterization,luo2025characterization}. The following definitions are crucial for this purpose. See Figure \ref{fig:example_MAG} for an example.

\begin{definition}[Unshielded Collider]\label{def:unshielded_collider}
    In a MAG, a path with three nodes $\langle X, Y, Z\rangle$ is called an unshielded collider if $X,Z$ are not adjacent and both the edge between $X,Y$ and the edge between $Y,Z$ are into $Y$.
\end{definition}

\begin{definition}[Discriminating Path]\label{def:disc_path}
    In a MAG, a path $\pi=\langle X,\ldots,W,V,Y\rangle$ between $X,Y$ is a discriminating path for $V$ if
    \begin{enumerate}[1)]
        \item $\pi$ includes at least three edges,
        \item $V$ is a non-endpoint node on $\pi$, and is adjacent to $Y$ on $\pi$,
        \item $X,Y$ are not adjacent, every node between $X,V$ is a collider on $\pi$ and is a parent of $Y$.
    \end{enumerate}
\end{definition}

A discriminating path is a specially structured path in which intermediate nodes are colliders that are also parents of an endpoint. This structure makes the collider status of the shielded node $V$ identifiable \citep{spirtes1997polynomial}.

\section{System-Induced Subgraph and its Identifiability}\label{sec:SIS}

\subsection{System-Induced Subgraph}

Intervention indicators $\bm I$ provide a graphical characterization of IMEC through augmented graphs \citep{yang2018characterizing,kocaoglu2019characterization,luo2025characterization}. For each $k\in[K]$, define the \textbf{augmented graph} $\Aug(\Gcal,I_k)$ on $(\bm Z,I_k)$ by adjoining $I_k$ to $\Gcal$, together with the edges $I_k\rightarrow X_k$, $I_k\rightarrow C$ for $C\in\bm C_2$, and $C\rightarrow I_k$ for $C\in\bm C_1$. Under the acyclic mechanism factorization in \eqref{eq:interventional_distribution}, the interventional distribution $P_{\bm Z}^{(k)}$ is Markov to the augmented graph $\Aug(\Gcal, I_k)$ conditional on $I_k$. See Figure \ref{fig:data_DAG_aug} for an example. 

\begin{figure}
    \centering
    \begin{subfigure}{0.18\textwidth}
        \centering
        \includegraphics[width=0.9\linewidth]{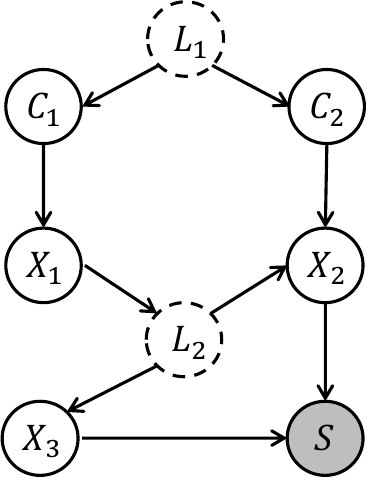}
        \caption{\footnotesize Full DAG $\Gcal$.}
        \label{fig:data_DAG}
    \end{subfigure}
    \hspace{12pt}
    \begin{subfigure}{0.24\textwidth}
        \centering
        \includegraphics[width=0.9\linewidth]{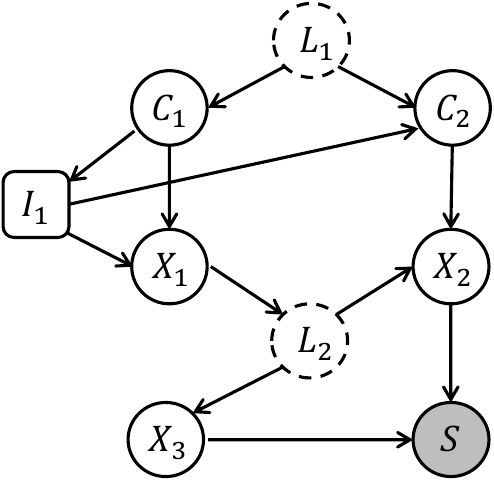}
        \caption{\footnotesize Augmented DAG $\Aug(\Gcal, I_1)$.}
        \label{fig:data_DAG_aug}
    \end{subfigure}
    \hspace{12pt}
    \begin{subfigure}{0.18\textwidth}
        \centering
        \includegraphics[width=0.9\linewidth]{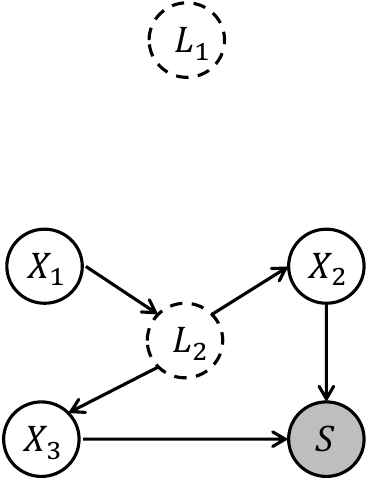}
        \caption{\footnotesize SIS $\Gcal_{\mid\bm C}$ of $\Gcal$.}
        \label{fig:data_SIS}
    \end{subfigure}
    \hspace{12pt}
    \begin{subfigure}{0.24\textwidth}
        \centering
        \includegraphics[width=0.9\linewidth]{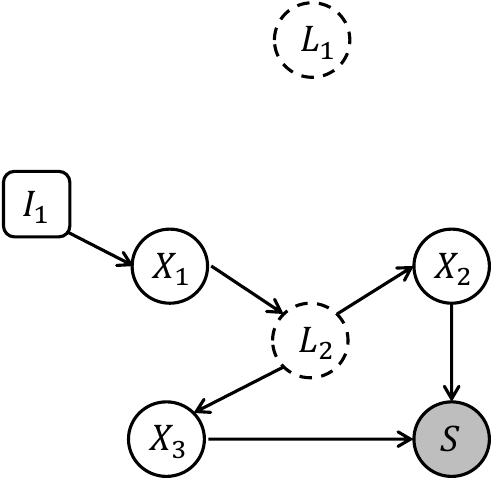}
        \caption{\footnotesize SIS $\Aug(\Gcal, I_1)_{\mid\bm C}$ of augmented DAG.}
        \label{fig:data_SIS_aug}
    \end{subfigure}
    \caption{\footnotesize Example of DAG $\Gcal$ (a), the augmented graph $\Aug(\Gcal, I_1)$ (b), and their SISs $\Gcal_{\mid\bm C}$, $\Aug(\Gcal, I_1)_{\mid\bm C}$ ((c) and (d)).}
    \label{fig:example_data}
\end{figure}

Rather than including all indicators $\bm I$ in a single augmented graph as in \cite{kocaoglu2019characterization,luo2025characterization,yang2018characterizing}, we define $K$ augmented graphs $\Aug(\Gcal, I_k)$ for each indicator $I_k$ to facilitate the graphical identification under selection and latent variables (Section \ref{sec:IMEC_SIS}). These augmented graphs offer unified representations of both the CI relations and the distributional invariance that arises under interventions \citep{yang2018characterizing,kocaoglu2019characterization,luo2025characterization}.

Recall that although the graph $\Gcal$ contains context variables $\bm C$, we are only interested in the causal relations among $\bm X$. To this end, we make the following assumptions throughout the paper. Assumption \ref{asm:context} is an extension of Assumptions 1 and 2 in \cite{mooij2020joint} by including selection variables. See Figure \ref{fig:context} for examples of prohibited structures under this assumption.

\begin{assumption}[Context Variables]\label{asm:context}
    No system variable in $\bm X$ causes any context variable in $\bm C$, i.e., $\ancestor_\Gcal(\bm C)\cap\bm X=\emptyset$. And context variables $\bm C$ are not latently confounded with system variables $\bm X$ or selection variables $\bm S$.
\end{assumption}

\begin{figure}
    \centering
    \begin{subfigure}{0.3\textwidth}
        \centering
        \includegraphics[width=0.35\linewidth]{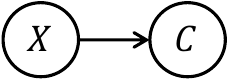}
    \end{subfigure}
    \hspace{10pt}
    \begin{subfigure}{0.3\textwidth}
        \centering
        \includegraphics[width=0.35\linewidth]{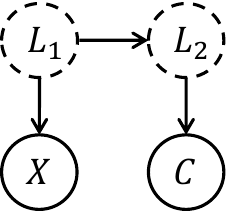}
    \end{subfigure}
    \hspace{10pt}
    \begin{subfigure}{0.3\textwidth}
        \centering
        \includegraphics[width=0.35\linewidth]{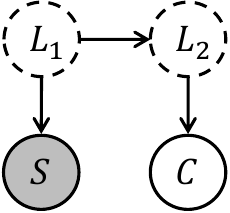}
    \end{subfigure}
    \caption{\footnotesize Forbidden structures under Assumption \ref{asm:context}.}
    \label{fig:context}
\end{figure}

Assumption \ref{asm:context} treats the context variables as exogenous to the system. This assumption allows us to isolate the causal relations between system variables while removing nuisance relations involving context variables. Here we say $Z,Z'\in\bm Z\setminus\bm L$ are \textbf{latently confounded} if there exists $L\in\bm L$ and a walk $Z\leftarrow \ldots\leftarrow L\rightarrow \ldots\rightarrow Z'$ such that 1) all middle nodes are in $\bm L$, 2) the subwalk $Z\leftarrow\ldots\leftarrow L$ has all arrowheads pointing to $Z$, and 3) the subwalk $L\rightarrow\ldots\rightarrow Z'$ has all arrowheads pointing to $Z'$.

To suppress the causal relations involving $\bm C$, we introduce the system-induced subgraph (SIS), an instance of induced subgraph \citep{richardson2002ancestral} that excludes $\bm C$. See Figures \ref{fig:data_SIS} and \ref{fig:data_DAG_aug} for examples.

\begin{definition}[System-Induced Subgraphs]
    \begin{enumerate}[1)]
        \item $\Gcal_{\mid\bm C}$ is a subgraph of $\Gcal$ with nodes $\Vcal(\Gcal_{\mid\bm C})=(\bm X, \bm L, \bm S)$. An edge in $\Ecal(\Gcal)$ is contained in $\Ecal(\Gcal_{\mid\bm C})$ if and only if neither end node is in $\bm C$.
        \item $\Aug(\Gcal,I_k)_{\mid\bm C}$ is a subgraph of $\Aug(\Gcal, I_k)$ with nodes $(\bm X, \bm L, \bm S, I_k)$. An edge in $\Ecal(\Aug(\Gcal, I_k))$ is contained in $\Ecal(\Aug(\Gcal, I_k)_{\mid\bm C})$ if and only if neither end node is in $\bm C$.
    \end{enumerate}
\end{definition}

Our use of context variables is closely related to the joint causal inference (JCI) framework of \cite{mooij2020joint}. However, the SIS serves a different purpose. Rather than learning a joint graph over both context and system variables, we treat context variables as nuisance variables and focus on recovering the causal relations among the remaining variables. Furthermore, although JCI accommodates latent variables, it doesn't explicitly account for selection bias, which plays a central role in our setting. 

Denote the augmented graph $\Aug(\Gcal_{\mid\bm C},I_k)$ of $\Gcal_{\mid\bm C}$ as a graph on $(\bm X,\bm L,\bm S,I_k)$ with an additional edge $I_k\rightarrow X_k$, it is clear that $\Aug(\Gcal_{\mid\bm C}, I_k)=\Aug(\Gcal,I_k)_{\mid\bm C}$. The goal of this paper is to recover the SIS from data.

Although the observed data need not be Markov with respect to SISs, the following proposition establishes a connection between the $d$-separation relations in the SISs $\Gcal_{\mid\bm C},\Aug(\Gcal_{\mid\bm C},I_k)$ and those of the original graphs $\Gcal,\Aug(\Gcal,I_k)$. 

\begin{proposition}[$d$-separations in SIS]\label{prop:context_conditional}
    \begin{enumerate}[1)]
        \item Under Assumption \ref{asm:context}, for any disjoint subsets $\bm Z_1,\bm Z_2,\bm Z_3$ with $\bm Z_1\subseteq\bm X$, $\bm Z_2\subseteq(\bm X,\bm L,\bm S)$, and $\bm Z_3\subseteq(\bm X,\bm S)$, we have $(\bm Z_1\indep\bm Z_2\mid\bm Z_3)_{\Gcal_{\mid\bm C}}$ if and only if $(\bm Z_1\indep\bm Z_2\mid\bm C, \bm Z_3)_{\Gcal}$.
        \item Under Assumption \ref{asm:context}, for any disjoint subsets $\bm Z_1,\bm Z_2,\bm Z_3$ with $\bm Z_1\subseteq\bm X$, $\bm Z_2\subseteq(\bm X,\bm L,\bm S,I_k)$, $\bm Z_3\subseteq(\bm X,\bm S, I_k)$, and $I_k\in(\bm Z_2,\bm Z_3)$, we have $(\bm Z_1\indep\bm Z_2\mid\bm Z_3)_{\Aug(\Gcal_{\mid\bm C},I_k)}$ if and only if $(\bm Z_1\indep\bm Z_2\mid\bm C, \bm Z_3)_{\Aug(\Gcal, I_k)}$.
    \end{enumerate}
\end{proposition}

Proposition \ref{prop:context_conditional} suggests that, to recover the causal relations among $\bm X$ encoded by the SISs $\Gcal_{\mid\bm C}$ and $\Aug(\Gcal_{\mid\bm C}, I_k)$, it suffices to recover the corresponding causal relations in the original graphs $\Gcal$ and $\Aug(\Gcal,I_k)$ conditional on $\bm C$. The following lemma connects $d$-separation in $\Aug(\Gcal_{\mid\bm C},I_k)$ to conditional invariance between control and $k$th interventional distributions.

\begin{lemma}\label{lem:Markov_Aug}
    Under Assumption \ref{asm:context}, for any $k\in[K]$ and disjoint sets $\bm X_1,\bm X_2\subseteq\bm X$, $(I_k\indep \bm X_1\mid\bm X_2,\bm S)_{\Aug(\Gcal_{\mid\bm C},I_k)}$ implies $P_{\bm X_1\mid\bm X_2,\bm C,\bm S}^{(k)}=P_{\bm X_1\mid\bm X_2,\bm C,\bm S}^{(0)}$, where equality means that the two conditional distributions admit a common version.
\end{lemma}

\subsection{Identifiablity of SIS under Intervention}\label{sec:IMEC_SIS}

As discussed in Section \ref{sec:MAG}, the identifiability of $\Gcal$ is typically characterized through MAGs, which has been studied extensively \citep{verma1991equivalence,spirtes1997polynomial,hauser2012characterization,yang2018characterizing,kocaoglu2019characterization,luo2025characterization}. However, when attention is restricted to causal relations among system variables, the IMEC of the subgraph $\Gcal_{\mid \bm C}$ becomes the central object of interest, and it can exhibit stronger identifiability than the original graph $\Gcal$. 

\begin{definition}[IMEC of $\Gcal_{\mid\bm C}$]\label{def:IMEC}
    In our context, for any SIS $\Gcal_{\mid\bm C}'$ on $(\bm X,\bm L',\bm S')$, $\Gcal_{\mid\bm C}'$ is interventional Markov equivalent to $\Gcal_{\mid\bm C}$ if and only if
    \begin{enumerate}[1)]
        \item for any disjoint sets $\bm X_1,\bm X_2,\bm X_3\subseteq\bm X$, $(\bm X_1\indep\bm X_2\mid\bm X_3,\bm S')_{\Gcal_{\mid\bm C}'}$ if and only if $(\bm X_1\indep\bm X_2\mid\bm X_3,\bm S)_{\Gcal_{\mid\bm C}}$,
        \item for any disjoint sets $\bm X_1,\bm X_2\subseteq\bm X$ and $k\in[K]$, $(I_k\indep\bm X_1\mid\bm X_2,\bm S')_{\Aug(\Gcal_{\mid\bm C}',I_k)}$ if and only if $(I_k\indep\bm X_1\mid\bm X_2,\bm S)_{\Aug(\Gcal_{\mid\bm C}, I_k)}$.
    \end{enumerate}
    Then the IMEC $\IMEC(\Gcal_{\mid\bm C})$ of $\Gcal_{\mid\bm C}$ consists of all SISs that are interventional Markov equivalent to $\Gcal_{\mid\bm C}$,
    \[\IMEC(\Gcal_{\mid\bm C})=\big\{\Gcal_{\mid\bm C}':\Gcal_{\mid\bm C}'\text{ interventional Markov equivalent to }\Gcal_{\mid\bm C}\big\}.\]
\end{definition}
Part 1) in Definition \ref{def:IMEC} requires $\Gcal_{\mid \bm C}',\Gcal_{\mid\bm C}$ to have the same CI relations among the observed variables, and part 2) asks for the same distributional invariance after interventions. Therefore, $\IMEC(\Gcal_{\mid\bm C})$ contains all SISs that are indistinguishable from $\Gcal_{\mid\bm C}$ based on the available distributions.

Similar to the characterizations in \citep{spirtes1997polynomial,yang2018characterizing,kocaoglu2019characterization,luo2025characterization}, we develop a graphical criterion for $\IMEC(\Gcal_{\mid\bm C})$ based on the MAGs of augmented graphs. 

\begin{proposition}[Graphical Characterization of IMEC]\label{prop:IMEC}
    The following statements are equivalent:
    \begin{enumerate}[1)]
        \item $\Gcal_{\mid\bm C}'\in\IMEC(\Gcal_{\mid\bm C})$,
        \item for all $k\in[K]$, $\MAG(\Aug(\Gcal_{\mid\bm C}', I_k))$ and $\MAG(\Aug(\Gcal_{\mid\bm C}, I_k))$ are $m$-separation equivalent,
        \item for all $k\in[K]$, $\MAG(\Aug(\Gcal_{\mid\bm C}',I_k))$ and $\MAG(\Aug(\Gcal_{\mid\bm C},I_k))$ satisfy
    \begin{enumerate}[a)]
        \item they have the same skeleton,
        \item they have the same unshielded colliders,
        \item if a path $\pi$ is a discriminating path for a node $X$ in both graphs, then $X$ is a collider on the path in one graph if and only if it is a collider on the path in the other.
    \end{enumerate}
    \end{enumerate}
\end{proposition}

Unlike existing works \citep{yang2018characterizing,kocaoglu2019characterization,luo2025characterization}, our criterion is based on a collection of $K$ MAGs rather than the single graph $\MAG(\Aug(\Gcal_{\mid\bm C},\bm I))$. The key distinction is that $\MAG(\Aug(\Gcal_{\mid\bm C},\bm I))$ also encodes the $d$-separations among indicators $\bm I$, which are not implied by Definition \ref{def:IMEC} under selection. Therefore, $\MAG(\Aug(\Gcal_{\mid\bm C},\bm I))$ fails to represent Definition \ref{def:IMEC} in our setting. See Example A.1 in the supplement \citep{hou2026supplement} for a counterexample.

By Proposition \ref{prop:IMEC}, $\IMEC(\Gcal_{\mid\bm C})$ is characterized by the $K$ MAGs. We will show that the $K$ MAGs share the same edges among $\bm X$ and the only differences are the nodes $I_k$ and their edges to $\bm X$. Therefore, all the $K$ MAGs can be represented by a single \textbf{combined-MAG} $\CM(\Gcal_{\mid\bm C},\bm I)$ on $(\bm X,\bm I)$, where $\CM(\Gcal_{\mid\bm C},\bm I)$ consists of all the edges from the $K$ MAGs, i.e., 
\begin{equation}\label{eq:def_CM}
    \Ecal\big(\CM(\Gcal_{\mid\bm C},\bm I)\big)=\bigcup_{k\in[K]}\Ecal\big(\MAG(\Aug(\Gcal_{\mid\bm C},I_k))\big).
\end{equation}
See Figure \ref{fig:CM} for an example. The combined-MAG provides a lossless representation of all the $K$ MAGs, and thus encodes all CI relations and distributional invariances in the data.

\begin{figure}
    \centering
    \begin{subfigure}{0.32\textwidth}
        \centering
        \includegraphics[width=0.45\linewidth]{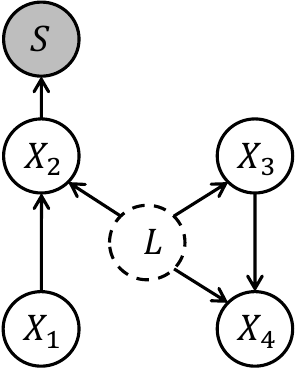}
        \caption{\footnotesize SIS $\Gcal_{\mid\bm C}$ of $\Gcal$.}
    \end{subfigure}
    \begin{subfigure}{0.32\textwidth}
        \centering
        \includegraphics[width=0.65\linewidth]{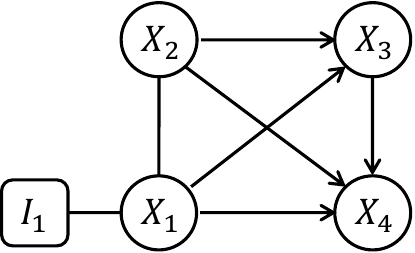}
        \caption{\footnotesize $\MAG(\Aug(\Gcal_{\mid\bm C},I_1))$.}
    \end{subfigure}
    \begin{subfigure}{0.32\textwidth}
        \centering
        \includegraphics[width=0.6\linewidth]{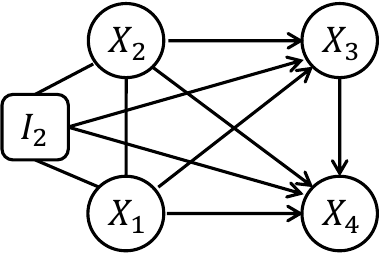}
        \caption{\footnotesize $\MAG(\Aug(\Gcal_{\mid\bm C},I_2))$.}
    \end{subfigure}

    \begin{subfigure}{0.32\textwidth}
        \centering
        \includegraphics[width=0.65\linewidth]{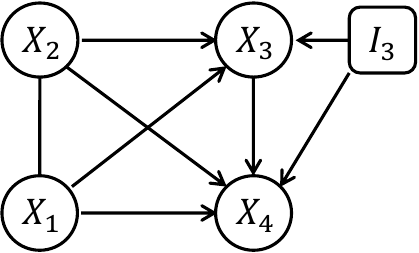}
        \caption{\footnotesize $\MAG(\Aug(\Gcal_{\mid\bm C},I_3))$.}
    \end{subfigure}
    \begin{subfigure}{0.32\textwidth}
        \centering
        \includegraphics[width=0.65\linewidth]{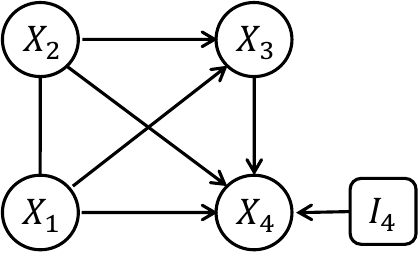}
        \caption{\footnotesize $\MAG(\Aug(\Gcal_{\mid\bm C},I_4))$.}
    \end{subfigure}
    \begin{subfigure}{0.32\textwidth}
        \centering
        \includegraphics[width=0.85\linewidth]{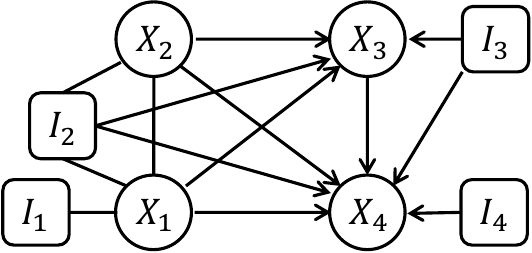}
        \caption{\footnotesize Combined-MAG $\CM(\Gcal_{\mid\bm C},\bm I)$.}
    \end{subfigure}
    \caption{\footnotesize Example of MAGs of augmented SIS $\MAG(\Aug(\Gcal_{\mid\bm C},I_k))$ ((b)-(e)) and combined-MAG $\CM(\Gcal_{\mid\bm C},\bm I)$  (f) for $K=4$.}
    \label{fig:CM}
\end{figure}

\begin{lemma}[Combined-MAG]\label{lem:CM}
    The $K$ MAGs $\big\{\MAG(\Aug(\Gcal_{\mid\bm C}, I_k)): k\in[K]\big\}$ admits a lossless representation by $\CM(\Gcal_{\mid\bm C},\bm I)$, i.e., for any $k\in[K]$,
    \[\Ecal\big(\MAG(\Aug(\Gcal_{\mid\bm C}, I_k))\big)=\big\{\text{edges in $\CM(\Gcal_{\mid\bm C},\bm I)$ with both end nodes in $(\bm X, I_k)$}\big\}.\]
\end{lemma}

\begin{remark}[Nonidentifiability of $\langle I_j, X_j\rangle$]
    Although $\CM(\Gcal_{\mid\bm C},\bm I)$ encodes all information of the IMEC, the edge type between $I_j,X_j$ is not identifiable in general. To see this, consider $\Gcal$ on $(X,S)$ with edge $X\rightarrow S$ and the trivial graph $\Gcal'$ on $X$. Suppose $X$ is intervened on, then $\Gcal,\Gcal'$ are interventional Markov equivalent. But $\CM(\Gcal,I)$ has edge $I-X$ whereas $\CM(\Gcal',I)$ has edge $I\rightarrow X$. The difference arises because for isolated selection variables with a single observed ancestor (Definition \ref{def:isolated_selection}), their ancestors are not identifiable from the data. Under additional assumptions, such as the absence of isolated selection, the edge type between $ I_j$ and $ X_j$ can be identified. Such assumptions, however, concern unobserved selection mechanism and are generally unverifiable. To avoid such assumptions, in the rest of this paper, we focus on the \textbf{trimmed combined-MAG} $\TCM(\Gcal_{\mid\bm C},\bm I)$ by removing the intervention-target edges, i.e.,
    \[\Ecal\big(\TCM(\Gcal_{\mid\bm C},\bm I)\big)=\Ecal\big(\CM(\Gcal_{\mid\bm C},\bm I)\big)\setminus\big\{\langle I_j,X_j\rangle: j\in[K]\big\}.\]
    $\CM(\Gcal_{\mid\bm C},\bm I)$ and $\TCM(\Gcal_{\mid\bm C},\bm I)$ share the same edges among $\bm X$, as well as the same edges $\langle I_k,X_j\rangle$ for $j\ne k$. In the remainder of the paper, we therefore do not distinguish between these two graphs when referring to these common edges.
\end{remark}

We take $\TCM(\Gcal_{\mid\bm C},\bm I)$ as our inferential target and consider its equivalent class, 
\[\Mcal(\TCM(\Gcal_{\mid\bm C},\bm I))=\big\{\TCM(\Gcal_{\mid\bm C}',\bm I):\Gcal_{\mid\bm C}'\in\IMEC(\Gcal_{\mid\bm C})\big\}.\]
The following theorem shows that $\Mcal(\TCM(\Gcal_{\mid\bm C},\bm I))$ reduces to a singleton if there are single-target interventions on all system variables, i.e., $K=d_X$. Moreover, $K=d_X$ is also necessary in the worst case. 

\begin{theorem}[Sufficient and Necessary $K$ for Identification]\label{thm:unique}
    If there are single-target interventions on all system variables, i.e., $K=d_X$, then $\Mcal(\TCM(\Gcal_{\mid\bm C},\bm I))$ reduces to the singleton $\Mcal(\TCM(\Gcal_{\mid\bm C},\bm I))=\big\{\TCM(\Gcal_{\mid\bm C},\bm I)\big\}$. Moreover, there exists a DAG $\Gcal$ such that for any $K<d_X$ single-target interventions, $|\Mcal(\TCM(\Gcal,\bm I))|>1$.
\end{theorem}

When there is no selection or latent variable, \cite{eberhardt2006n,eberhardt2012number} showed that $K=d_X-1$ single-target interventions are sufficient for identifying the underlying DAG. In contrast, under selection and latent variables, Theorem \ref{thm:unique} shows that $K=d_X-1$ is no longer sufficient and $K=d_X$ is necessary for identifying the underlying graph.

\subsection{Interpretation of the Trimmed Combined-MAG when $K=d_X$}\label{sec:interpret}

While the trimmed combined MAG $\TCM(\Gcal_{\mid\bm C},\bm I)$ is identifiable under per-node interventions, it is not immediately clear how to extract the underlying causal relations from it. This section discusses the interpretation of $\TCM(\Gcal_{\mid\bm C},\bm I)$ when $K=d_X$. Here, we consider the setting with selection and latent variables and defer the cases without selection or latent variables to Section B of the supplement \citep{hou2026supplement}. The following lemma summarizes some causal relations. 

\begin{lemma}\label{lem:interpretation_latent_selection}
    For $j\ne k\in[d_X]$, we have the following rules,
    \begin{enumerate}[1)]
        \item if $X_j-X_k$ is in $\TCM(\Gcal_{\mid\bm C},\bm I)$, then $X_j,X_k\in\ancestor_{\Gcal_{\mid\bm C}}(\bm S)$,
        \item if $X_j\leftrightarrow X_k$ is in $\TCM(\Gcal_{\mid\bm C},\bm I)$, then $X_j\not\in\ancestor_{\Gcal_{\mid\bm C}}(X_k,\bm S),X_k\not\in\ancestor_{\Gcal_{\mid\bm C}}(X_j,\bm S)$ and all inducing paths in $\Gcal_{\mid\bm C}$ between $X_j,X_k$ have the form $X_j\leftarrow L\ldots L'\rightarrow X_k$ for some $L,L'\in\bm L$, 
        \item if $X_j\rightarrow X_k$ is in $\TCM(\Gcal_{\mid\bm C},\bm I)$, then $X_j\in\ancestor_{\Gcal_{\mid\bm C}}(X_k,\bm S)$, $X_k\not\in\ancestor_{\Gcal_{\mid\bm C}}(X_j,\bm S)$ and all inducing paths in $\Gcal_{\mid\bm C}$ between $X_j,X_k$ are into $X_k$,
        \item $I_j-X_k$ is in $\TCM(\Gcal_{\mid\bm C},\bm I)$ if and only if $X_j,X_k\in\ancestor_{\Gcal_{\mid\bm C}}(\bm S)$ and $\Gcal_{\mid\bm C}$ contains an inducing path with either $X_j\leftarrow X_k$ or $X_j\leftarrow L\ldots X_k$ for some $L\in\bm L$,
        \item $I_j\rightarrow X_k$ is in $\TCM(\Gcal_{\mid\bm C},\bm I)$ if and only if $X_j\in\ancestor_{\Gcal_{\mid\bm C}}(X_k,\bm S)$, $X_k\not\in\ancestor_{\Gcal_{\mid\bm C}}(X_j,\bm S)$ and $\Gcal_{\mid\bm C}$ contains an inducing path $X_j\leftarrow L\ldots L'\rightarrow X_k$ for some $L,L'\in\bm L$.
    \end{enumerate}
\end{lemma}

\begin{remark}\label{rem:interpretation_latent_selection}
    \begin{enumerate}[1)]
        \item While undirected edges $X_j-X_k$ or $I_j-X_k$ imply $X_j,X_k\in\ancestor_{\Gcal_{\mid\bm C}}(\bm S)$, arrowheads $\rightarrow X_k$ or $\leftrightarrow X_k$ guarantees $X_k\not\in\ancestor_{\Gcal_{\mid\bm C}}(\bm S)$. For $X_j\not\in\ancestor_{\Gcal_{\mid\bm C}}(\bm S)$, $X_j\rightarrow X_k$ or $I_j\rightarrow X_k$ in $\TCM(\Gcal_{\mid\bm C},\bm I)$ implies $X_j\in\ancestor_{\Gcal_{\mid\bm C}}(X_k)$.
        \item As shown in Lemma F.5 in the supplement \citep{hou2026supplement}, any inducing path of the form $X_j\leftarrow L\ldots L'\rightarrow X_k$ implies the existence of a sequence of latently confounded nodes. Therefore, if $X_j\leftrightarrow X_k$ or $I_j\rightarrow X_k$ is in $\TCM(\Gcal_{\mid\bm C},\bm I)$, there exist a sequence of nodes $Z_{j_0}\overset{\triangle}{=}X_j, Z_{j_1},\ldots,Z_{j_m},Z_{j_{m+1}}\overset{\triangle}{=}X_k\in\bm X\cup\bm S$ such that $Z_{j_{l-1}}$ is latently confounded with $Z_{j_l}$ for $l\in[m+1]$.
        \item If $I_j-X_k$ and $I_k-X_j$ are both in $\TCM(\Gcal_{\mid\bm C},\bm I)$, then at least one of $X_j,X_k$ is latently confounded with some nodes in $\bm X\cup\bm S$.
        \item Conversely, if $X_j$ and $X_k$ are latently confounded, Lemma \ref{lem:interpretation_latent_selection} implies that one of $X_j\leftrightarrow X_k$, $I_j\rightarrow X_k$, $I_k\rightarrow X_j$ or $\{I_j-X_k,I_k-X_j\}$ must appear.
    \end{enumerate}
\end{remark}
By Remark \ref{rem:interpretation_latent_selection}, structures in $\bm E_{\bm L}$ below indicate the existence of latent confounding in $\Gcal_{\mid\bm C}$.
\begin{equation}\label{eq:EL}
    \begin{aligned}
        \bm E_{\bm L}=&\big\{\{X_j,X_k\}:X_j\leftrightarrow X_k\in\Ecal(\TCM(\Gcal_{\mid\bm C},\bm I))\text{ or }I_j\rightarrow X_k\in\Ecal(\TCM(\Gcal_{\mid\bm C},\bm I))\\
        &\text{or both }I_j-X_k,I_k-X_j\in\Ecal(\TCM(\Gcal_{\mid\bm C},\bm I))\big\}.
    \end{aligned}
\end{equation}

\begin{example}\label{exa:latent_selection}
    Consider the DAG $\Gcal_{\mid\bm C}$ shown in Figure \ref{fig:example_latent_selection}. In $\TCM(\Gcal_{\mid\bm C},\bm I)$, edges always connect to $\ancestor_{\Gcal_{\mid\bm C}}(S)\cap\bm X$ through tail marks, such as $X_2-X_3$ even though $X_2$ is a parent of $X_3$ in $\Gcal_{\mid\bm C}$, and $X_3\rightarrow X_4$ even though $X_3$ is not an ancestor of $X_4$. By Remark \ref{rem:interpretation_latent_selection}, undirected edges imply $X_1,X_2,X_3\in\ancestor_{\Gcal_{\mid\bm C}}(\bm S)$ and arrowheads indicate $X_4,X_5\not\in\ancestor_{\Gcal_{\mid\bm C}}(\bm S)$. These results together ensure $\ancestor_{\Gcal_{\mid\bm C}}(\bm S)\cap\bm X=\{X_1,X_2,X_3\}$.
    
    The edge $I_3\rightarrow X_4$ implies the existence of a sequence of latently confounded nodes between $X_3,X_4$. In this example, $I_3\rightarrow X_4$ is due to $X_3\leftarrow L_3\rightarrow X_4$. The two edges $I_1-X_3,I_3-X_1$ indicate at least one of $X_1,X_3$ is latently confounded with some node in $\bm X\cup\bm S$. In this example, the underlying causal structure is $X_1\leftarrow L_1\rightarrow X_3$.

    \begin{figure}
        \centering
        \begin{subfigure}{0.45\textwidth}
            \centering
            \includegraphics[width=0.65\linewidth]{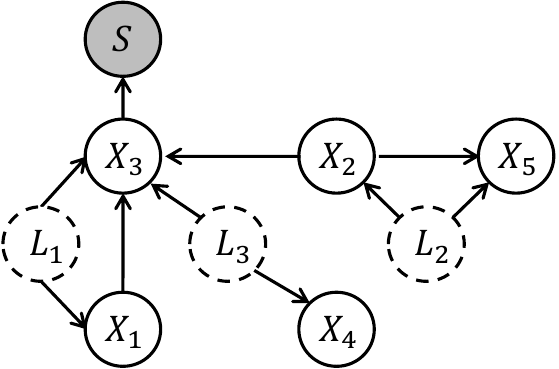}
            \caption{\footnotesize SIS $\Gcal_{\mid\bm C}$ of some $\Gcal$.}
        \end{subfigure}
        \begin{subfigure}{0.45\textwidth}
            \centering
            \includegraphics[width=0.65\linewidth]{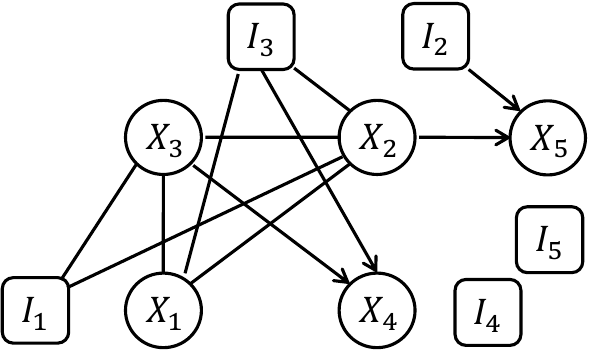}
            \caption{\footnotesize Trimmed combined-MAG $\TCM(\Gcal_{\mid\bm C},\bm I)$.}
        \end{subfigure}
        \caption{\footnotesize Causal structures in Example \ref{exa:latent_selection}.}
        \label{fig:example_latent_selection}
    \end{figure}
\end{example}

\section{Inference under Selection and Latent Variables}\label{sec:latent_selection}

In this section, we suppose every system variable in $\bm X$ has a single-target intervention, i.e., $K=d_X$. When $K<d_X$, we can treat the remaining system variables as latent and learn the graph among intervened variables. Then we propose an inference framework for $\TCM(\Gcal,\bm I)$ with rigorous type I error control and establish its constraint-query optimality.

\subsection{Algorithm and Inference}

Denote $\bm A_j\subseteq\bm X$ to be the set of system variables intervening which affect $X_j$,
\[\bm A_j=\big\{X_k\in\bm X: k\ne j, (I_k\not\indep X_j\mid\bm S)_{\Aug(\Gcal_{\mid\bm C}, I_k)}\big\}.\]
The following proposition shows that $\bm A_j\cup\{X_j\}$ is exactly the anterior set of $X_j$ in $\MAG(\Gcal_{\mid\bm C})$, and a single set determines the adjacency in $\TCM(\Gcal_{\mid\bm C}, \bm I)$. 
\begin{proposition}[Adjacency]\label{prop:adjacency_latent_selection}
    \begin{enumerate}[1)]
        \item For any $j\in[d_X]$, we have $\ancestor_{\Gcal_{\mid\bm C}}(X_j)\setminus(X_j,\bm L)\subseteq\bm A_j=\anterior_{\MAG(\Gcal_{\mid\bm C})}(X_j)\setminus\{X_j\}\subseteq\ancestor_{\Gcal_{\mid\bm C}}(X_j,\bm S)\setminus(X_j,\bm L,\bm S)$.
        \item For any $j\ne k\in[d_X]$ and any set $\bm A$ with $\bm A_j\cup\bm A_k\subseteq\bm A\subseteq\ancestor_{\Gcal_{\mid\bm C}}(X_j,X_k,\bm S)\setminus\bm L$, we have $X_j, X_k$ are adjacent in $\TCM(\Gcal_{\mid\bm C}, \bm I)$ if and only if $\big(X_j\not\indep X_k\mid (\bm A,\bm S)\setminus(X_j,X_k)\big)_{\Gcal_{\mid\bm C}}$.
        \item For any $j\ne k\in[d_X]$ and any set $\bm A_j'$ with $\bm A_j\subseteq\bm A_j'\subseteq\ancestor_{\Gcal_{\mid\bm C}}(X_j,\bm S)\setminus\bm L$, we have $I_k,X_j$ are adjacent in $\TCM(\Gcal_{\mid\bm C}, \bm I)$ if and only if $\big(I_k\not\indep X_j\mid (\bm A_j',\bm S)\setminus\{X_j\}\big)_{\Aug(\Gcal_{\mid\bm C},I_k)}$.
    \end{enumerate}
\end{proposition}

The anterior set $\bm A_j$ corresponds to the set of interventions that influence $X_j$, which can be learned by conditional two-sample (CTS) tests (which we call \textbf{anterior tests}) of
\[P^{(k)}_{X_j\mid\bm C,\bm S=\bm 1}=P_{X_j\mid\bm C,\bm S=\bm 1}^{(0)},\quad k\ne j.\]
Therefore Proposition \ref{prop:adjacency_latent_selection} is a practical refinement of Theorem 4.2 in \cite{richardson2002ancestral}, where the separation set there is $\ancestor_{\Gcal_{\mid\bm C}}(X_j,X_k,\bm S)\setminus(X_j,X_k,\bm L)$. In general, $(\bm A_j\cup\bm A_k\cup\bm S)\setminus(X_j,X_k,\bm L)$ is smaller than $\ancestor_{\Gcal_{\mid\bm C}}(X_j,X_k,\bm S)\setminus(X_j,X_k,\bm L)$ since the former only contains a part of $\ancestor_{\Gcal_{\mid\bm C}}(\bm S)\cap\bm X$ that connects with $X_j$ or $X_k$. 

\begin{proposition}[Orientation]\label{prop:direction_latent_selection}
    \begin{enumerate}[1)]
        \item For any $j\ne k\in[d_X]$, if $X_j,X_k$ are adjacent in $\TCM(\Gcal_{\mid\bm C},\bm I)$, for any $\bm A_j',\bm A_k'$ with $\bm A_j\subseteq\bm A_j'\subseteq\ancestor_{\Gcal_{\mid\bm C}}(X_j,\bm S)\setminus\bm L$ and $\bm A_k\subseteq\bm A_k'\subseteq\ancestor_{\Gcal_{\mid\bm C}}(X_k,\bm S)\setminus\bm L$, we have the edge in $\TCM(\Gcal_{\mid\bm C},\bm I)$ between $X_j,X_k$ is
    \begin{enumerate}[a)]
        \item $X_j\rightarrow X_k$, if $X_j\in\bm A_k'$ and $X_k\not\in\bm A_j'$,
        \item $X_j\leftarrow X_k$, if $X_j\not\in\bm A_k'$ and $X_k\in\bm A_j'$,
        \item $X_j\leftrightarrow X_k$, if $X_j\not\in\bm A_k'$ and $X_k\not\in\bm A_j'$,
        \item $X_j - X_k$, if $X_j\in\bm A_k'$ and $X_k\in\bm A_j'$.
    \end{enumerate}
    \item For any $j\ne k\in[d_X]$, if $I_k,X_j$ are adjacent in $\TCM(\Gcal_{\mid\bm C},\bm I)$, for any $\bm A_k'$ with $\bm A_k\subseteq\bm A_k'\subseteq\ancestor_{\Gcal_{\mid\bm C}}(X_k,\bm S)\setminus\bm L$, we have the edge in $\TCM(\Gcal_{\mid\bm C},\bm I)$ between $I_k,X_j$ is
    \begin{enumerate}[a)]
        \item $I_k\rightarrow X_j$, if $X_j\not\in\bm A_k'$,
        \item $I_k - X_j$, if $X_j\in\bm A_k'$.
    \end{enumerate}
    \end{enumerate}
\end{proposition}

Since $\bm A_j,\bm A_k$ are identifiable, Propositions \ref{prop:adjacency_latent_selection} and \ref{prop:direction_latent_selection} allow us to learn the adjacency and edge types from data. Although $\bm A_j$ is already learnable, it may require a strong faithfulness assumption. In practice, we choose a specific set between $\bm A_j$ and $\ancestor_{\Gcal_{\mid\bm C}}(X_j,\bm S)\setminus\bm L$ that requires weaker faithfulness assumption to learn. To introduce the set, we first define a subset of selections whose existence can not be identified. 

\begin{definition}[Isolated Selection]\label{def:isolated_selection}
    We say $\bm S_{\rm is}\subseteq\bm S$ be all isolated selection variables if for all $S\in\bm S_{\rm is}$, $S$ satisfies one of the following:
    \begin{enumerate}[1)]
        \item $S$ has no observed ancestor, $|\ancestor_{\Gcal_{\mid\bm C}}(S)\cap\bm X|=0$,
        \item $S$ has only one observed ancestor $X$ and $X$ is $d$-separated from other observed ancestors of selections given $\bm S$, i.e., $\ancestor_{\Gcal_{\mid\bm C}}(S)\cap\bm X=\{X\}$ and $(X\indep X'\mid\bm S)_{\Gcal_{\mid\bm C}}$ for all $X'\in\ancestor_{\Gcal_{\mid\bm C}}(\bm S)\cap\bm X\setminus\{X\}$.
    \end{enumerate}
\end{definition}

The following lemma shows that the observed ancestors of $\bm S\setminus\bm S_{\rm is}$ are identifiable. 

\begin{lemma}[Observed Ancestors of Selection]\label{lem:S}
    For any $X_j\in\bm X$, $X_j\in\bm A_{\bm S}\overset{\triangle}{=}\ancestor_{\Gcal_{\mid\bm C}}(\bm S\setminus\bm S_{\rm is})\cap\bm X$ if and only if there exists $X_k\in\bm X\setminus\{X_j\}$ such that $X_k\in\bm A_j$ and $X_j\in\bm A_k$.
\end{lemma}

Lemma \ref{lem:S} implies that for any $X_j\in\ancestor_{\Gcal_{\mid\bm C}}(\bm S\setminus\bm S_{\rm is})\cap\bm X$, there exists some $X_k$ such that $(X_j,X_k)$ exhibit cyclic behavior: intervention on either variable affects the other. This unique behavior allows for an efficient detection of ancestors of selection.

Motivated by Propositions \ref{prop:adjacency_latent_selection}, \ref{prop:direction_latent_selection}, and Lemma \ref{lem:S}, we choose $\bm A_j'=\bm A_j\cup\ancestor_{\Gcal_{\mid\bm C}}(\bm S\setminus\bm S_{\rm is})\cap\bm X$, which allows a weaker faithfulness assumption. Suppose that the sets $\bm A_{[d_X]}=\{\bm A_j:j\in[d_X]\}$ and $\bm A_{\bm S}$ are known. Proposition \ref{prop:adjacency_latent_selection} suggests to learn the adjacency between $X_j,X_k$ by the CI test, which we call an \textbf{adjacency test},
\[X_j\indep X_k\mid \bm S=\bm 1,\bm C\cup\bm A_j\cup\bm A_k\cup\bm A_{\bm S}\setminus\{X_j,X_k\}.\]
Similarly, the adjacency between $I_k,X_j$ can be learned by the CTS adjacency test
\[P_{X_j\mid\bm S=\bm 1, \bm C\cup\bm A_j\cup\bm A_{\bm S}\setminus\{X_j\}}^{(k)}=P_{X_j\mid\bm S=\bm 1,\bm C\cup\bm A_j\cup\bm A_{\bm S}\setminus\{X_j\}}^{(0)}.\]
Once an adjacency is detected, its edge type is determined from the anterior sets according to Proposition \ref{prop:direction_latent_selection}.

The sets $\bm A_{\bm S}$ and $\bm A_{[d_X]}$ can be learned in parallel, and, conditional on these sets, all adjacency tests are also parallelizable. This motivates the two-stage procedure in Algorithm \ref{alg:conf_MAG_latent_selection}. In Stage I, we estimate $\bm A_{\bm S}$ and $\bm A_{[d_X]}$ using the anterior tests, with family-wise type I error ($\text{FWER}_1$) controlled below $\frac{\alpha}{2}$. In Stage II, we replace the unknown anterior sets by $\widehat{\bm A}_{[d_X]}, \widehat{\bm A}_{\bm S}$ and perform adjacency tests with $\text{FWER}_1$ below $\frac{\alpha}{2}$. Together, the two stages yield a unified procedure for learning edges of $\TCM(\Gcal_{\mid\bm C}, \bm I)$ with $\text{FWER}_1$ control. Although the anterior and adjacency tests use the same data and the Stage II tests are selected using the Stage I results, Algorithm \ref{alg:conf_MAG_latent_selection} doesn't account for this dependence and performs Stage II by treating $\widehat{\bm A}_j,\widehat{\bm A}_{\bm S}$ as fixed. The resulting error control is formalized in Theorem \ref{thm:coverage_latent_selection}.

In Steps 5 and 6 of Algorithm \ref{alg:conf_MAG_latent_selection}, we apply the adjacency tests based on the pooled data rather than the control data. This is motivated by the following lemma, which says that the conditional distributions remain the same across the pooled samples. 

\begin{lemma}[Data Pooling]\label{lem:pool}
    For any $j,k\in[d_X]$, any $\bm A$ with $\bm A_j\cup\bm A_k\subseteq\bm A\subseteq\ancestor_{\Gcal_{\mid\bm C}}(X_j,X_k,\bm S)\setminus\bm L$, and any $l\in[d_X]$ such that $X_l\not\in\bm A_j\cup\bm A_k\cup\{X_j,X_k\}$, we have $\big(I_l\indep (X_j,X_k)\mid\bm A\cup\bm S\setminus\{X_j,X_k\}\big)_{\Aug(\Gcal_{\mid\bm C},I_l)}$.
\end{lemma}

\begin{algorithm}
\caption{\footnotesize Lower Confidence Set for $\Ecal(\TCM(\Gcal_{\mid\bm C},\bm I))$ with Selection and Latent Variables}
\label{alg:conf_MAG_latent_selection}
\footnotesize
\begin{algorithmic}[1]
\Statex{\bf Input:} Observational data $\Dcal_0$, interventional data $\Dcal_k$, $k\in[d_X]$, confidence level $1-\alpha$.
\Statex{\bf Output:} Lower Confidence set $\Ccal_\alpha$ for $\Ecal(\TCM(\Gcal_{\mid\bm C}, \bm I))$.

\Statex \textbf{Stage I. Anterior Sets Recovery}
\State{\bf Step 1:} For $k\in[d_X]$, $j\in[d_X]\setminus\{k\}$, conduct CTS test $\psi_j^{\anterior(k)}\in\{0,1\}$ based on $\Dcal_0$ and $\Dcal_k$,
\[H_{j,0}^{\anterior(k)}:\Prob\big(p^{(k)}_{X_j\mid\bm C,\bm S=\bm 1}\ne p_{X_j\mid\bm C,\bm S=\bm 1}^{(0)}\big)=0\quad \text{vs} \quad H_{j,1}^{\anterior(k)}:\Prob\big(p^{(k)}_{X_j\mid\bm C,\bm S=\bm 1}\ne p_{X_j\mid\bm C,\bm S=\bm 1}^{(0)}\big)>0,\]
with $\text{FWER}_1$ controlled below $\frac{\alpha}{2}$. 
\State{\bf Step 2:} For $j\in[d_X]$, initialize the detected anterior set for $X_j$ and $\bm S$ as
\[\widehat{\bm A}_j=\bigg\{X_k:\exists m\ge 0,j_0=j,j_{m+1}=k, (j_0,\ldots,j_{m+1})\subseteq[d_X]\text{ s.t. }\prod_{l=0}^m\psi_{j_l}^{\anterior(j_{l+1})}=1\bigg\}\setminus\{X_j\},\]
\[\widehat{\bm A}_{\bm S}=\big\{X_k:\exists j\ne k\text{ s.t. }\psi_j^{\anterior(k)}\psi_k^{\anterior(j)}=1\big\}.\]
\State{\bf Step 3:} Iteratively update $\{\widehat{\bm A}_j:j\in[d_X]\}$ and $\widehat{\bm A}_{\bm S}$ until convergence:
\begin{itemize}
    \item For any $X_k\in\widehat{\bm A}_{\bm S}$ and $X_j\in\widehat{\bm A}_k$, update 
    \[\widehat{\bm A}_{\bm S} \leftarrow \widehat{\bm A}_{\bm S}\cup \{X_j\},\quad \widehat{\bm A}_j\leftarrow \widehat{\bm A}_j\cup\{X_k\}.\]
    \item For all $j\in[d_X]$, update 
    \[\widehat{\bm A}_j \leftarrow\big\{X_k:\exists m\ge 0,j_0=j,j_{m+1}=k,(j_0,\ldots,j_{m+1})\subseteq[d_X]\text{ s.t. }\forall l\in[m+1],X_{j_l}\in\widehat{\bm A}_{j_{l-1}}\big\}\setminus\{X_j\}.\]
\end{itemize}

\Statex \textbf{Stage II. Adjacency and Orientation Recovery}
\State{\bf Step 4:} Initialize the detected edge set $\Ccal_\alpha=\emptyset$, and conduct the following parallel CI and CTS tests with $\text{FWER}_1$ controlled below $\frac{\alpha}{2}$, pretending $\widehat{\bm A}_{\bm S}$ and $\widehat{\bm A}_{[d_X]}=\{\widehat{\bm A}_{j}:j\in[d_X]\}$ are fixed.
\State{\bf Step 5:} For $j,k\in[d_X]$, $j>k$, conduct CI test $\psi_{j,k}^{\rm Adj}\big(\widehat{\bm A}_{\bm S},\widehat{\bm A}_{[d_X]}\big)\in\{0,1\}$ based on the pooled data $\Dcal_0\cup\big\{\Dcal_l:X_l\not\in\widehat{\bm A}_j\cup\widehat{\bm A}_k\cup\{X_j,X_k\}\big\}$: 
\begin{align*}
    &H_{j,k,0}^{\rm Adj}\big(\widehat{\bm A}_{\bm S},\widehat{\bm A}_{[d_X]}\big):X_j\indep X_k\mid \bm S=\bm 1,\bm C\cup\widehat{\bm A}_j\cup\widehat{\bm A}_k\cup\widehat{\bm A}_{\bm S}\setminus\{X_j,X_k\}\\
    \text{vs} \quad &H_{j,k,1}^{\rm Adj}\big(\widehat{\bm A}_{\bm S},\widehat{\bm A}_{[d_X]}\big):X_j\not\indep X_k\mid\bm S=\bm 1, \bm C\cup \widehat{\bm A}_j\cup\widehat{\bm A}_k\cup\widehat{\bm A}_{\bm S}\setminus\{X_j,X_k\}.
\end{align*}
\begin{itemize}
    \item If $\psi_{j,k}^{\rm Adj}\big(\widehat{\bm A}_{\bm S},\widehat{\bm A}_{[d_X]}\big)=1$, $X_j\in \widehat{\bm A}_k\cup\widehat{\bm A}_{\bm S}$, and $X_k\not\in\widehat{\bm A}_j\cup\widehat{\bm A}_{\bm S}$, append $X_j\rightarrow X_k$ to $\Ccal_\alpha$. 
    \item If $\psi_{j,k}^{\rm Adj}\big(\widehat{\bm A}_{\bm S},\widehat{\bm A}_{[d_X]}\big)=1$, $X_j\in \widehat{\bm A}_k\cup\widehat{\bm A}_{\bm S}$, and $X_k\in\widehat{\bm A}_j\cup\widehat{\bm A}_{\bm S}$, append $X_j-X_k$ to $\Ccal_\alpha$. 
    \item If $\psi_{j,k}^{\rm Adj}\big(\widehat{\bm A}_{\bm S},\widehat{\bm A}_{[d_X]}\big)=1$, $X_j\not\in \widehat{\bm A}_k\cup\widehat{\bm A}_{\bm S}$, and $X_k\in\widehat{\bm A}_j\cup\widehat{\bm A}_{\bm S}$, append $X_j\leftarrow X_k$ to $\Ccal_\alpha$. 
    \item If $\psi_{j,k}^{\rm Adj}\big(\widehat{\bm A}_{\bm S},\widehat{\bm A}_{[d_X]}\big)=1$, $X_j\not\in \widehat{\bm A}_k\cup\widehat{\bm A}_{\bm S}$, and $X_k\not\in\widehat{\bm A}_j\cup\widehat{\bm A}_{\bm S}$, append $X_j\leftrightarrow X_k$ to $\Ccal_\alpha$. 
\end{itemize}
\State{\bf Step 6: } For $j\in[d_X]$, $X_k\in\widehat{\bm A}_j$, conduct CTS test $\psi_j^{{\rm Adj}(k)}\big(\widehat{\bm A}_{\bm S},\widehat{\bm A}_{[d_X]}\big)\in\{0,1\}$ based on $\Dcal_k$ and the pooled data $\Dcal_0\cup\{\Dcal_l:X_l\not\in\widehat{\bm A}_j\cup\{X_j,X_k\}\}$
\begin{align*}
    &H_{j,0}^{{\rm Adj}(k)}\big(\widehat{\bm A}_{\bm S},\widehat{\bm A}_{[d_X]}\big): \Prob\big(p_{X_j\mid\bm S=\bm 1, \bm C\cup\widehat{\bm A}_j\cup\widehat{\bm A}_{\bm S}\setminus\{X_j\}}^{(k)}\ne p_{X_j\mid\bm S=\bm 1,\bm C\cup\widehat{\bm A}_j\cup\widehat{\bm A}_{\bm S}\setminus\{X_j\}}^{(0)}\big)=0\\
    \text{vs} \quad & H_{j,1}^{{\rm Adj}(k)}\big(\widehat{\bm A}_{\bm S},\widehat{\bm A}_{[d_X]}\big): \Prob\big(p_{X_j\mid\bm S=\bm 1, \bm C\cup\widehat{\bm A}_j\cup\widehat{\bm A}_{\bm S}\setminus\{X_j\}}^{(k)}\ne p_{X_j\mid\bm S=\bm 1,\bm C\cup\widehat{\bm A}_j\cup\widehat{\bm A}_{\bm S}\setminus\{X_j\}}^{(0)}\big)>0.
\end{align*}
\begin{itemize}
    \item If $\psi_j^{{\rm Adj}(k)}\big(\widehat{\bm A}_{\bm S},\widehat{\bm A}_{[d_X]}\big)=1$ and $X_j\in \widehat{\bm A}_k\cup\widehat{\bm A}_{\bm S}$, append $I_k- X_j$ to $\Ccal_\alpha$.
    \item If $\psi_j^{{\rm Adj}(k)}\big(\widehat{\bm A}_{\bm S},\widehat{\bm A}_{[d_X]}\big)=1$ and $X_j\not\in \widehat{\bm A}_k\cup\widehat{\bm A}_{\bm S}$, append $I_k\rightarrow X_j$ to $\Ccal_\alpha$.
\end{itemize}

\end{algorithmic}
\end{algorithm}

To illustrate the performance of Algorithm \ref{alg:conf_MAG_latent_selection}, we require the following faithfulness assumptions. While Assumption \ref{asm:faithful_descendant_latent_selection} enables accurate ancestor detection, Assumption \ref{asm:faithful_adjacency_latent_selection} allows consistent adjacency detection. By transitivity of ancestry \citep{park2026confounder}, Assumption \ref{asm:faithful_descendant_latent_selection} only requires the marginal distribution shift for variables directly affected by intervention targets. 

\begin{assumption}[Faithful Interventional Effect]\label{asm:faithful_descendant_latent_selection}
    \begin{enumerate}[1)]
        \item For any $j\in[d_X]$ and $X_k\in\ancestor_{\Gcal_{\mid\bm C}}(X_j)\cap \bm X\setminus\{X_j\}$, if there is no directed path $X_k\rightarrow\ldots\rightarrow X_j$ whose nodes intersect $\bm X\setminus\{X_j,X_k\}$, then  $P_{X_j\mid\bm C,\bm S=\bm 1}^{(k)}\ne P_{X_j\mid\bm C,\bm S=\bm 1}^{(0)}$.
        \item For any $X_j\in\ancestor_{\Gcal_{\mid\bm C}}(\bm S)\cap\bm X$ and $X_k\not\in\ancestor_{\Gcal_{\mid\bm C}}(X_j)$ such that there exists an open path between $X_j,X_k$ given $\bm S$ that is disjoint from $\bm X\setminus\{X_j,X_k\}$, we have $P_{X_k\mid\bm C,\bm S=\bm 1}^{(j)}\ne P_{X_k\mid\bm C,\bm S=\bm 1}^{(0)}$.
        \item For any $X_j\in\ancestor_{\Gcal_{\mid\bm C}}(\bm S)\cap\bm X$, if there is no selection variables $S\in\bm S$ and directed path $X_j\rightarrow\ldots\rightarrow S$ whose nodes intersect $\bm X\setminus\{X_j\}$, then for any $X_k\in\bm X\setminus\{X_j\}$ such that there is an open path between $X_j,X_k$ given $\bm S$ that is disjoint from $\bm X\setminus\{X_j,X_k\}$, we have $P_{X_k\mid\bm C,\bm S=\bm 1}^{(j)}\ne P_{X_k\mid\bm C,\bm S=\bm 1}^{(0)}$.
    \end{enumerate}
\end{assumption}

\begin{assumption}[Faithful Conditional Dependence]\label{asm:faithful_adjacency_latent_selection}
    \begin{enumerate}[1)]
        \item For all $j,k\in[d_X]$, $j>k$, if $\big(X_j\not\indep X_k\mid \bm A_j\cup\bm A_k\cup\bm A_{\bm S}\cup\bm S\setminus\{X_j,X_k\}\big)_{\Gcal_{\mid\bm C}}$, then $X_j\not\indep X_k\mid\bm S=\bm 1, \bm C,\bm A_j\cup\bm A_k\cup\bm A_{\bm S}\setminus\{X_j,X_k\}$.
        \item For all $j\in[d_X]$, $k\in\bm A_j$, if $\big(I_k\not\indep X_j\mid\bm A_j\cup\bm A_{\bm S}\cup\bm S\setminus\{X_j\}\big)_{\Aug(\Gcal_{\mid\bm C},I_k)}$, then $P_{X_j\mid\bm S=\bm 1, \bm C\cup\bm A_j\cup\bm A_{\bm S}\setminus\{X_j\}}^{(k)}\ne P_{X_j\mid\bm S=\bm 1,\bm C\cup\bm A_j\cup\bm A_{\bm S}\setminus\{X_j\}}^{(0)}$.
    \end{enumerate}
\end{assumption}

To control the FWER of Algorithm \ref{alg:conf_MAG_latent_selection}, we assume that the parallel tests in Stage I and II each achieve the desired FWER control. Although the tests conducted in Stage II are selected based on the results of Stage I and the two stages use the same data, the following assumption requires FWER control for Stage II only when the true anterior sets $\bm A_j,\bm A_{\bm S}$ are used.
\begin{assumption}[${\rm FWER}_1$]\label{asm:FWER_latent_selection}
    The ${\rm FWER}_1$ of the fully parallel anterior tests $\{\psi_j^{\anterior(k)}:k\in[d_X],j\in[d_X]\setminus\{k\}\}$ are controlled below $\frac{\alpha}{2}$. And the ${\rm FWER}_1$ of the fully parallel adjacency tests $\big\{\psi_{j,k}^{\rm Adj}\big(\bm A_{\bm S},\bm A_{[d_X]}\big):j,k\in[d_X],j>k\big\}\cup\big\{\psi_j^{{\rm Adj}(k)}\big(\bm A_{\bm S},\bm A_{[d_X]}\big):j\in[d_X],k\in\bm A_j\big\}$ are also below $\frac{\alpha}{2}$.
\end{assumption}

Let FWER$_2^\anterior$ and FWER$_2^{\rm Adj}$ denote the family-wise type II error rates of the anterior and adjacency tests in Assumption \ref{asm:FWER_latent_selection}, respectively, with the latter evaluated at the true anterior sets. Under Assumptions \ref{asm:context}, \ref{asm:faithful_descendant_latent_selection}, and \ref{asm:FWER_latent_selection}, the following proposition shows that the anterior sets $\bm A_j$ and $\bm A_{\bm S}$ can be estimated under powerful anterior tests. 

\begin{proposition}[Coverage and Power for Anterior Sets]\label{prop:consistent_anterior_latent_selection}
    Under Assumptions \ref{asm:context} and \ref{asm:FWER_latent_selection}, the estimated anterior sets $\widehat{\bm A}_j$ and $\widehat{\bm A}_{\bm S}$ do not contain incorrect nodes,
    \[\Prob\big(\exists j\in[d_X]{\rm ~s.t.~}\widehat{\bm A}_j\not\subseteq\bm A_j{\rm ~or~}\widehat{\bm A}_{\bm S}\not\subseteq\bm A_{\bm S}\big)\le\frac{\alpha}{2}.\]
    In addition, if Assumption \ref{asm:faithful_descendant_latent_selection} holds, 
    \[\Prob\big(\exists j\in[d_X]{\rm ~s.t.~}\widehat{\bm A}_j\ne\bm A_j{\rm ~or~}\widehat{\bm A}_{\bm S}\ne\bm A_{\bm S}\big)\le\frac{\alpha}{2}+{\rm FWER}_2^\anterior.\]
\end{proposition}

Given the estimators of the anterior sets, the following theorem shows that the constructed confidence set $\Ccal_\alpha$ has the desired coverage. 

\begin{theorem}[Coverage]\label{thm:coverage_latent_selection}
    Under Assumptions \ref{asm:context}, \ref{asm:faithful_descendant_latent_selection}, and \ref{asm:FWER_latent_selection}, 
    \[\Prob\big(\Ccal_\alpha\not\subseteq\Ecal(\TCM(\Gcal_{\mid\bm C},\bm I))\big)\le\alpha+{\rm FWER}_2^\anterior.\]
\end{theorem}

In Theorem \ref{thm:coverage_latent_selection}, the additional term FWER$_2^{\anterior}$ accounts for missed anterior relations, which can affect subsequent conditioning sets and edge orientations. Such error propagation is not unique to our procedure. In classical PC/FCI-type algorithms, successive updates to the estimated graph determine the conditioning sets and tests in subsequent rounds, allowing testing errors to propagate repeatedly \citep{spirtes2000causation,spirtes1995causal,spirtes2001anytime,kocaoglu2019characterization,luo2025characterization}. In contrast, our procedure restricts propagation between testing stages to a single step, from the anterior tests in Stage I to the adjacency tests in Stage II. Moreover, when no context variables are present, the anterior tests reduce to univariate two-sample tests for equality of distributions, for which nonparametric tests consistent against any fixed alternative are available \citep{lehmann2005testing}. For example, consider Kolmogorov-Smirnov statistic \citep{lehmann2005testing} with Bonferroni correction and threshold chosen by the Dvoretzky-Kiefer-Wolfowitz inequality \citep{massart1990tight}, we have 
\[{\rm FWER}_2^\anterior\le 4d_X^2\exp\bigg(-\frac{n_{\min}}{2}\bigg(\Delta_{\min}-\sqrt{\frac{2}{n_{\min}}\log\frac{8d_X^2}{\alpha}}\bigg)_+^2\bigg)\rightarrow 0\quad\text{as}\quad n_{\min}\rightarrow \infty,\]
with $n_{\min}=\min_{0\le k\le d_X}n_k$, $\Delta_{\min}=\min_{j\ne k,F_j^{(k)}\ne F_j^{(0)}}\|F_j^{(k)}-F_j^{(0)}\|_\infty$ and $F_j^{(k)}$ be the CDF of $P_{X_j\mid\bm S=\bm 1}^{(k)}$. Therefore, $\Ccal_\alpha$ is asymptotically valid.

In addition, if the faithfulness Assumption \ref{asm:faithful_adjacency_latent_selection} holds, $\Ccal_\alpha$ is guaranteed to recover $\TCM(\Gcal_{\mid\bm C},\bm I)$. 
\begin{theorem}[Power]\label{thm:power_latent_selection}
    Under Assumptions \ref{asm:context}, \ref{asm:faithful_descendant_latent_selection}, \ref{asm:faithful_adjacency_latent_selection}, and \ref{asm:FWER_latent_selection}, 
    \[\Prob\big(\Ccal_\alpha\ne\Ecal(\TCM(\Gcal_{\mid\bm C},\bm I))\big)\le\alpha+{\rm FWER}_2^\anterior+{\rm FWER}_2^{\rm Adj}.\]
\end{theorem}

\begin{remark}[Applications to Downstream Inference]
    The constructed lower confidence set $\Ccal_\alpha$ can be applied to many downstream causal structure inference problems by checking the alignments between the causal structures and $\Ccal_\alpha$. Suppose we want to test the existence of latent confounding in $\Gcal_{\mid\bm C}$, we can set the test statistic as $\1\big(\text{exists }X_j\leftrightarrow X_k,\text{ or }I_j\rightarrow X_k,\text{ or }\{I_j-X_k,I_k-X_j\}\text{ in }\Ccal_\alpha\big)$. Then, Theorem \ref{thm:coverage_latent_selection} implies the asymptotic type I error control. If we want to test the existence of selection effect in $\Gcal_{\mid\bm C}$, the test statistic can be $\1\big(\widehat{\bm A}_{\bm S}\ne\emptyset\big)$, and Proposition \ref{prop:consistent_anterior_latent_selection} guarantees the type I error control.
\end{remark}

Algorithm \ref{alg:conf_MAG_latent_selection} is agnostic to the specific choice of CTS and CI tests. In Section C of the supplement \citep{hou2026supplement}, we show that the generalized covariance measure (GCM), originally developed for CI testing \citep{shah2020hardness}, is also applicable to CTS testing and provides doubly robust asymptotic type I error control. Power properties of GCM and related tests have been established under specific signal and conditions \citep{shah2020hardness,lundborg2024projected,scheidegger2022weighted}.

\subsection{Constraint-Query Optimality}

If $\bm C=\bm L=\bm S=\emptyset$, then $\TCM(\Gcal_{\mid\bm C},\bm I)=\Gcal$. By Theorem \ref{thm:power_latent_selection}, Algorithm \ref{alg:conf_MAG_latent_selection} identifies $\Gcal$ through at most $\frac{5}{2}d_X^2$ statistical tests. The $\frac{1}{2}d_X(d_X-1)$ possible adjacencies among $d_X$ system variables also suggest a quadratic lower bound. 

\begin{definition}[$m$-Constraint-Based Algorithm]\label{def:algorithm}
    Given a set of distributions $\Pcal_\Gcal$ associated with a graph $\Gcal$ of $\bm X$, a deterministic algorithm $\Acal$ is $m$-constraint-based if it maps the distributions $\Pcal_\Gcal$ to a graph, and accesses $\Pcal_\Gcal$ only through testing $m$ adaptively selected constraints. At step $j$, $\Acal$ selects a constraint $c_j$ based on the history $\big((c_l,\psi_l):l\in[j-1]\big)$, and queries the testing oracle for a response $\psi_j\in\{0,1\}$, indicating whether $c_j$ is satisfied. After $m$ queries, $\Acal$ outputs a graph based on $\big((c_j,\psi_j):\psi_j\in\{0,1\},j\in[m]\big)$.
\end{definition}

This algorithm class accommodates adaptive testing, including PC \citep{spirtes2000causation} and FCI algorithm \citep{spirtes1995causal,spirtes2001anytime}. Algorithm \ref{alg:conf_MAG_latent_selection} belongs to the class with $m\le \frac{5}{2}d_X^2$. The following theorem shows that any constraint-based algorithm requires $\frac{1}{2}d_X(d_X-1)$ queries in the worst case. Therefore, Algorithm \ref{alg:conf_MAG_latent_selection} achieves constraint-query optimality up to a multiplicative constant. 

\begin{theorem}[Constraint-Query Lower Bound]\label{thm:lower_bound}
    For any $m<\frac{1}{2}d_X(d_X-1)$, any $m$-constraint-based algorithm $\Acal$, and any distributions class $\big\{\Pcal_\Gcal:\text{$\Gcal$ is a DAG on }\bm X\big\}$, there exist a DAG $\Gcal$ such that $\Acal\big(\Pcal_\Gcal\big)\ne\Gcal$.
\end{theorem}

\section{Inference under Selection in the Absence of Latent Variables}\label{sec:selection}

In this section, we consider the setting where there is no latent variable in $\Gcal$, but there might be selection, i.e., $\bm Z=(\bm C, \bm X, \bm S)$. Here, we do not assume all system variables are intervened, and thus allow $K \le d_X$. In this case, we will show that the parent nodes of intervened variables in $\Gcal_{\mid\bm C}$ can be identified without introducing the MAG. To this end, we introduce the induced subgraph \citep{richardson2002ancestral} $\Gcal_{\bm X}$ of $\Gcal_{\mid\bm C}$ on $\bm X$ with edges
\[\Ecal(\Gcal_{\bm X})=\big\{\text{edge in }\Gcal_{\mid\bm C}: \text{both end nodes in }\bm X\big\}.\]

Our algorithm is motivated by the following lemma, which shows that a single test determines the parent set of the intervened variable. 

\begin{lemma}[Parents of Intervention Targets]\label{lem:parent_intervention} In the absence of latent variables, 
    for any $k\in[K]$ and $j\in[d_X]\setminus\{k\}$, we have $X_j\in\parent_{\Gcal_{\mid\bm C}}(X_k)$ if and only if $\big(I_k\not\indep X_j\mid\bm X_{-j},\bm S\big)_{\Aug(\Gcal_{\mid\bm C},I_k)}$.
\end{lemma}

Lemma \ref{lem:parent_intervention} implies that $K(d_X-1)$ fully parallel CTS tests, which we call \textbf{parent tests},  
\[P^{(k)}_{X_j\mid\bm C,\bm X_{-j},\bm S=\bm 1}=P_{X_j\mid\bm C,\bm X_{-j},\bm S=\bm 1}^{(0)},\quad j\ne k\]
identify the parent sets of intervened nodes. We summarize the procedure in Algorithm \ref{alg:conf_DAG_selection}.

\begin{algorithm}
\footnotesize
\caption{\footnotesize Lower Confidence Set for $\Ecal(\Gcal_{\bm X})$}
\label{alg:conf_DAG_selection}
\begin{algorithmic}[1]
\Statex{\bf Input:} Observational data $\Dcal_0$, interventional data $\Dcal_k$, $k\in[K]$, confidence level $1-\alpha$.
\Statex{\bf Output:} Lower Confidence set $\Ccal_\alpha$ for $\Ecal(\Gcal_{\bm X})$.
\State{\bf Step 1:} Set $\Ccal_\alpha=\emptyset$.
\State{\bf Step 2:} For $k\in[K]$, $j\in[d_X]\setminus\{k\}$, conduct CTS test $\psi_j^{P(k)}\in\{0,1\}$ for
\[H_{j,0}^{P(k)}:\Prob\big(p^{(k)}_{X_j\mid\bm C,\bm X_{-j},\bm S=\bm 1}\ne p_{X_j\mid\bm C,\bm X_{-j},\bm S=\bm 1}^{(0)}\big)=0\quad \text{vs} \quad H_{j,1}^{P(k)}:\Prob\big(p^{(k)}_{X_j\mid\bm C,\bm X_{-j},\bm S=\bm 1}\ne p_{X_j\mid\bm C,\bm X_{-j},\bm S=\bm 1}^{(0)}\big)>0.\]
with $\text{FWER}_1$ controlled below $\alpha$.
\State{\bf Step 3:} For $k\in[K]$, $j\in[d_X]\setminus\{k\}$, if $\psi_j^{P(k)}=1$, append $X_j\rightarrow X_k$ to $\Ccal_\alpha$.
\end{algorithmic}
\end{algorithm}

If all the parallel tests have family-wise type I error rate below $\alpha$, we can show that the coverage of $\Ccal_\alpha$ can be guaranteed without any faithfulness assumption.

\begin{assumption}[FWER$_1$]\label{asm:FWER_selection}
    The ${\rm FWER}_1$ of the fully parallel CTS parent tests $\big\{\psi_j^{P(k)}:k\in[K],j\in[d_X]\setminus\{k\}\big\}$ are controlled below $\alpha$.
\end{assumption}

\begin{theorem}[Coverage]\label{thm:coverage_selection}
    Under Assumptions \ref{asm:context} and \ref{asm:FWER_selection}, $\Ccal_\alpha$ has the desired coverage, i.e.,
    \[\Prob\big(\Ccal_\alpha\not\subseteq\Ecal(\Gcal_{\bm X})\big)\le\alpha.\]
\end{theorem}

To illustrate the power of Algorithm \ref{alg:conf_DAG_selection}, we make the following faithfulness assumption about the conducted parent tests.

\begin{assumption}[Faithful Conditional Dependence]\label{asm:faithful_parent_selection}
    For any $k\in[K]$ and $j\in[d_X]\setminus\{k\}$, if $\big(I_k\not\indep X_j\mid\bm X_{-j},\bm S\big)_{\Aug(\Gcal_{\mid\bm C},I_k)}$, then $\Prob\big(p_{X_j\mid\bm C,\bm X_{-j},\bm S=\bm 1}^{(k)}\ne p_{X_j\mid\bm C,\bm X_{-j},\bm S=\bm 1}^{(0)}\big)>0$.
\end{assumption}

Then we can guarantee that $\Ccal_\alpha$ identifies all the parent sets of the intervened variables. Particularly, if $K=d_X$, we recover $\Gcal_{\bm X}$.

\begin{theorem}[Power]\label{thm:power_selection}
    Under Assumptions \ref{asm:context}, \ref{asm:FWER_selection} and \ref{asm:faithful_parent_selection}, 
    \[\Prob\big(\Ccal_\alpha\ne \{X_j\rightarrow X_k: k\in[K], X_j\in\parent_{\Gcal_{\mid\bm C}}(X_k)\}\big)\le\alpha+{\rm FWER}_2^{P},\]
    where ${\rm FWER}_2^{P}$ is the family-wise type II error rate of the parent tests. If in addition $K=d_X$, 
    \[\Prob\big(\Ccal_\alpha\ne\Ecal(\Gcal_{\bm X})\big)\le\alpha+{\rm FWER}_2^P.\]
\end{theorem}

When only a subset of observed variables are intervened on, Algorithm \ref{alg:conf_MAG_latent_selection} discards all variables without intervention and learns a MAG on the intervened variables. Since edges in this latent projected graph may not exist in the MAG on all observed variables, it is unclear how to combine the learned MAG and unintervened variables for a finer causal discovery. Unlike Algorithm \ref{alg:conf_MAG_latent_selection}, even if $K<d_X$, all edges recovered by Algorithm \ref{alg:conf_DAG_selection} are in $\Gcal_{\bm X}$. Then, we can treat $\Ccal_\alpha$ as prior knowledge for standard algorithms such as PC and FCI to learn the remaining structures in $\Gcal_{\bm X}$ more efficiently.

A line of work \citep{fan2024environment,gu2025causality,peters2016causal,gu2025fundamental} aims to identify causal parents of the response through many interventions on covariates using the invariance principle. While \cite{gu2025fundamental} shows that learning from invariance is NP-hard even under linear models, Lemma \ref{lem:parent_intervention} and Algorithm \ref{alg:conf_DAG_selection} imply that a single intervention on the response variable and $d$ tests are sufficient for identifying the parents of the response from $d$ covariates.

\section{Numerical Results}\label{sec:numerical}

In this section, we evaluate the numerical performance of the proposed methods using both simulation data and real data.

\subsection{Simulation Evaluations}\label{sec:simulation}

For the simulation study, we consider two data-generating models with continuous (M1) and count variables (M2), respectively, with interventions on all system variables, i.e., $K=d_X$. Due to limited space, we only present the count model (M2) and defer the complete model specification of the continuous model (M1) to Section E of the supplement \citep{hou2026supplement}.

\begin{enumerate}[(M1)]
    \item $\bm Z=(\bm C,\bm X,\bm L,\bm S)$, $\bm C\in\R^2$, $\bm X\in\R^{20}$, $\bm L^{(k)}\in\R^{d_L}$, $\bm S\in\{0,1\}^{d_S}$.
    \item $\bm Z=(\bm C,\bm X,\bm L,\bm S)$, $\bm C=(C_1,C_2,\ell)$ with $(C_1^{(k)},C_2^{(k)})\sim N(\bm 0,I_2), \log \ell^{(k)}\sim N(1,0.01)$, $\bm X\in\R^{20}$, $\bm L\in\R^{d_L}$, $\bm S\in\{0,1\}^{d_S}$, $\epsilon_{L_j}^{(k)}\sim N(0,0.2^2)$, $\epsilon_j^{(k)}\sim N(0,0.2^2)$,
    \[L_j^{(k)}\mid\bm C^{(k)},\epsilon_{L_j}^{(k)}\sim {\rm NB}\big(\ell^{(k)}\exp(\text{clip}_{[-10,10]}(b_{j,1}C_1^{(k)}+b_{j,2}C_2^{(k)}+\epsilon_{L_j}^{(k)})), 10\big),\]
    \[X_j^{(k)}\mid \parent_{\Gcal_{\mid\bm C}}(X_j^{(k)}),\bm C^{(k)},\epsilon_j^{(k)}\sim{\rm NB}\big(\ell\exp(\text{clip}_{[-10,10]}(\eta_j^{(k)})),10\big),\]
    \begin{eqnarray*}\eta_j^{(k)}=\bm\gamma_j^\top\bm C^{(k)}&+&\sum_{Z_l^{(k)}\in\parent_{\Gcal_{\mid\bm C}}(X_j^{(k)})}\beta_{j,l}H_l^{(k)}\\
    &+&0.02\sum_{\{Z_l^{(k)},Z_m^{(k)}\}\subset\parent_{\Gcal_{\mid\bm C}}(X_j^{(k)})}\beta_{j,l}\beta_{j,m}H_l^{(k)}H_m^{(k)}-4\1(j=k)+\epsilon_j^{(k)},
    \end{eqnarray*}
    \[H_l^{(k)}=\log\bigg(1+\frac{10^4Z_l^{(k)}}{\ell^{(k)}}\bigg),\quad S_i=\1\bigg(\sum_{Z_l\in\parent_{\Gcal}(S_i)}a_{i,l}Z_l>\tau\bigg),\quad b_{j,l},\gamma_{j,l},\beta_{j,l},a_{j,l}\in\{0.2,-0.2\},\]
   and  $\tau$ is the largest real number such that $\Prob_{P^{(k)}}(\bm S=\bm 1) \ge 0.05$ for all $k=0,\ldots,20$.
\end{enumerate}

In models (M1) and (M2), we have $d_X=20$ observed variables with fixed topological ordering $X_1,\ldots,X_{20}$. $0-2$ parents are randomly selected from preceding variables. Each latent variable has two randomly selected observed children, and each selection variable has two randomly selected observed parents. The true DAGs are reported in Section E of the supplement \citep{hou2026supplement}. In both models, the control group and each intervention group have $n_k=500$ samples. (M1) considers continuous $\bm X$ with two context variables. (M2) mimics the Perturb-seq data, where we consider negative binomial (NB) count data $\bm X$ and include the library size $\ell$ as an additional context variable. For the $d_S$ selection variables and $d_L$ latent variables, we consider three combinations $(d_S,d_L)\in\{(2,2),(2,0),(0,2)\}$. 

In addition to Algorithms \ref{alg:conf_MAG_latent_selection} and \ref{alg:conf_DAG_selection}, to demonstrate the necessity of dealing with selection, we introduce another Algorithm D.1 in the supplement \citep{hou2026supplement} by assuming the absence of selection in Algorithm \ref{alg:conf_MAG_latent_selection}. We also compare with the state-of-the-art $\Fcal$-FCI algorithm in \cite{luo2025characterization}. In Algorithms \ref{alg:conf_MAG_latent_selection}, \ref{alg:conf_DAG_selection} and D.1, we use the GCM-based test described in Section C of the supplement \citep{hou2026supplement} for all tests, with FWER control at $\alpha=0.05$ implemented via multiplier bootstrap \citep{shah2020hardness}. The nuisance functions in GCM are estimated by Super Learner \citep{van2007super}. For $\Fcal$-FCI, we follow the default implementation in \cite{luo2025characterization}, which chooses KCI \citep{KCI} for CI tests with significance levels $\alpha=0.1$ in skeleton search and $\alpha=0.05$ in orientation. We didn't apply a multiplicity correction to $\Fcal$-FCI because its released implementation uses per-test significance thresholds and adaptively determines which CI tests to perform. To adjust for the context variables, we additionally condition on $\bm C$ in all KCI tests in $\Fcal$-FCI.

Although our algorithm is designed for learning causal graphs, Remark \ref{rem:interpretation_latent_selection} reveals that some information has a more direct causal interpretation. $\bm A_j$ is the causally anterior set of $X_j$. When there is no isolated selection, $\bm A_{\bm S}$ contains all observed ancestors of selection. Any pair in $\bm E_{\bm L}$ \eqref{eq:EL} indicates the existence of latent confounding.

Since Algorithms \ref{alg:conf_MAG_latent_selection}, D.1 learn $\TCM(\Gcal_{\mid\bm C},\bm I)$ on $(\bm X,\bm I)$, Algorithm \ref{alg:conf_DAG_selection} learns $\Gcal_{\bm X}$, and $\Fcal$-FCI learns the $\Fcal$-partial ancestral graph ($\Fcal$-PAG) on $\bm X$, these algorithms are not directly comparable. However, when $K=d_X$, both $\TCM(\Gcal_{\mid\bm C},\bm I)$ and $\Fcal$-PAG imply $\bm A_j$, $\bm A_{\bm S}$ and $\MAG(\Gcal_{\mid\bm C})$. When $d_L=0$, $\Gcal_{\bm X}$ can be identified from both $\TCM(\Gcal_{\mid\bm C},\bm I)$ and $\Fcal$-PAG. In summary, we consider 6 inferencial targets $\bm A_j,\bm A_{\bm S},\bm E_{\bm L},\TCM(\Gcal_{\mid\bm C},\bm I),\MAG(\Gcal_{\mid\bm C}),\Gcal_{\bm X}$ and conduct 100 experiments. For each target and algorithm, we calculate the FWER$_1$ and FWER$_2$ among the 100 experiments and calculate the average false discovery rate (FDR) and false negative rate (FNR). All the simulation results are reported in Tables \ref{tab:error_MAG_latent_selection}, \ref{tab:error_MAG_selection} and \ref{tab:error_MAG_latent}. The sizes of these true targets are reported in Table \ref{tab:simulation_size}.

\begin{table}
    \caption{\footnotesize Sizes of true inferential targets under 3 combinations of $(d_S,d_L)$ in models (M1) and (M2).}
    \label{tab:simulation_size} 
    \footnotesize
    \centering
    \begin{tabular*}{\columnwidth}{@{\extracolsep\fill}cccccccc@{\extracolsep\fill}}\hline
    Model & $(d_S,d_L)$ & $\sum_{j\in[d_X]}|\bm A_j|$ & $|\bm A_{\bm S}|$ & $|\bm E_{\bm L}|$ & $|\Ecal(\TCM(\Gcal_{\mid\bm C},\bm I))|$ & $|\Ecal(\MAG(\Gcal_{\mid\bm C}))|$ & $|\Ecal(\Gcal_{\bm X})|$ \\\hline
    \multirow{3}{*}{(M1)} & (2,2) & 74 & 6 & 2 & 30 & 20 & 16 \\
     & (2,0) & 98 & 7 & 0 & 26 & 20 & 18 \\
     & (0,2) & 40 & 0 & 2 & 21 & 20 & 18 \\[8pt]
    \multirow{3}{*}{(M2)} & (2,2) & 237 & 12 & 2 & 73 & 50 & 37 \\
     & (2,0) & 237 & 12 & 0 & 66 & 45 & 37 \\
     & (0,2) & 100 & 0 & 2 & 42 & 41 & 37 \\\hline
    \end{tabular*}
 
\end{table}

\begin{table}
      \caption{\footnotesize Errors of Algorithms \ref{alg:conf_MAG_latent_selection}, \ref{alg:conf_DAG_selection}, D.1 and $\Fcal$-FCI under (M1) and (M2) with selection and latent variables $d_S=2, d_L=2$. Alg \ref{alg:conf_MAG_latent_selection}: proposed algorithm allowing for latent variables and selection; Alg \ref{alg:conf_DAG_selection}: proposed algorithm in the absence of latent variables;  Alg D.1: proposed algorithm in the absence of selection; $\Fcal$-FCI:  FCI algorithm of \cite{luo2025characterization}.}
    \label{tab:error_MAG_latent_selection} \footnotesize
    \centering
    \begin{tabular*}{\columnwidth}{@{\extracolsep\fill}cccccccccc@{\extracolsep\fill}}\hline
     & & \multicolumn{4}{c}{(M1)} & \multicolumn{4}{c}{(M2)}\\\cline{3-6}\cline{7-10}
    Target & Methods & FWER$_1$ & FWER$_2$ & FDR & FNR & FWER$_1$ & FWER$_2$ & FDR & FNR\\\hline
    \multirow{3}{*}{$\bm A_j$} & Alg \ref{alg:conf_MAG_latent_selection} & 0.020 & 0.050 & 0.007 & 0.009 & 0.010 & 0.060 & 0.002 & 0.009 \\
    & Alg D.1 & 0.000 & 1.000 & 0.000 & 0.518 & 0.010 & 1.000 & 0.002 & 0.279 \\
    & $\Fcal$-FCI & 0.480 & 0.250 & 0.053 & 0.046 & 0.560 & 1.000 & 0.038 & 0.589 \\[8pt]
    \multirow{2}{*}{$\bm A_{\bm S}$} & Alg \ref{alg:conf_MAG_latent_selection} & 0.020 & 0.010 & 0.007 & 0.003 & 0.010 & 0.030 & 0.003 & 0.009 \\
    & $\Fcal$-FCI & 0.390 & 0.180 & 0.108 & 0.035 & 0.500 & 1.000 & 0.114 & 0.513 \\[8pt]
    \multirow{2}{*}{$\bm E_{\bm L}$} & Alg \ref{alg:conf_MAG_latent_selection} & 0.030 & 0.040 & 0.010 & 0.020 & 0.020 & 1.000 & 0.015 & 0.645 \\
    & Alg D.1 & 0.950 & 0.000 & 0.432 & 0.000 & 1.000 & 0.890 & 0.846 & 0.570 \\[8pt]
    \multirow{2}{*}{$\TCM(\Gcal_{\mid\bm C},\bm I)$} & Alg \ref{alg:conf_MAG_latent_selection} & 0.060 & 0.080 & 0.008 & 0.008 & 0.070 & 1.000 & 0.009 & 0.239 \\
    & Alg D.1 & 1.000 & 1.000 & 0.459 & 0.603 & 1.000 & 1.000 & 0.556 & 0.751 \\[8pt]
    \multirow{3}{*}{$\MAG(\Gcal_{\mid\bm C})$} & Alg \ref{alg:conf_MAG_latent_selection} & 0.050 & 0.030 & 0.004 & 0.002 & 0.060 & 1.000 & 0.005 & 0.153 \\
    & Alg D.1 & 1.000 & 1.000 & 0.172 & 0.189 & 1.000 & 1.000 & 0.174 & 0.427 \\
    & $\Fcal$-FCI & 0.990 & 1.000 & 0.117 & 0.164 & 1.000 & 1.000 & 0.278 & 0.610 \\[8pt]
    $\Gcal_{\bm X}$ & Alg \ref{alg:conf_DAG_selection} & 1.000 & 0.000 & 0.209 & 0.000 & 0.530 & 1.000 & 0.064 & 0.752 \\\hline
    
    \end{tabular*}
 
\end{table}

\begin{table}
    \caption{\footnotesize Errors of Algorithms \ref{alg:conf_MAG_latent_selection}, \ref{alg:conf_DAG_selection}, D.1 and $\Fcal$-FCI under (M1) and (M2) with selection in the absence of latent variables $d_S=2, d_L=0$. Alg \ref{alg:conf_MAG_latent_selection}: proposed algorithm allowing for latent variables and selection; Alg \ref{alg:conf_DAG_selection}: proposed algorithm in the absence of latent variables;  Alg D.1: proposed algorithm in the absence of selection; $\Fcal$-FCI:  FCI algorithm of \cite{luo2025characterization}.}
    \label{tab:error_MAG_selection}
    \footnotesize
    \centering
    \begin{tabular*}{\columnwidth}{@{\extracolsep\fill}cccccccccc@{\extracolsep\fill}}\hline
     & & \multicolumn{4}{c}{(M1)} & \multicolumn{4}{c}{(M2)}\\\cline{3-6}\cline{7-10}
    Target & Methods & FWER$_1$ & FWER$_2$ & FDR & FNR & FWER$_1$ & FWER$_2$ & FDR & FNR \\\hline
    \multirow{3}{*}{$\bm A_j$} & Alg \ref{alg:conf_MAG_latent_selection} & 0.010 & 0.000 & 0.001 & 0.000 & 0.000 & 0.170 & 0.000 & 0.052 \\
    & Alg D.1 & 0.010 & 1.000 & 0.000 & 0.613 & 0.010 & 1.000 & 0.001 & 0.393 \\
    & $\Fcal$-FCI & 0.500 & 0.540 & 0.070 & 0.173 & 0.670 & 1.000 & 0.047 & 0.549 \\[8pt]
    \multirow{2}{*}{$\bm A_{\bm S}$} & Alg \ref{alg:conf_MAG_latent_selection} & 0.000 & 0.000 & 0.000 & 0.000 & 0.000 & 0.170 & 0.000 & 0.071 \\
    & $\Fcal$-FCI & 0.460 & 0.420 & 0.108 & 0.077 & 0.640 & 1.000 & 0.163 & 0.540 \\[8pt]
    \multirow{2}{*}{$\bm E_{\bm L}$} & Alg \ref{alg:conf_MAG_latent_selection} & 0.040 & 0.000 & 0.040 & 0.000 & 0.010 & 0.000 & 0.010 & 0.000 \\
    & Alg D.1 & 1.000 & 0.000 & 1.000 & 0.000 & 0.930 & 0.000 & 0.930 & 0.000 \\[8pt]
    \multirow{2}{*}{$\TCM(\Gcal_{\mid\bm C},\bm I)$} & Alg \ref{alg:conf_MAG_latent_selection} & 0.040 & 0.910 & 0.002 & 0.039 & 0.190 & 1.000 & 0.043 & 0.214 \\
    & Alg D.1 & 1.000 & 1.000 & 0.481 & 0.538 & 1.000 & 1.000 & 0.601 & 0.758 \\[8pt]
    \multirow{3}{*}{$\MAG(\Gcal_{\mid\bm C})$} & Alg \ref{alg:conf_MAG_latent_selection} & 0.020 & 0.910 & 0.001 & 0.051 & 0.180 & 1.000 & 0.019 & 0.102 \\
    & Alg D.1 & 1.000 & 1.000 & 0.174 & 0.164 & 1.000 & 1.000 & 0.203 & 0.346 \\
    & $\Fcal$-FCI & 0.980 & 0.980 & 0.096 & 0.094 & 1.000 & 1.000 & 0.322 & 0.586 \\[8pt]
    \multirow{3}{*}{$\Gcal_{\bm X}$} & Alg \ref{alg:conf_MAG_latent_selection} & 0.000 & 0.000 & 0.000 & 0.000 & 0.190 & 1.000 & 0.011 & 0.179 \\
    & Alg \ref{alg:conf_DAG_selection} & 0.050 & 0.000 & 0.003 & 0.000 & 0.060 & 1.000 & 0.008 & 0.790 \\
    & Alg D.1 & 1.000 & 1.000 & 0.194 & 0.110 & 0.990 & 0.820 & 0.070 & 0.076 \\
    & $\Fcal$-FCI & 0.930 & 0.260 & 0.185 & 0.016 & 1.000 & 1.000 & 0.263 & 0.425 \\\hline
    
    \end{tabular*}
 
\end{table}

\begin{table}
   \caption{\footnotesize Errors of Algorithms \ref{alg:conf_MAG_latent_selection}, \ref{alg:conf_DAG_selection}, D.1 and $\Fcal$-FCI under (M1) and (M2) with latent variables in the absence of selection $d_S=0, d_L=2$. Alg \ref{alg:conf_MAG_latent_selection}: proposed algorithm allowing for latent variables and selection; Alg \ref{alg:conf_DAG_selection}: proposed algorithm in the absence of latent variables;  Alg D.1: proposed algorithm in the absence of selection; $\Fcal$-FCI:  FCI algorithm of \cite{luo2025characterization}.}
    \label{tab:error_MAG_latent}
    \footnotesize
    \centering
    \begin{tabular*}{\columnwidth}{@{\extracolsep\fill}cccccccccc@{\extracolsep\fill}}\hline
     & & \multicolumn{4}{c}{(M1)} & \multicolumn{4}{c}{(M2)}\\\cline{3-6}\cline{7-10}
    Target & Methods & FWER$_1$ & FWER$_2$ & FDR & FNR & FWER$_1$ & FWER$_2$ & FDR & FNR\\\hline
    \multirow{3}{*}{$\bm A_j$} & Alg \ref{alg:conf_MAG_latent_selection} & 0.000 & 0.000 & 0.000 & 0.000 & 0.010 & 0.000 & 0.000 & 0.000 \\
    & Alg D.1 & 0.000 & 0.000 & 0.000 & 0.000 & 0.010 & 0.000 & 0.000 & 0.000 \\
    & $\Fcal$-FCI & 0.790 & 0.070 & 0.174 & 0.005 & 0.900 & 1.000 & 0.143 & 0.413 \\[8pt]
    \multirow{2}{*}{$\bm A_{\bm S}$} & Alg \ref{alg:conf_MAG_latent_selection} & 0.000 & 0.000 & 0.000 & 0.000 & 0.000 & 0.000 & 0.000 & 0.000 \\
    & $\Fcal$-FCI & 0.750 & 0.000 & 0.750 & 0.000 & 0.850 & 0.000 & 0.850 & 0.000 \\[8pt]
    \multirow{2}{*}{$\bm E_{\bm L}$} & Alg \ref{alg:conf_MAG_latent_selection} & 0.000 & 0.000 & 0.000 & 0.000 & 0.000 & 0.000 & 0.000 & 0.000 \\
    & Alg D.1 & 0.000 & 0.000 & 0.000 & 0.000 & 0.000 & 0.000 & 0.000 & 0.000 \\[8pt]
    \multirow{2}{*}{$\TCM(\Gcal_{\mid\bm C},\bm I)$} & Alg \ref{alg:conf_MAG_latent_selection} & 0.000 & 0.000 & 0.000 & 0.000 & 0.000 & 0.030 & 0.000 & 0.001 \\
    & Alg D.1 & 0.000 & 0.000 & 0.000 & 0.000 & 0.000 & 0.040 & 0.000 & 0.001 \\[8pt]
    \multirow{3}{*}{$\MAG(\Gcal_{\mid\bm C})$} & Alg \ref{alg:conf_MAG_latent_selection} & 0.000 & 0.000 & 0.000 & 0.000 & 0.000 & 0.030 & 0.000 & 0.001 \\
    & Alg D.1 & 0.000 & 0.000 & 0.000 & 0.000 & 0.000 & 0.040 & 0.000 & 0.001 \\
    & $\Fcal$-FCI & 0.980 & 0.990 & 0.088 & 0.073 & 1.000 & 1.000 & 0.155 & 0.442 \\[8pt]
    $\Gcal_{\bm X}$ & Alg \ref{alg:conf_DAG_selection} & 1.000 & 0.010 & 0.214 & 0.001 & 0.840 & 1.000 & 0.162 & 0.808 \\\hline
    
    \end{tabular*}
  
\end{table}

In Table \ref{tab:error_MAG_latent_selection} with $d_S=2,d_L=2$, we have both Algorithms \ref{alg:conf_DAG_selection} and D.1 to be misspecified. For Algorithm \ref{alg:conf_MAG_latent_selection}, the FWER$_1$'s for $\bm A_j,\bm A_{\bm S}$ are controlled below $\frac{\alpha}{2}=0.025$ for both (M1) and (M2). The FWER$_1$'s for $\bm E_{\bm L},\TCM(\Gcal_{\mid\bm C},\bm I),\MAG(\Gcal_{\mid\bm C})$ are approximately at the level of $\alpha=0.05$. Meanwhile, other algorithms all have a significant FWER$_1$. Moreover, the average FDR and FNR are also smaller than those of other algorithms. Compared to (M1), Algorithm \ref{alg:conf_MAG_latent_selection} has worse FWER$_2$ and FNR under (M2). This is largely due to the less information contained in the count data and the additional noise introduced by the negative binomial link. For $\bm E_{\bm L}$ under (M2), we have $|\bm E_{\bm L}|=2$, but Algorithm \ref{alg:conf_MAG_latent_selection} only identifies at most one true pair in $\bm E_{\bm L}$ in all 100 experiments, leading to an FWER$_2=1$ and an FNR of 0.645. For Algorithm D.1, due to the no-selection assumption, when both $X_j\in\bm A_k$ and $X_k\in\bm A_j$ appear, Algorithm D.1 has to discard one of the ancestry structures, leading to a large FNR for $\bm A_j$. For Algorithm \ref{alg:conf_DAG_selection}, due to $d$-connections implied through latent variables, the parent test in Algorithm \ref{alg:conf_DAG_selection} fails to control the type I error, leading to a large FWER$_1$ and FDR.

In Table \ref{tab:error_MAG_selection} with $d_S=2,d_L=0$, only Algorithm D.1 is misspecified. Under (M1), both Algorithms \ref{alg:conf_MAG_latent_selection} and \ref{alg:conf_DAG_selection} have controlled FWER$_1$ with small FDR and FNR. Under (M2), Algorithm \ref{alg:conf_MAG_latent_selection} has a FWER$_1=0.19$ for $\TCM(\Gcal_{\mid\bm C},\bm I)$. Among the 19 experiments with false discovery, 17 of them have false negatives in Stage I of Algorithm \ref{alg:conf_MAG_latent_selection} due to a weak selection effect in our parameter realization. This leads to false negatives in $\widehat{\bm A}_j$ and $\widehat{\bm A}_{\bm S}$ in the 17 experiments. Since the type I error control for Stage II relies on the correct detection of $\bm A_j$ and $\bm A_{\bm S}$, these false negatives in Stage I lead to an inflated FWER$_1$, which aligns with Theorem \ref{thm:coverage_latent_selection}. For Algorithm \ref{alg:conf_DAG_selection} under (M2), the strong knockdown of $I_k$ makes the intervened variable $X_k$ almost degenerate at 0, which is poorly overlapped with $X_k$ in the control group. Then the conditioning set for the test $I_k, X_j\mid\bm C, \bm X_{-j},\bm S=\bm 1$ becomes almost disjoint under control and intervention. This lack of comparable information leads to low power of the conditional two-sample tests. Algorithm D.1 and $\Fcal$-FCI have a significant FWER$_1$ with larger FDR and FNR.

In Table \ref{tab:error_MAG_latent} with $d_S=0,d_L=2$, Algorithm \ref{alg:conf_DAG_selection} is misspecified. Both Algorithms \ref{alg:conf_MAG_latent_selection} and D.1 have near-perfect FWER$_1$, FWER$_2$, FDR, and FNR. Whereas Algorithm \ref{alg:conf_DAG_selection} and $\Fcal$-FCI have significant FWER$_1$ with larger FDR and FNR.

\subsection{Analysis of A549 lung cancer cells stimulated with interferon-$\beta$}

For real data analysis, we focus on a Perturb-seq dataset consisting of A549 lung cancer cells stimulated with interferon-$\beta$ (IFN-$\beta$) that was reported in \cite{jiang2025systematic}, to study the causal gene regulatory relationships within the IFN signaling pathway. 
In this experiment, cells expressing a CRISPRi system (dCas9–KRAB–MeCP2) were subjected to targeted knockdown of 61 IFN-$\beta$ pathway regulator genes, each perturbed using multiple single guide RNAs (sgRNAs), alongside control cells with non-targeting guides. Following 24 hours of IFN-$\beta$ stimulation, single-cell RNA sequencing yielded an expression matrix containing 34,025 gene features (including non-coding RNA genes, pseudogenes, and about 20,000 protein-coding genes) across 32,233 cells, with 124--950 cells per perturbation. The resulting data provide a high-dimensional mapping from targeted regulatory perturbations to genome-wide gene expression responses under a fixed signaling context.

We treat batch and the percentage of mitochondrial transcripts from the original data as context variables, following the covariate-adjustment practice in \cite{barry2024robust}. Batch accounts for systematic technical variation across experimental batches, whereas mitochondrial transcript percentage serves as a proxy for cell quality. For each gene, we normalize the UMI count by library size and apply a log transformation, $X_{ij}=\log\big(1+\frac{10^4g_{ij}}{\ell_i}\big)$, where $g_{ij}$ denotes the UMI count of the $j$th gene in the $i$th cell, and $\ell_i$ denotes the total UMI count in the $i$th cell $i$. We then select the intervention groups for which the intervention causes a significant change in the distribution of the intervened gene conditional on context variables, at a significance level of 0.05. Restricting the analyzed system to the corresponding target genes gives $d_X=51$ and $d_C=2$. The final dataset contains 27,555 cells, including 1,379 controls.

We implement Algorithm \ref{alg:conf_MAG_latent_selection} in the same way as in Section \ref{sec:simulation}, resulting in a graph $\Ccal_\alpha$ with 246 edges. For simplicity, we only report some implied causal information. Figure E.7 in the supplement \citep{hou2026supplement} encodes the anterior sets $\widehat{\bm A}_j$, which characterize the genes that need to be adjusted when studying causal structures (Propositions \ref{prop:adjacency_latent_selection} and \ref{prop:direction_latent_selection}). Among the 51 genes, we detect 6 genes in $\bm A_{\bm S}\subseteq\ancestor_{\Gcal_{\mid\bm C}}(\bm S)\cap\bm X$, including IFNAR1, IFNAR2, JAK1, STAT2, TYK2, USP18. Figure \ref{fig:selection} illustrates the construction of $\widehat{\bm A}_{\bm S}$. The CTS tests show $\text{USP18}\in\widehat{\bm A}_\text{STAT2}$ and $\text{STAT2}\in\widehat{\bm A}_\text{USP18}$, this mutual influence implies $\{\text{USP18, STAT2}\}\subseteq\widehat{\bm A}_{\bm S}$. Then the detected anteriority $\{\text{TYK2, IFNAR1, IFNAR2, JAK1}\}\subseteq\widehat{\bm A}_\text{STAT2}$ concludes all the 6 genes are in $\widehat{\bm A}_{\bm S}$. This finding is biologically plausible because IFN-$\beta$ exerts antiproliferative pressure on A549 cells, such that alterations in IFNAR--JAK--STAT signaling affect cell growth and the likelihood that specific subclones are retained in the population \citep{hiebinger2023tumour}. Moreover, USP18 protects A549 cell growth under interferon stimulation by non-catalytically attenuating JAK--STAT signaling \citep{clancy2023isgylation}.

\begin{figure}[!htbp]
    \centering
    \includegraphics[width=0.7\linewidth]{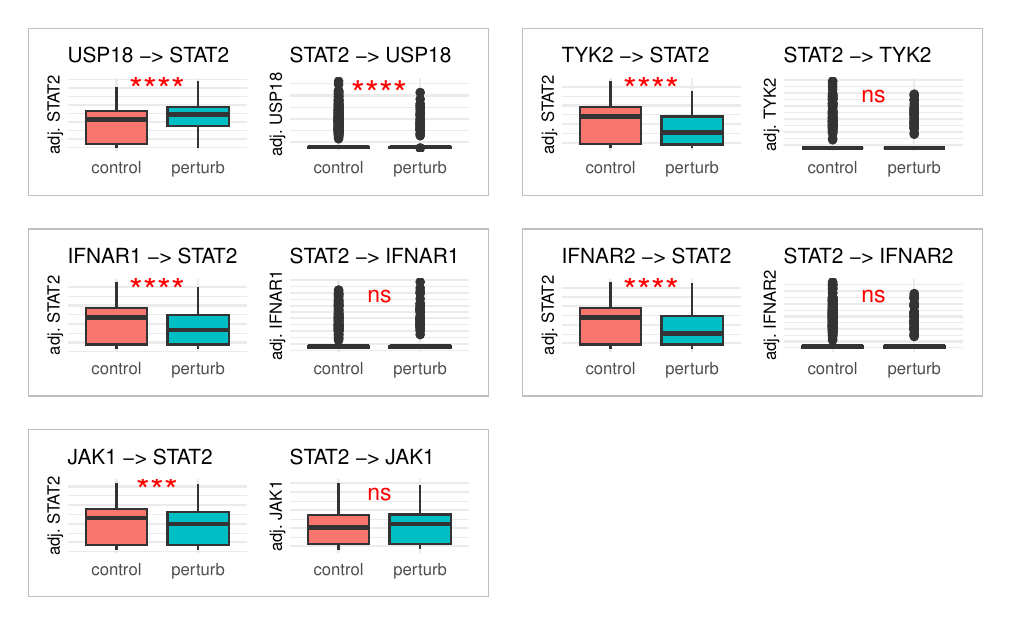}
    \caption{\footnotesize Demonstration of selection ancestors in $\bm A_{\bm S}$. The panel title $X_j\rightarrow X_k$ indicates that $X_j$ is the intervention target and we assess a distribution shift in $X_k$ conditional on context variables. Boxplots show the regression residuals of $X_k$ on context variables for control and $j$th interventional groups. Red stars indicate significance of CTS tests, ns indicates not significant.}
    \label{fig:selection}
\end{figure}

We identified 200 gene pairs in $\widehat{\bm E}_{\bm L}$ exhibiting graph features that imply the presence of latent confounding under Remark \ref{rem:interpretation_latent_selection}, where the latent confounders may include genes that are outside the 51-gene system. The prevalence of such pairs is plausible because our analysis is restricted to only 51 genes, representing a small subset of the underlying gene regulatory system. Numerous unmeasured genes may act as common causes of the analyzed genes and therefore appear as latent confounders in the causal graph. The 200 pairs in $\widehat{\bm E}_{\bm L}$ are reported in Section E in the supplement \citep{hou2026supplement}. Figure \ref{fig:confounder} illustrates 6 randomly selected bidirected edges. Consider TRIM22 $\leftrightarrow$ IFI16 as an example. Intervening on either gene doesn't induce a conditional distribution shift in the other, providing no evidence of an ancestral relationship between them. Nevertheless, TRIM22 and IFI16 remain dependent (with CI test $p$-value $\approx 10^{-19}$) after conditioning on their detected anterior genes, supporting the presence of latent confounding between them.

\begin{figure}[!htbp]
    \centering
    \includegraphics[width=0.9\linewidth]{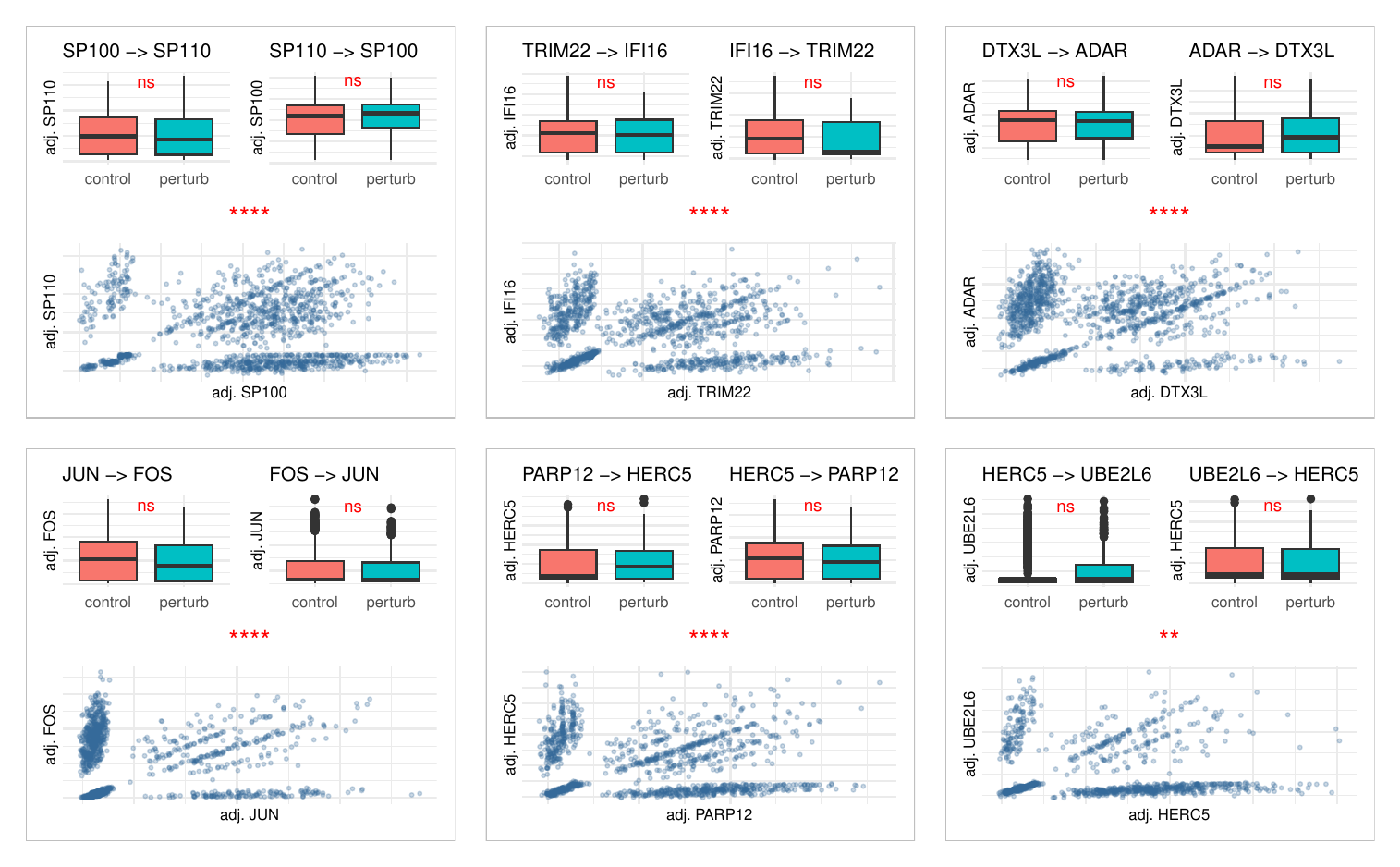}
    \caption{\footnotesize Demonstration of unmeasured confounders for 6 bidirected edges. The boxplot title $X_j\rightarrow X_k$ indicates that $X_j$ is the intervention target and we assess a distribution shift in $X_k$ conditional on context variables. Boxplots show the regression residuals of $X_k$ on context variables for control and $j$th interventional groups. Scatterplots show the regression residuals of $X_k$ versus those of $X_j$, after adjustment for context variables and anterior genes. Red stars indicate significance of CTS tests for boxplots and CI tests for scatterplots, ns indicates not significant.}
    \label{fig:confounder}
\end{figure}

A directed edge $X_j\rightarrow X_k$ in $\TCM(\Gcal_{\mid\bm C},\bm I)$ for $X_j\in \ancestor_{\Gcal_{\mid\bm C}}(\bm S)$ may reflect ancestry of selection rather than causal ancestry $X_j\in\ancestor_{\Gcal_{\mid\bm C}}(X_k)$. By Remark \ref{rem:interpretation_latent_selection}, undirected edges detect some ancestors of selection, and arrowheads at a variable rule out its ancestry of selection. We therefore divide the 51 genes into 3 groups, $|\widehat{\bm A}_{\bm S}|=6,|\widehat{\bm U}_{\bm S}|=7,|\widehat{\bm N}_{\bm S}|=38$, where $\widehat{\bm N}_{\bm S}$ contains all nodes with incoming arrowheads and $\widehat{\bm U}_{\bm S}=\bm X\setminus\widehat{\bm A}_{\bm S}\setminus\widehat{\bm N}_{\bm S}$. We restrict causal ancestry interpretations to directed edges whose source genes belong to $\widehat{\bm N}_{\bm S}$. Considering both directly detected edges $X_j\rightarrow X_k$ and edges $X_j\rightarrow X_k$ implied by $I_j\rightarrow X_k$ according to Lemma \ref{lem:interpretation_latent_selection}, we identify 3 edges, including STAT3 $\rightarrow$ IRF1, IFI16 $\rightarrow$ ID3, ADAR $\rightarrow$ NFKB1. We also detect 2 directed edges with source genes in $\widehat{\bm U}_{\bm S}$, including ID1 $\rightarrow$ ELK1 and DRAP1 $\rightarrow$ TRIM22. For these pairs, their ancestral interpretation requires further inspection. Figure \ref{fig:causal} contrasts interventions in both directions for these five directed edges. Consider STAT3 $\rightarrow$ IRF1 as an example. Intervening on STAT3 induces a significant conditional distribution shift in IRF1, whereas the reverse comparison is not significant. Since both genes are not ancestor of selection, this asymmetric structure implies that STAT3 is an ancestor of IRF1.

\begin{figure}[!htbp]
    \centering
    \begin{subfigure}{0.4\textwidth}
        \centering
        \includegraphics[width=\linewidth]{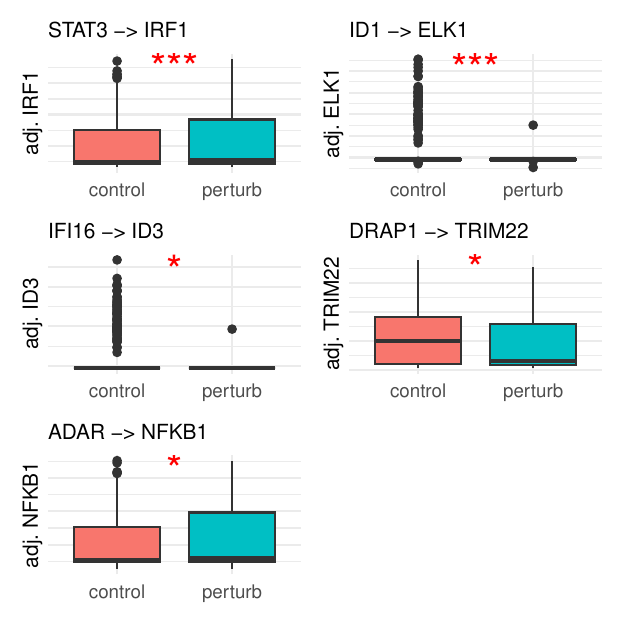}
        \caption{\footnotesize Causal directions.}
    \end{subfigure}
    \hfill
    \begin{subfigure}{0.4\textwidth}
        \centering
        \includegraphics[width=\linewidth]{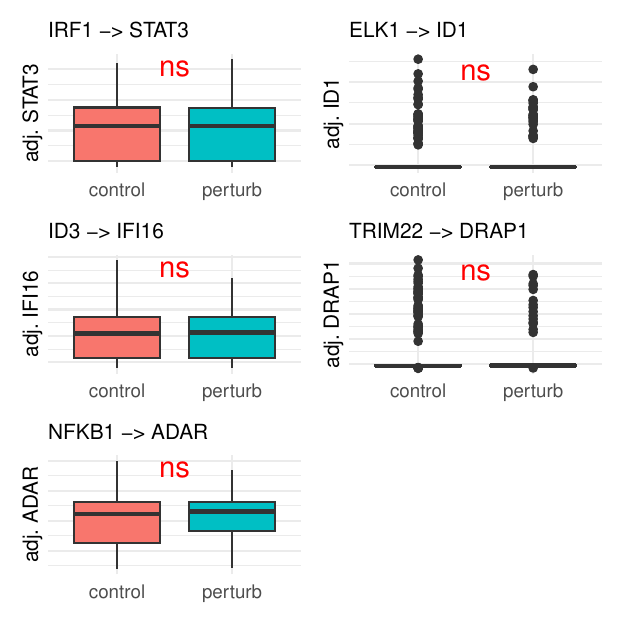}
        \caption{\footnotesize Reverse causal directions.}
    \end{subfigure}
    \caption{\footnotesize Demonstration of 5 learned directed edges. The panel title $X_j\rightarrow X_k$ indicates that $X_j$ is the intervention target and we assess a distribution shift in $X_k$ conditional on context variables. Boxplots show the regression residuals of $X_k$ on context variables for the control and the $j$th interventional groups. Red stars indicate significance of CTS tests, ns indicates not significant.}
    \label{fig:causal}
\end{figure}

Finally, to evaluate the sensitivity of Algorithm \ref{alg:conf_MAG_latent_selection} to the choice of test statistics, we replace the GCM-based test with RCIT \citep{strobl2019approximate}. Table \ref{tab:error_MAG_RCIT} reports the differences between the resulting estimates, using the results obtained with GCM as the reference. The two approaches yield identical estimates of $\bm A_{\bm S}$ and closely agree on $\bm A_j$. These differences propagate to the subsequent estimates of $\bm E_{\bm L},\TCM(\Gcal_{\mid\bm C},\bm I)$ and $\MAG(\Gcal_{\mid\bm C})$, for which the FDR and FNR remain below 0.18. Overall, the results show that the main structural conclusions are reasonably stable to the choice of test statistics given the relatively small number of cells for some perturbations. 

\begin{table}[!htbp]
    \caption{\footnotesize Sensitivity of Algorithm \ref{alg:conf_MAG_latent_selection} to the choice of test statistics. FDR and FNR are errors of the graph obtained by replacing the GCM-based test with RCIT, where we treat the graph under GCM as the reference.}
    \label{tab:error_MAG_RCIT}
    \footnotesize
    \centering
    \begin{tabular*}{\columnwidth}{@{\extracolsep\fill}cccccc@{\extracolsep\fill}}\hline
     & $\bm A_j$ & $\bm A_{\bm S}$ & $\bm E_{\bm L}$ & $\TCM(\Gcal_{\mid\bm C},\bm I)$ & $\MAG(\Gcal_{\mid\bm C})$ \\\hline
     FDR & 0.000 & 0.000 & 0.128 & 0.136 & 0.176\\
     FNR & 0.058 & 0.000 & 0.115 & 0.118 & 0.138\\\hline
    \end{tabular*}
 
\end{table}

\section{Discussion}

We develop an inference framework for causal discovery under latent variables and selection bias using single-target interventions. This result provides a basis for downstream causal structural testing and causal inference. 

Several directions remain for future research. When interventions are available only for a subset of observed variables, our algorithm restricts attention to the intervention targets and treats the remaining observed variables as latent. An important extension is to retain information from these untargeted variables and develop efficient procedures for the IMEC over all observed system variables.

Another direction concerns the additional information available from hard interventions. Although $\TCM(\Gcal_{\mid\bm C},\bm I)$ is identifiable under per-node interventions, it doesn't recover all underlying causal relations when both selection and latent variables are present. Hard interventions remove incoming edges to their targets, including those from latent causes, and may therefore provide additional identifying information. Extending our framework to exploit this information would allow us to investigate which remaining causal ambiguities can be resolved and how stronger identification results can be combined with efficient statistical inference.

\section{Acknowledgements}

This research was supported by NIH grants U01HG013841 and R01GM129781.

\bibliographystyle{plainnat}

\bibliography{reference}

\newpage

\begin{center}
    {\Large Supplement to ``Statistical Inference for Causal Discovery under Selection and Latent Variables via Single-Target Interventions''}
\end{center}

\appendix

\numberwithin{lemma}{section}
\numberwithin{equation}{section}
\numberwithin{definition}{section}
\numberwithin{assumption}{section}
\numberwithin{corollary}{section}
\numberwithin{algorithm}{section}
\numberwithin{remark}{section}
\numberwithin{theorem}{section}
\numberwithin{proposition}{section}
\numberwithin{example}{section}
\numberwithin{figure}{section}

\section{Example: \texorpdfstring{$\MAG(\Aug(\Gcal_{\mid\bm C},\bm I))$}{MAG(Aug(G,I))} Cannot Represent IMEC}
\begin{example}\label{exa:I-I}

    Consider two DAGs $\Gcal,\Gcal'$ and their augmented graphs $\Aug(\Gcal,\bm I),\Aug(\Gcal',\bm I)$ with $\bm I=(I_1,I_2)$ as in Figure \ref{fig:counterexample_aug}. One can verify that for any $\tilde\Gcal\in\{\Gcal,\Gcal'\}$ and $j,k\in[2]$, $(X_1\not\indep X_2\mid\bm S)_{\tilde\Gcal}$, $(I_j\not\indep X_k\mid\bm S)_{\Aug(\tilde\Gcal, I_j)}$, and $(I_j\not\indep X_k\mid X_{3-k},\bm S)_{\Aug(\tilde\Gcal,I_j)}$. Therefore $\Gcal'\in\IMEC(\Gcal)$. But $\MAG(\Aug(\Gcal,\bm I))$ and $\MAG(\Aug(\Gcal',\bm I))$ have different skeleton.

    \begin{figure}[!htbp]
    \centering
    \begin{subfigure}{0.45\textwidth}
        \centering
        \includegraphics[width=0.8\linewidth]{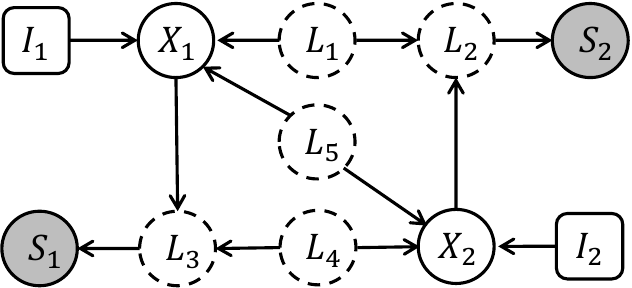}
        \caption{$\Aug(\Gcal,\bm I)$}
    \end{subfigure}
    \begin{subfigure}{0.45\textwidth}
        \centering
        \includegraphics[width=0.8\linewidth]{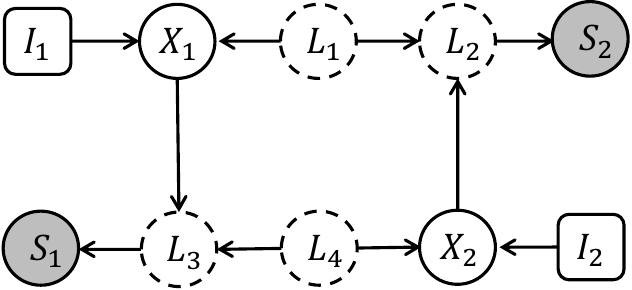}
        \caption{$\Aug(\Gcal',\bm I)$}
    \end{subfigure}

    \begin{subfigure}{0.45\textwidth}
        \centering
        \includegraphics[width=0.3\linewidth]{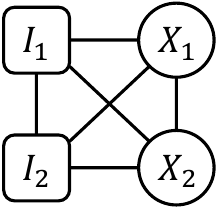}
        \caption{$\MAG(\Aug(\Gcal,\bm I))$}
    \end{subfigure}
    \begin{subfigure}{0.45\textwidth}
        \centering
        \includegraphics[width=0.3\linewidth]{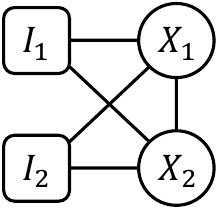}
        \caption{$\MAG(\Aug(\Gcal',\bm I))$}
    \end{subfigure}
    \caption{Causal structures in Example \ref{exa:I-I}.}
    \label{fig:counterexample_aug}
\end{figure}
\end{example}

\section{Interpretation of the Trimmed Combined-MAG when $K=d_X$}\label{supp:sec:interpret}

While the trimmed combined MAG $\TCM(\Gcal_{\mid\bm C},\bm I)$ is identifiable under per-node interventions, it is not immediately clear how to extract the underlying causal relations from $\TCM(\Gcal_{\mid\bm C},\bm I)$. This section discusses the interpretation of $\TCM(\Gcal_{\mid\bm C},\bm I)$ based on the existence of selection and latent variables. Throughout Section \ref{supp:sec:interpret}, we assume there is a single-target intervention for each observed variable, i.e., $K=d_X$.

\subsection{Interpretation in the Absence of Latent Variable}\label{sec:interpret_selection}

In this section, we consider the case where there is no latent variable. Although $\TCM(\Gcal_{\mid\bm C},\bm I)$ may contain undirected edges, we will show that the causal relations among $\bm X$ can be uniquely identified. Denote $\Gcal_{\bm X}$ as the induced subgraph \citep{richardson2002ancestral} of $\Gcal_{\mid\bm C}$, and thus $\Gcal$, on $\bm X$, with edges
\begin{equation}\label{eq:GX}
    \Ecal(\Gcal_{\bm X})=\big\{\text{edge in }\Gcal_{\mid\bm C}: \text{both end nodes in }\bm X\big\}.
\end{equation}
The following lemma shows that if $K=d_X$ and there is no latent variable, $\Gcal_{\bm X}$ can be recovered from $\TCM(\Gcal_{\mid\bm C},\bm I)$, and thus is identifiable according to Theorem \ref{thm:unique}. 

\begin{lemma}[Identification of $\Gcal_{\bm X}$]\label{lem:identify_GX_selection}
    If $K=d_X$ and $d_L=0$, $\Gcal_{\bm X}$ can be uniquely recovered from $\TCM(\Gcal_{\mid\bm C},\bm I)$ through the following rules:
    \begin{enumerate}[1)]
        \item if $X_j, X_k$ are not adjacent in $\TCM(\Gcal_{\mid\bm C},\bm I)$, then they are not adjacent in $\Gcal_{\bm X}$,
        \item if $X_j\rightarrow X_k$ is in $\TCM(\Gcal_{\mid\bm C},\bm I)$, then $X_j\rightarrow X_k$ is in $\Gcal_{\bm X}$,
        \item if $X_j-X_k$ is in $\TCM(\Gcal_{\mid\bm C},\bm I)$, then $X_j\rightarrow X_k$ is in $\Gcal_{\bm X}$ if and only if $I_k-X_j$ is in $\TCM(\Gcal_{\mid\bm C},\bm I)$.
    \end{enumerate}
\end{lemma}

We illustrate the identification rule of $\Gcal_{\bm X}$ through the following example.

\begin{example}\label{exa:selection}

    Consider the DAG $\Gcal_{\mid\bm C}$ without latent variables shown in Figure \ref{fig:example_selection}. $X_3$ is a direct cause of $X_2, X_4$, but the relation between $X_1, X_4$ is due to selection bias. These underlying causal relations among $\bm X$ are represented in $\Gcal_{\bm X}$. 

    Now we recover $\Gcal_{\bm X}$ from $\TCM(\Gcal_{\mid\bm C},\bm I)$. Since $(X_1,X_2), (X_1,X_3), (X_2,X_4)$ are not adjacent in $\TCM(\Gcal_{\mid\bm C},\bm I)$, they are also not adjacent in $\Gcal_{\bm X}$. Since $X_3\rightarrow X_2$ is in $\TCM(\Gcal_{\mid\bm C},\bm I)$, it also exists in $\Gcal_{\bm X}$. For the undirected edge $X_1-X_4$, since neither $I_1-X_4$ nor $I_4-X_1$ exists in $\TCM(\Gcal_{\mid\bm C},\bm I)$, we know there is no direct causal relation between $X_1, X_4$ and the undirected edge is merely due to selection effects. For the undirected edge $X_3-X_4$, edge $I_4-X_3$ implies the inducing path $I_4\rightarrow X_4\leftarrow X_3$. Therefore, $X_3$ is a direct cause of $X_4$.

    \begin{figure}
        \centering
        \begin{subfigure}{0.32\textwidth}
            \centering
            \includegraphics[width=0.5\linewidth]{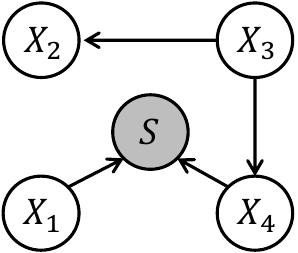}
            \caption{$\Gcal_{\mid\bm C}$}
        \end{subfigure}
        \begin{subfigure}{0.32\textwidth}
            \centering
            \includegraphics[width=0.95\linewidth]{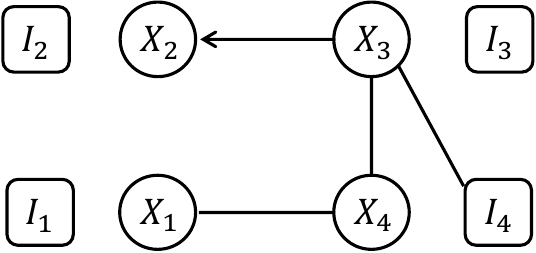}
            \caption{$\TCM(\Gcal_{\mid\bm C},\bm I)$}
        \end{subfigure}
        \begin{subfigure}{0.32\textwidth}
            \centering
            \includegraphics[width=0.5\linewidth]{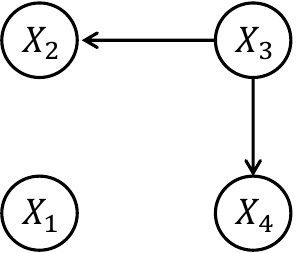}
            \caption{$\Gcal_{\bm X}$}
        \end{subfigure}
        \caption{Causal structures in Example \ref{exa:selection}.}
        \label{fig:example_selection}
    \end{figure}
\end{example}

\subsection{Interpretation in the Absence of Selection}\label{sec:interpret_latent}

In this section, we focus on the case where there is no selection variable. Although the graph $\Gcal_{\mid\bm C}$ is not identifiable, we can extract some information about it from $\TCM(\Gcal_{\mid\bm C},\bm I)$.

\begin{lemma}\label{lem:interpretation_latent}
    If $d_S=0$, then for $j\ne k\in[d_X]$,
    \begin{enumerate}[1)]
        \item if $X_j\rightarrow X_k$ is in $\TCM(\Gcal_{\mid\bm C},\bm I)$, then $X_j\in\ancestor_{\Gcal_{\mid\bm C}}(X_k)$ and all inducing paths in $\Gcal_{\mid\bm C}$ between $X_j,X_k$ are into $X_k$,
        \item if $X_j\leftrightarrow X_k$ is in $\TCM(\Gcal_{\mid\bm C},\bm I)$, then $X_j\not\in\ancestor_{\Gcal_{\mid\bm C}}(X_k),X_k\not\in\ancestor_{\Gcal_{\mid\bm C}}(X_j)$ and all inducing paths in $\Gcal_{\mid\bm C}$ between $X_j,X_k$ have the form $X_j\leftarrow L\ldots L'\rightarrow X_k$ for some $L,L'\in\bm L$,
        \item $I_j\rightarrow X_k$ is in $\TCM(\Gcal_{\mid\bm C},\bm I)$ if and only if $X_j\in\ancestor_{\Gcal_{\mid\bm C}}(X_k)$ and $\Gcal_{\mid\bm C}$ contains an inducing path $X_j\leftarrow L\ldots L'\rightarrow X_k$ for some $L,L'\in\bm L$.
    \end{enumerate}
\end{lemma}

\begin{remark}\label{rem:interpretation_latent}
    As shown in Lemma \ref{lem:confounding_seq}, any inducing path of the form $X_j\leftarrow L\ldots L'\rightarrow X_k$ implies the existence of a sequence of latently confounded nodes. Therefore, if $X_j\leftrightarrow X_k$ or $I_j\rightarrow X_k$ is in $\TCM(\Gcal_{\mid\bm C},\bm I)$, there exist a sequence of nodes $X_{j_0}\overset{\triangle}{=}X_j, X_{j_1},\ldots,X_{j_m},X_{j_{m+1}}\overset{\triangle}{=}X_k\in\bm X$ such that $X_{j_{l-1}}$ is latently confounded with $X_{j_l}$ for $l\in[m+1]$. Conversely, if $X_j$ and $X_k$ are latently confounded, Lemma \ref{lem:interpretation_latent} implies that one of $X_j\leftrightarrow X_k$, $I_j\rightarrow X_k$ or $I_k\rightarrow X_j$ must appear.
\end{remark}

We illustrate this interpretation through the following example.

\begin{example}\label{exa:latent}
    Consider the DAG $\Gcal_{\mid\bm C}$ without selection variables shown in Figure \ref{fig:example_latent}, we have $X_1$ depends on $X_2$ purely through the latent confounder $L_2$, and $X_2$ causes $X_3$ indirectly through the latent mediator $L_2$. $X_3$ is a direct cause of $X_4$ and they share a latent confounder $L_3$. Similarly, $X_4$ is a direct cause of $X_5$, and they share a latent confounder $L_4$.

    In $\TCM(\Gcal_{\mid\bm C},\bm I)$, all the directed edges among $\bm X$ represent causal effects, although they may be indirect through mediators. For example, $X_2\to X_3$ in $\TCM(\Gcal_{\mid\bm C},\bm I)$ is through the latent mediator $L_2$. Likewise, for $X_2\to X_5$ in $\TCM(\Gcal_{\mid\bm C},\bm I)$, the causal relation is through mediators $(L_2,X_3,X_4)$. 
    
    In $\TCM(\Gcal_{\mid\bm C},\bm I)$, the bidirected edge $X_1\leftrightarrow X_2$ indicates that $X_1$ and $X_2$ are not causally related, but are affected by latent variables through an inducing path as in part 2) of Lemma \ref{lem:interpretation_latent}. In $\Gcal_{\mid\bm C}$, this inducing path equals $X_1\leftarrow L_1\rightarrow X_2$. 

    The graph $\TCM(\Gcal_{\mid\bm C},\bm I)$ also contains additional directed edges $I_3\to X_4, I_3\to X_5$ and $I_4\to X_5$. For instance, $I_3\to X_5$ in $\TCM(\Gcal_{\mid\bm C},\bm I)$ is due to the inducing path $I_3\to X_3\leftarrow L_3\to X_4\leftarrow L_4\to X_5$ in $\Aug(\Gcal_{\mid\bm C},I_3)$. Although we cannot recover this path based on the observed variables, part 3) of Lemma \ref{lem:interpretation_latent} ensures the existence of latent variables and some path of the form $X_3\leftarrow L\ldots L'\rightarrow X_5$. 

    \begin{figure}
        \centering
        \begin{subfigure}{0.45\textwidth}
            \centering
            \includegraphics[width=0.55\linewidth]{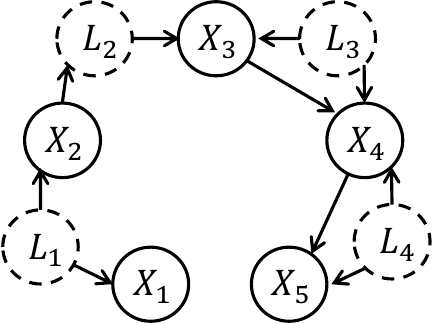}
            \caption{$\Gcal_{\mid\bm C}$}
        \end{subfigure}
        \begin{subfigure}{0.45\textwidth}
            \centering
            \includegraphics[width=0.8\linewidth]{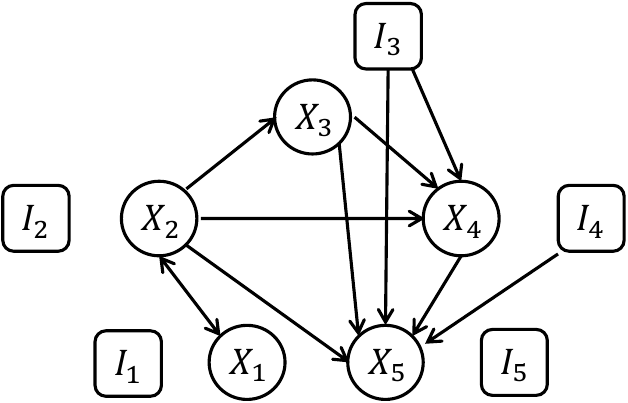}
            \caption{$\TCM(\Gcal_{\mid\bm C},\bm I)$}
        \end{subfigure}
        \caption{Causal structures in Example \ref{exa:latent}.}
        \label{fig:example_latent}
    \end{figure}
\end{example}

\section{Conditional Two-Sample Tests}\label{sec:CTS}

Algorithms \ref{alg:conf_MAG_latent_selection} and \ref{alg:conf_DAG_selection} crucially rely on the CTS tests and CI tests. While CI testing has been extensively studied \citep{shah2020hardness, lundborg2024projected, strobl2019approximate, neykov2021minimax}, CTS testing has gained attention only recently \citep{lee2024general,hu2024two,chen2025biased,yan2022distance,chatterjee2024kernel}. Although these two problems seem closely related, there are subtle differences between them \citep{lee2024general}. \cite{lee2024general} proposed a general sampling strategy to convert any CI test into a CTS test. But this sampling strategy incurs additional randomness and may make the test more conservative. 

In this section, we show that some CI tests are directly applicable to the CTS problem, and the additional sampling step is not always necessary. We focus on the generalized covariance measure (GCM) in \cite{shah2020hardness} for its simplicity and flexibility, but we believe similar results hold for general CI tests based on asymptotically linear smooth functional estimation \citep{lundborg2024projected, scheidegger2022weighted} as well.

Suppose we observe two datasets $\Dcal_0=\{(Y_i,\bm X_i):i\in[n_0]\}\overset{\rm i.i.d.}{\sim}P_{Y,\bm X}^{(0)}$ and $\Dcal_1=\{(Y_i,\bm X_i):i\in[n_1]+n_0\}\overset{\rm i.i.d.}{\sim}P_{Y,\bm X}^{(1)}$ with $Y\in\R$, $\bm X\in\R^{d_X}$, and the goal is testing
\[H_0:\Prob\big(p_{Y\mid\bm X}^{(0)}\ne p_{Y\mid\bm X}^{(1)}\big)=0\quad\longleftrightarrow \quad H_1:\Prob\big(p_{Y\mid\bm X}^{(0)}\ne p_{Y\mid\bm X}^{(1)}\big)>0,\]
where the probability $\Prob$ is over both distributions $\frac{1}{2}P_{\bm X}^{(1)}+\frac{1}{2}P_{\bm X}^{(0)}$. An intuitive way is to introduce a binary variable $I\in\{0,1\}$ for each sample and set $I_i=0$ for $i\in[n_0]$, $I_i=1$ for $i\in[n_1]+n_0$. Since
\[Y\mid(\bm X,I=0)\overset{d}{=}Y\mid(\bm X,I=1)\quad\Longleftrightarrow\quad Y\indep I\mid\bm X,\]
we could then apply GCM to $\Dcal=\{(Y_i,\bm X_i, I_i):i\in[N]\}$ with $N=n_0+n_1$, as described in Algorithm \ref{alg:GCM}. 

\begin{algorithm}
    \caption{GCM for CTS Test}
    \label{alg:GCM}
    \begin{algorithmic}
        \State{\bf Input:} Datasets $\Dcal_0=\{(Y_i,\bm X_i):i\in[n_0]\}$, $\Dcal_1=\{(Y_i,\bm X_i):i\in[n_1]+n_0\}$, significance level $\alpha$.
        \State{\bf Output:} Decision $\psi_{N,\alpha}\in\{0,1\}$ for $H_0\longleftrightarrow H_1$.
        \State{\bf Step 1:} Set $I_i=0$ for $i\in[n_0]$ and $I_i=1$ for $i\in[n_1]+n_0$. Define $\Dcal=\{(Y_i,\bm X_i,I_i):i\in[N]\}$ and $P_{Y,\bm X,I}=\frac{n_0}{N}P_{Y,\bm X}^{(0)}\1(I=0)+\frac{n_1}{N}P_{Y,\bm X}^{(1)}\1(I=1)$.
        \State{\bf Step 2:} Fit $\hat f(\bm X)$ for $f_P(\bm X)=\E_{P_{Y,\bm X,I}}(I\mid\bm X)$ and $\hat g(\bm X)$ for $g_P(\bm X)=\E_{P_{Y,\bm X,I}}(Y\mid\bm X)$.
        \State{\bf Step 3:} Set $R_i=\{I_i-\hat f(\bm X_i)\}\{Y_i-\hat g(\bm X_i)\}$ and compute
        \[T_N=\sqrt{N}\frac{\hat\E_{\Dcal}R}{\sqrt{\widehat{\Var}_{\Dcal}(R)}}.\]
        \State{\bf Step 4:} Set $\psi_{N,\alpha}=\1(|T_N|>z_{\frac{\alpha}{2}})$, where $z_{\frac{\alpha}{2}}$ is the upper $\frac{\alpha}{2}$-quantile of standard normal.
    \end{algorithmic}
\end{algorithm}

Although the variables $I_i$'s are deterministic rather than i.i.d. samples, $f_P$ and $g_P$ satisfy
\[f_P=\argmin_f\E_{P^{(0)},P^{(1)}}\hat\E_{\Dcal}\ell\big(I,f(\bm X)\big),\quad g_P=\argmin_g\E_{P^{(0)},P^{(1)}}\hat\E_{\Dcal}\big(Y-g(\bm X)\big)^2,\]
for some binary classification loss $\ell$, such as $\ell(I,f)=-I\log f-(1-I)\log(1-f)$. Therefore, they can be consistently estimated by general machine learning methods, 
\[\hat f=\argmin_{f\in\mathscr{F}}\hat\E_{\Dcal}\ell\big(I,f(\bm X)\big),\quad \hat g=\argmin_{g\in\mathscr{G}}\hat\E_{\Dcal}\big(Y-g(\bm X)\big)^2.\]
The following proposition shows that Algorithm \ref{alg:GCM} is a valid CTS test when the estimators $\hat f$ and $\hat g$ are sufficiently accurate. A proof is provided in Section \ref{sec:proof:prop:GCM}.
\begin{proposition}[Asymptotic Validity]\label{prop:GCM}
    Let $\Pcal$ be a class of distribution pairs $(P_{Y,\bm X}^{(0)},P_{Y,\bm X}^{(1)})$ such that the null hypothesis $H_0$ holds and
    \[N\hat\E_{\Dcal}\big(f_P-\hat f\big)^2\hat\E_{\Dcal}\big(g_P-\hat g\big)^2=o_{\Pcal}(1),\quad\hat\E_{\Dcal}v_Y\big(f_P-\hat f\big)^2=o_{\Pcal}(1),\]
    \[\sqrt{N}\big(\hat\E_{\Dcal}-\E\hat\E_{\Dcal}\big)\xi_I\big(g_P-\hat g\big)=o_{\Pcal}(1),\quad\hat\E_{\Dcal}\xi_I^2\big(g_P-\hat g\big)^2=o_{\Pcal}(1),\]
    \[\inf_{(P^{(0)},P^{(1)})\in\Pcal}\E_{P^{(0)},P^{(1)}}\hat\E_{\Dcal}\xi_I^2\xi_Y^2\gtrsim 1,\quad\sup_{(P^{(0)},P^{(1)})\in\Pcal}\E_{P^{(0)},P^{(1)}}\hat\E_{\Dcal}\big|\xi_I\xi_Y\big|^{2+\eta}\lesssim 1\quad\text{for some $\eta>0$},\]
    where $v_Y(\bm X)=\Var_{P^{(0)}}(Y\mid\bm X)$, $\xi_I=I-f_P(\bm X)$, and $\xi_Y=Y-g_P(\bm X)$. Then
    \[\lim_{n_0\rightarrow\infty,n_1\rightarrow\infty}\sup_{(P^{(0)},P^{(1)})\in\Pcal}\big|\E_{P^{(0)},P^{(1)}}\psi_{N,\alpha}-\alpha\big|=0.\]
\end{proposition}

\section{Algorithm in the Absence of Selection}

To show the necessity of accounting for selection effects, we modify Algorithm \ref{alg:conf_MAG_latent_selection} into the following Algorithm \ref{alg:conf_MAG_latent} by assuming no selection.

\begin{algorithm}
\caption{Lower Confidence Set for $\Ecal(\TCM(\Gcal_{\mid\bm C},\bm I))$ with Latent Variables in the Absence of Selection}
\label{alg:conf_MAG_latent}
\footnotesize
\begin{algorithmic}[1]
\Statex{\bf Input:} Observational data $\Dcal_0$, interventional data $\Dcal_k$, $k\in[d_X]$, confidence level $1-\alpha$.
\Statex{\bf Output:} Lower Confidence set $\Ccal_\alpha$ for $\Ecal(\TCM(\Gcal_{\mid\bm C}, \bm I))$.

\Statex \textbf{Stage I. Ancestral Set Recovery}
\State{\bf Step 1:} For $k\in[d_X]$, $j\in[d_X]\setminus\{k\}$, conduct CTS test $\psi_j^{\anterior(k)}\in\{0,1\}$ based on $\Dcal_0$ and $\Dcal_k$,
\[H_{j,0}^{\anterior(k)}:\Prob\big(p^{(k)}_{X_j\mid\bm C}\ne p_{X_j\mid\bm C}^{(0)}\big)=0\quad\longleftrightarrow \quad H_{j,1}^{\anterior(k)}:\Prob\big(p^{(k)}_{X_j\mid\bm C}\ne p_{X_j\mid\bm C}^{(0)}\big)>0,\]
with $\text{FWER}_1$ controlled below $\frac{\alpha}{2}$. 
\State{\bf Step 2:} For $k\in[d_X]$, $j\in[d_X]\setminus\{k\}$, 
\begin{itemize}
    \item if $\psi_j^{\anterior(k)}=\psi_k^{\anterior(j)}=1$, apply any rule to update $\psi_j^{\anterior(k)}, \psi_k^{\anterior(j)}$ such that $\psi_j^{\anterior(k)}+\psi_k^{\anterior(j)}\le 1$.
\end{itemize}
\State{\bf Step 3:} For $j\in[d_X]$, calculate the ancestral set for $X_j$,
\[\widehat{\bm A}_j=\bigg\{X_k:\exists m\ge 0,j_0=j,j_{m+1}=k, (j_0,\ldots,j_{m+1})\subseteq[d_X]\text{ s.t. }\prod_{l=0}^m\psi_{j_l}^{\anterior(j_{l+1})}=1\bigg\}\setminus\{X_j\}.\]

\Statex \textbf{Stage II. Adjacency and Orientation Recovery}
\State{\bf Step 4:} Initialize the detected edge set $\Ccal_\alpha=\emptyset$, and conduct the following parallel CI and CTS tests with $\text{FWER}_1$ controlled below $\frac{\alpha}{2}$, pretending $\widehat{\bm A}_{[d_X]}=\{\widehat{\bm A}_{j}:j\in[d_X]\}$ are fixed.
\State{\bf Step 5:} For $j,k\in[d_X]$, $j>k$, conduct CI test $\psi_{j,k}^{\rm Adj}\big(\widehat{\bm A}_{[d_X]}\big)\in\{0,1\}$ based on the pooled data $\Dcal_0\cup\big\{\Dcal_l:X_l\not\in\widehat{\bm A}_j\cup\widehat{\bm A}_k\cup\{X_j,X_k\}\big\}$: 
\begin{align*}
    &H_{j,k,0}^{\rm Adj}\big(\widehat{\bm A}_{[d_X]}\big):X_j\indep X_k\mid\bm C\cup\widehat{\bm A}_j\cup\widehat{\bm A}_k\setminus\{X_j,X_k\}\\
    \longleftrightarrow\quad &H_{j,k,1}^{\rm Adj}\big(\widehat{\bm A}_{[d_X]}\big):X_j\not\indep X_k\mid\bm C\cup \widehat{\bm A}_j\cup\widehat{\bm A}_k\setminus\{X_j,X_k\}.
\end{align*}
\begin{itemize}
    \item If $\psi_{j,k}^{\rm Adj}\big(\widehat{\bm A}_{[d_X]}\big)=1$, $X_j\in \widehat{\bm A}_k$, and $X_k\not\in\widehat{\bm A}_j$, append $X_j\rightarrow X_k$ to $\Ccal_\alpha$. 
    \item If $\psi_{j,k}^{\rm Adj}\big(\widehat{\bm A}_{[d_X]}\big)=1$, $X_j\not\in \widehat{\bm A}_k$, and $X_k\in\widehat{\bm A}_j$, append $X_j\leftarrow X_k$ to $\Ccal_\alpha$. 
    \item If $\psi_{j,k}^{\rm Adj}\big(\widehat{\bm A}_{[d_X]}\big)=1$, $X_j\not\in \widehat{\bm A}_k$, and $X_k\not\in\widehat{\bm A}_j$, append $X_j\leftrightarrow X_k$ to $\Ccal_\alpha$. 
\end{itemize}
\State{\bf Step 6: } For $j\in[d_X]$, $X_k\in\widehat{\bm A}_j$, conduct CTS test $\psi_j^{{\rm Adj}(k)}\big(\widehat{\bm A}_{[d_X]}\big)\in\{0,1\}$ based on $\Dcal_k$ and the pooled data $\Dcal_0\cup\{\Dcal_l:X_l\not\in\widehat{\bm A}_j\cup\{X_j,X_k\}\}$
\begin{align*}
    H_{j,0}^{{\rm Adj}(k)}\big(\widehat{\bm A}_{[d_X]}\big): \Prob\big(p_{X_j\mid\bm C\cup\widehat{\bm A}_j}^{(k)}\ne p_{X_j\mid\bm C\cup\widehat{\bm A}_j}^{(0)}\big)=0\quad
    \longleftrightarrow \quad  H_{j,1}^{{\rm Adj}(k)}\big(\widehat{\bm A}_{[d_X]}\big): \Prob\big(p_{X_j\mid\bm C\cup\widehat{\bm A}_j}^{(k)}\ne p_{X_j\mid\bm C\cup\widehat{\bm A}_j}^{(0)}\big)>0.
\end{align*}
\begin{itemize}
    \item If $\psi_j^{{\rm Adj}(k)}\big(\widehat{\bm A}_{[d_X]}\big)=1$, append $I_k\rightarrow X_j$ to $\Ccal_\alpha$.
\end{itemize}

\end{algorithmic}
\end{algorithm}

\section{Additional Numerical Results}

This section provides model (M1), the true $\Gcal_{\mid\bm C}$ and $\TCM(\Gcal_{\mid\bm C},\bm I)$ for (M1) in Figures \ref{fig:true_M1_22}, \ref{fig:true_M1_20}, \ref{fig:true_M1_02} and for (M2) in Figures \ref{fig:true_M2_22}, \ref{fig:true_M2_20}, \ref{fig:true_M2_02}. We also present the $\widehat{\bm A}_j$ sets in the real data analysis in Figure \ref{fig:A549_Aj} and all 200 pairs in $\widehat{\bm E}_{\bm L}$ in Figure \ref{fig:A549_latent_pairs}.

\begin{enumerate}[(M1)]
    \item $\bm Z=(\bm C,\bm X,\bm L,\bm S)$, $\bm C\sim N(\bm 0,I_2)$, $\bm X\in\R^{20}$, $\bm L^{(k)}\sim N(\bm 0,I_{d_L})$, $\bm S\in\{0,1\}^{d_S}$, $\epsilon_j^{(k)}\sim N(0,0.25)$,
    \begin{eqnarray*}X_j^{(k)}=\bm \gamma_j^\top\bm C^{(k)} &+ &\sum_{Z_l^{(k)}\in\parent_{\Gcal_{\mid\bm C}}(X_j^{(k)})}\beta_{j,l}Z_l^{(k)}\\
    &+&\sum_{\{Z_l^{(k)},Z_m^{(k)}\}\subset\parent_{\Gcal_{\mid\bm C}}(X_j^{(k)})}0.2\beta_{j,l}\beta_{j,m}Z_l^{(k)}Z_m^{(k)}+1.5\1(j=k)+\epsilon_j^{(k)},
    \end{eqnarray*}
    \[S_i=\1\bigg(\sum_{Z_l\in\parent_{\Gcal}(S_i)}a_{i,l}Z_l>\tau\bigg),\quad\beta_{j,l},a_{i,l},\gamma_{j,1},\gamma_{j,2}\in\{0.5,-0.5\},\]
   where  $\tau$ is the largest real number such that $\Prob_{P^{(k)}}(\bm S=\bm 1) \ge 0.1$ for all $k=0,\ldots,20$.
\end{enumerate}

\begin{figure}
    \centering
    \begin{subfigure}{0.45\textwidth}
        \centering
        \includegraphics[width=\linewidth]{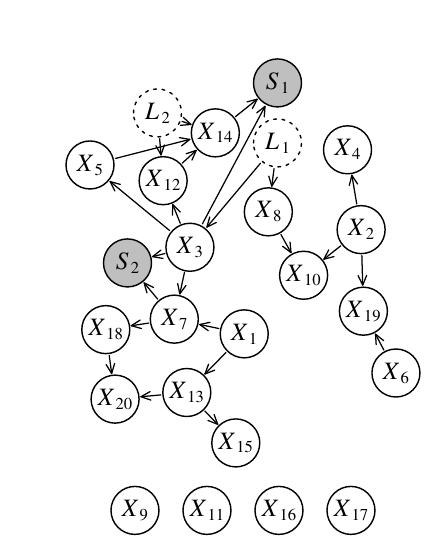}
        \caption{$\Gcal_{\mid\bm C}$.}
    \end{subfigure}
    \begin{subfigure}{0.45\textwidth}
        \centering
        \includegraphics[width=\linewidth]{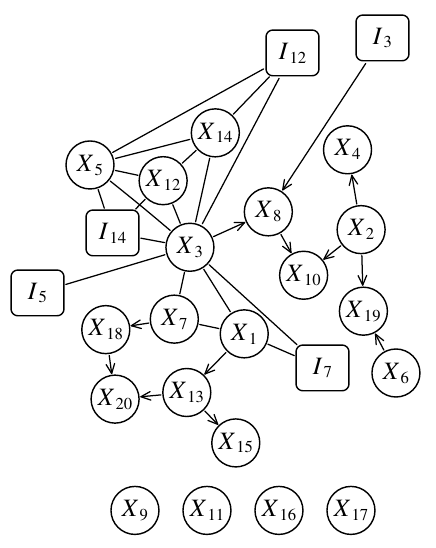}
        \caption{$\TCM(\Gcal_{\mid\bm C},\bm I)$}
    \end{subfigure}
    \caption{True graphs $\Gcal_{\mid\bm C}$ (a) and $\TCM(\Gcal_{\mid\bm C},\bm I)$ (b) for (M1) under $d_S=2,d_L=2$.}
    \label{fig:true_M1_22}
\end{figure}

\begin{figure}
    \centering
    \begin{subfigure}{0.45\textwidth}
        \centering
        \includegraphics[width=\linewidth]{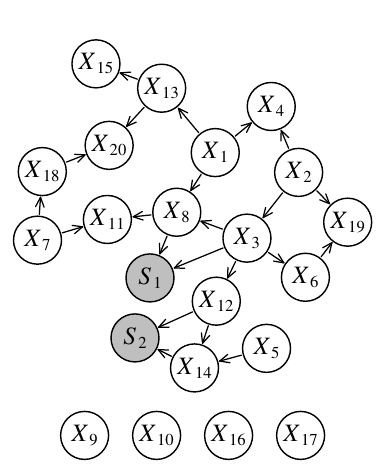}
        \caption{$\Gcal_{\mid\bm C}$.}
    \end{subfigure}
    \begin{subfigure}{0.45\textwidth}
        \centering
        \includegraphics[width=\linewidth]{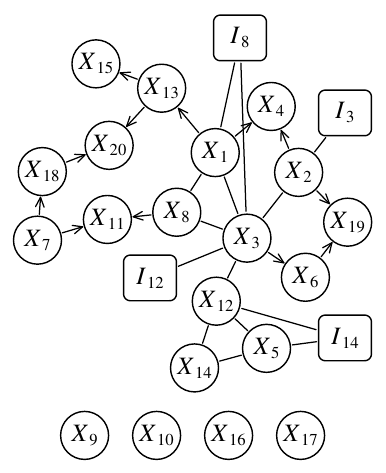}
        \caption{$\TCM(\Gcal_{\mid\bm C},\bm I)$}
    \end{subfigure}
    \caption{True graphs $\Gcal_{\mid\bm C}$ (a) and $\TCM(\Gcal_{\mid\bm C},\bm I)$ (b) for (M1) under $d_S=2,d_L=0$.}
    \label{fig:true_M1_20}
\end{figure}

\begin{figure}
    \centering
    \begin{subfigure}{0.45\textwidth}
        \centering
        \includegraphics[width=\linewidth]{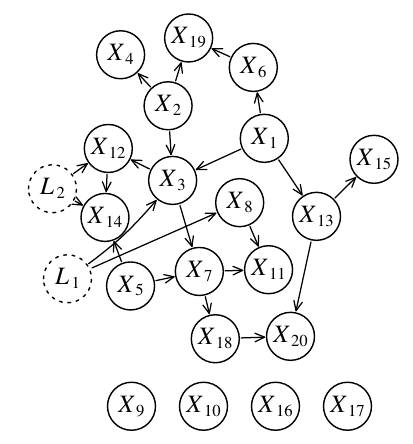}
        \caption{$\Gcal_{\mid\bm C}$.}
    \end{subfigure}
    \begin{subfigure}{0.45\textwidth}
        \centering
        \includegraphics[width=\linewidth]{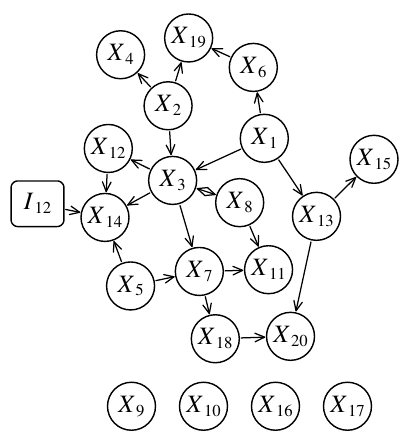}
        \caption{$\TCM(\Gcal_{\mid\bm C},\bm I)$}
    \end{subfigure}
    \caption{True graphs $\Gcal_{\mid\bm C}$ (a) and $\TCM(\Gcal_{\mid\bm C},\bm I)$ (b) for (M1) under $d_S=0,d_L=2$.}
    \label{fig:true_M1_02}
\end{figure}

\begin{figure}
    \centering
    \begin{subfigure}{0.45\textwidth}
        \centering
        \includegraphics[width=\linewidth]{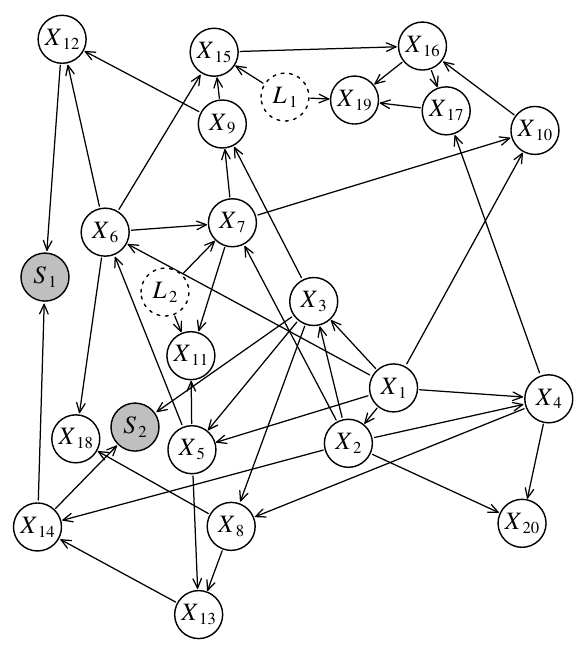}
        \caption{$\Gcal_{\mid\bm C}$.}
    \end{subfigure}
    \begin{subfigure}{0.45\textwidth}
        \centering
        \includegraphics[width=\linewidth]{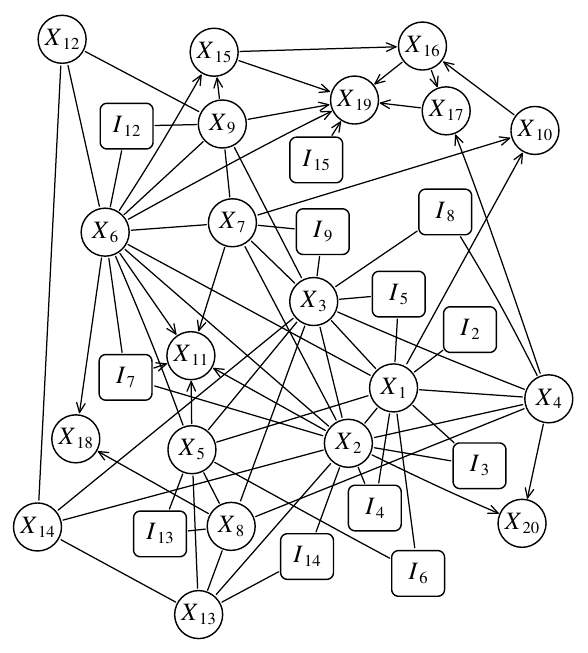}
        \caption{$\TCM(\Gcal_{\mid\bm C},\bm I)$}
    \end{subfigure}
    \caption{True graphs $\Gcal_{\mid\bm C}$ (a) and $\TCM(\Gcal_{\mid\bm C},\bm I)$ (b) for (M2) under $d_S=2,d_L=2$.}
    \label{fig:true_M2_22}
\end{figure}

\begin{figure}
    \centering
    \begin{subfigure}{0.45\textwidth}
        \centering
        \includegraphics[width=\linewidth]{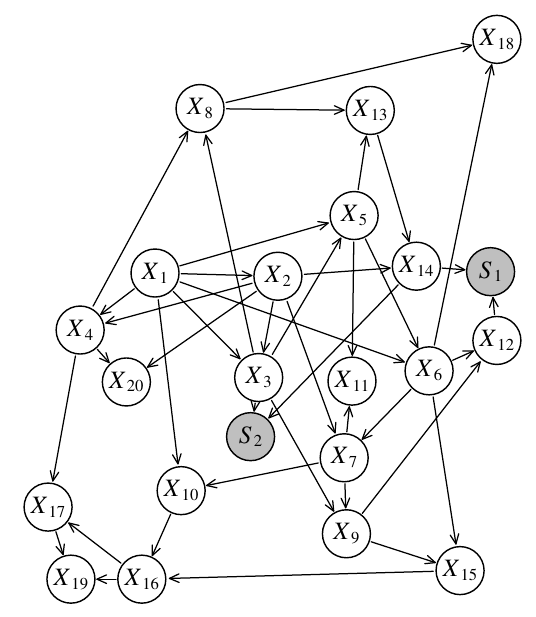}
        \caption{$\Gcal_{\mid\bm C}$.}
    \end{subfigure}
    \begin{subfigure}{0.45\textwidth}
        \centering
        \includegraphics[width=\linewidth]{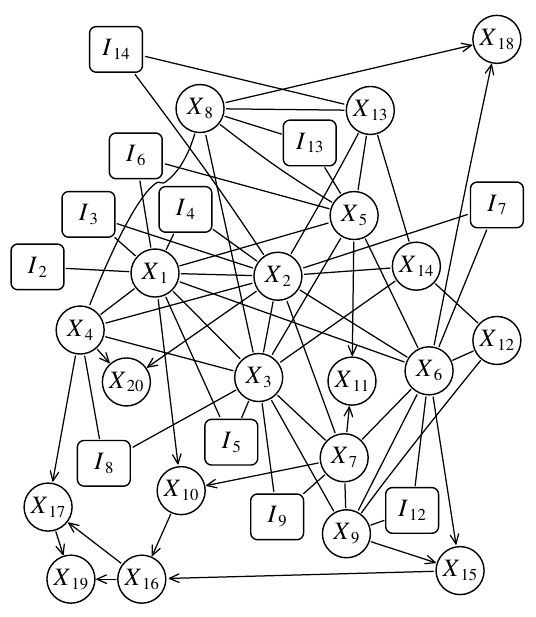}
        \caption{$\TCM(\Gcal_{\mid\bm C},\bm I)$}
    \end{subfigure}
    \caption{True graphs $\Gcal_{\mid\bm C}$ (a) and $\TCM(\Gcal_{\mid\bm C},\bm I)$ (b) for (M2) under $d_S=2,d_L=0$.}
    \label{fig:true_M2_20}
\end{figure}

\begin{figure}
    \centering
    \begin{subfigure}{0.45\textwidth}
        \centering
        \includegraphics[width=\linewidth]{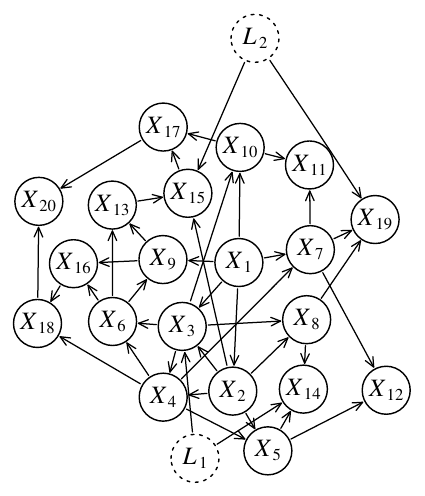}
        \caption{$\Gcal_{\mid\bm C}$.}
    \end{subfigure}
    \begin{subfigure}{0.45\textwidth}
        \centering
        \includegraphics[width=\linewidth]{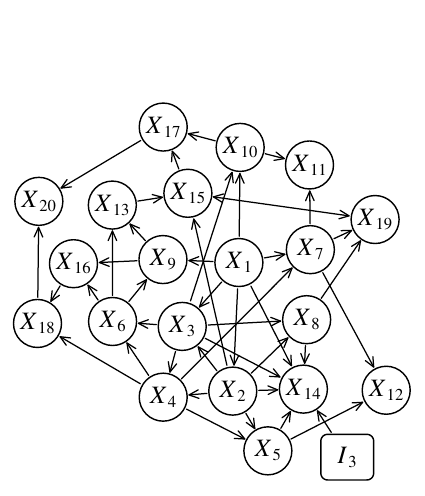}
        \caption{$\TCM(\Gcal_{\mid\bm C},\bm I)$}
    \end{subfigure}
    \caption{True graphs $\Gcal_{\mid\bm C}$ (a) and $\TCM(\Gcal_{\mid\bm C},\bm I)$ (b) for (M2) under $d_S=0,d_L=2$.}
    \label{fig:true_M2_02}
\end{figure}

\begin{figure}[!htbp]
    \centering
    \includegraphics[width=\linewidth]{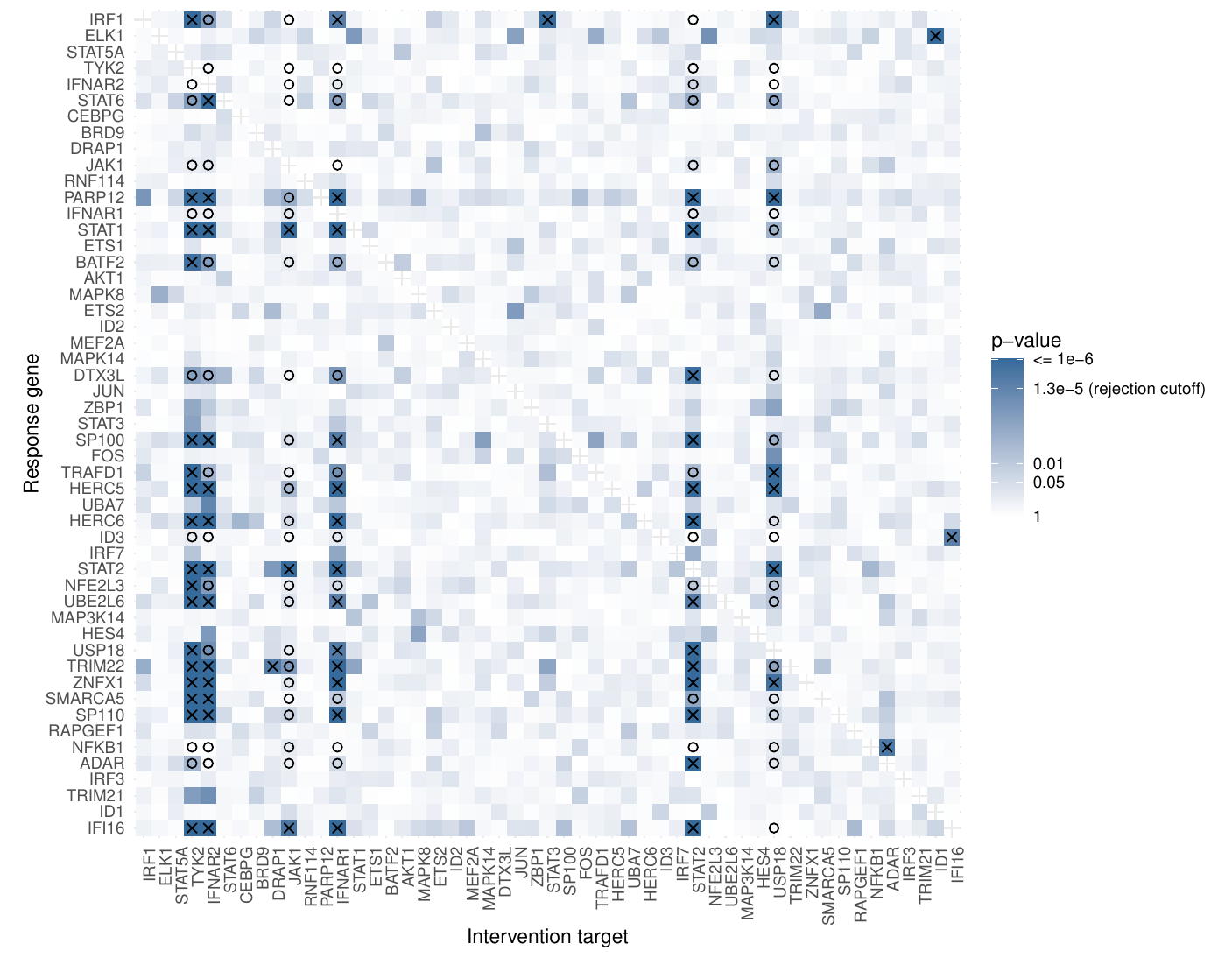}
    \caption{\footnotesize Column $k$ indicates the distribution change of other genes when $X_j$ was intervened. Row $j$ represents the detected $\widehat{\bm A}_j$. $\times$ indicates a significant distribution shift. $\circ$ represents additional anteriority obtained by Step 3 of Algorithm \ref{alg:conf_MAG_latent_selection}. The color reflects the $p$-values of CTS tests.}
    \label{fig:A549_Aj}
\end{figure}

\begin{figure}
    \centering
    \includegraphics[width=\linewidth]{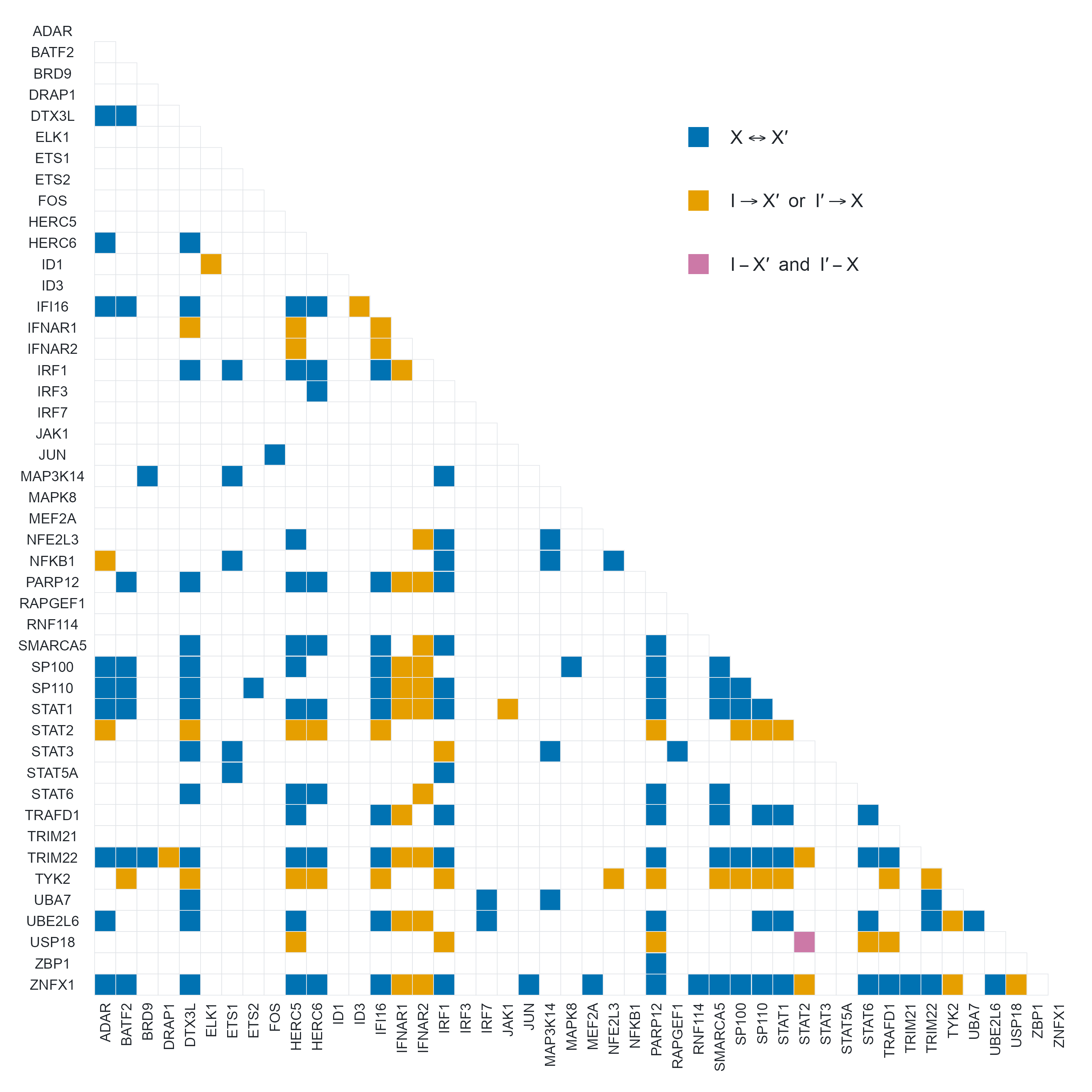}
    \caption{All 200 pairs in $\widehat{\bm E}_{\bm L}$ in real data analysis.}
    \label{fig:A549_latent_pairs}
\end{figure}

\section{Proofs}\label{sec:proof}

\subsection{Proof of Proposition \ref{prop:context_conditional}}\label{sec:proof:prop:context_conditional}

\begin{proof}[Proof of Proposition \ref{prop:context_conditional}]
    Since part 1) is a special case of part 2) with $K=0$, we only consider part 2). Then it suffices to prove the case where $\bm Z_1,\bm Z_2$ are one-dimensional random variables $Z_1,Z_2$. We separate the proof into two directions.
    
    \noindent\textbf{a) $(Z_1\not\indep Z_2\mid\bm Z_3)_{\Aug(\Gcal, I_k)_{\mid\bm C}}\Longrightarrow (Z_1\not\indep Z_2\mid\bm C, \bm Z_3)_{\Aug(\Gcal, I_k)}$.}

    When the left-hand side holds, there must exist an open path connecting $Z_1,Z_2$ in $\Aug(\Gcal, I_k)_{\mid\bm C}$ given $\bm Z_3$. Then non-colliders on the path are not in $\bm Z_3$ and colliders on the path must be ancestors of $\bm Z_3$. Since $\Aug(\Gcal, I_k)_{\mid\bm C}$ is an induced subgraph of $\Aug(\Gcal, I_k)$, the same path also exists in $\Aug(\Gcal, I_k)$, the status of collider or non-collider on the path remains the same, and an ancestor in $\Aug(\Gcal, I_k)_{\mid\bm C}$ is also an ancestor in $\Aug(\Gcal, I_k)$. Therefore, in $\Aug(\Gcal, I_k)$, the non-colliders on the path are also not in $(\bm Z_3,\bm C)$ and colliders on the path are ancestors of $\bm Z_3$ as well. Consequently, the path is open in $\Aug(\Gcal, I_k)$ and $(Z_1\not\indep Z_2\mid\bm C, \bm Z_3)_{\Aug(\Gcal, I_k)}$.
    
    \noindent\textbf{b) $(Z_1\not\indep Z_2\mid\bm C,\bm Z_3)_{\Aug(\Gcal, I_k)}\Longrightarrow (Z_1\not\indep Z_2\mid\bm Z_3)_{\Aug(\Gcal, I_k)_{\mid\bm C}}$.}

    When the left-hand side holds, there must exist an open path $\pi$ connecting $Z_1,Z_2$ in $\Aug(\Gcal, I_k)$ given $\bm C,\bm Z_3$. Then we make the following claim.
    \begin{claim}\label{clm:no_context}
        No middle node of $\pi$ is in $\ancestor_{\Aug(\Gcal, I_k)}(\bm C)$.
    \end{claim}
    By Claim \ref{clm:no_context}, no node in $\bm C$ is on $\pi$. Therefore, removing edges involving $\bm C$ will not affect $\pi$, and $\pi$ is still a path connecting $Z_1,Z_2$ in $\Aug(\Gcal, I_k)_{\mid\bm C}$. Since $\pi$ is open in $\Aug(\Gcal, I_k)$, non-colliders on $\pi$ are not in $(\bm C,\bm Z_3)$, and colliders are ancestors of $(\bm C,\bm Z_3)$. Following Claim \ref{clm:no_context}, colliders are in $\ancestor_{\Aug(\Gcal,I_k)}(\bm Z_3)\setminus\ancestor_{\Aug(\Gcal,I_k)}(\bm C)$. Therefore, for any collider $Z'$ on $\pi$, there is a path $Z'\rightarrow\ldots\rightarrow \tilde Z$ for some $\tilde Z\in\bm Z_3$ with no element of $\bm C$ on the path, therefore the path remains in $\Aug(\Gcal, I_k)_{\mid\bm C}$. Consequently, non-colliders on $\pi$ are not in $\bm Z_3$ and colliders are ancestors of $\bm Z_3$ in $\Aug(\Gcal, I_k)_{\mid\bm C}$. So $\pi$ is an open path connecting $Z_1,Z_2$ in $\Aug(\Gcal, I_k)_{\mid\bm C}$ given $\bm Z_3$, and we conclude the proof.
\end{proof}

\begin{proof}[Proof of Claim \ref{clm:no_context}]
    According to Assumption \ref{asm:context} and the definition of $\bm S$, we know $\bm X$ and $\bm S$ are not ancestors of $\bm C$, and thus $\ancestor_{\Aug(\Gcal, I_k)}(\bm C)\subseteq(\bm C,\bm L, I_k)$. Now we prove by contradiction, suppose $Z$ is the middle node in $\pi\cap\ancestor_{\Aug(\Gcal, I_k)}(\bm C)$ closest to $Z_1$.

    \noindent\textbf{1)} If $Z\in\bm C$, for $\pi$ to be open given $\bm C,\bm Z_3$, $Z$ must be a collider on $\pi$. Then there is no node between $Z_1$ and $Z$ on $\pi$. Otherwise, if there are nodes between them, the node $Z'$ closest to $Z$ must be a parent of $Z$ since $Z$ is a collider. Then $Z'\in \ancestor_{\Aug(\Gcal, I_k)}(\bm C)$ and is closer to $Z_1$ than $Z$, contradicting the definition of $Z$. 
    
    Therefore, $Z_1$ is a parent of $Z$, and thus $Z_1\in \ancestor_{\Aug(\Gcal, I_k)}(\bm C)$. Since $Z_1\in \bm X$, this contradicts Assumption \ref{asm:context}.

    \noindent\textbf{2)} If $Z\in\bm L$, the node $Z'$ between $Z_1,Z$ closest to $Z$ should have edge $Z'\leftarrow Z$. Otherwise, $Z'$ is in $\ancestor_{\Aug(\Gcal, I_k)}(\bm C)$. If $Z'\ne Z_1$, $Z'$ closer to $Z_1$ than $Z$ which contradicts the definition of $Z$, if $Z'=Z_1$, then $Z_1\in\bm X\cap\ancestor_{\Aug(\Gcal, I_k)}(\bm C)$ which contradicts Assumption \ref{asm:context}.

    Then $Z'$ must be in $\bm L$ as well. Otherwise, if $Z'\in\bm C$, it contradicts either the closeness property of $Z$ or Assumption \ref{asm:context}. If $Z'\in(\bm X,\bm S)$, since $Z\in\ancestor_{\Aug(\Gcal,I_k)}(\bm C)$ and $\ancestor_{\Aug(\Gcal,I_k)}(\bm C)\subset(\bm C,\bm L,I_k)$, there must be a directed path $Z\rightarrow\ldots\rightarrow C$ for some $C\in\bm C$ with all middle nodes in $(\bm L,I_k)$. Since $I_k$ is not adjacent to $\bm L$ by definition, all middle nodes on $Z\rightarrow\ldots\rightarrow C$ are in $\bm L$. Then $\bm X,\bm C$ or $\bm S,\bm C$ must share unmeasured confounders, contradicting Assumption \ref{asm:context}. If $Z'=I_k$, $I_k$ has a latent variable $Z$ as a parent, which contradicts the definition.

    If $Z'$ is not a collider on $\pi$, the node $Z''$ between $Z_1,Z'$ closest to $Z'$ must have edge $Z'' \leftarrow Z'$. Using the same argument as above, $Z''$ must be in $\bm L$ as well. By induction, for the collider $\tilde Z$ between $Z_1,Z'$ closest to $Z'$, the subpath between $\tilde Z,Z'$ has the form $\tilde Z\leftarrow\ldots\leftarrow Z'$ and consists of only latent variables in $\bm L$. Note that the collider must exist, since otherwise $Z_1$ and $\bm C$ share unmeasured confounders. 

    For $\pi$ to be open, the collider $\tilde Z$ must be in $\ancestor_{\Aug(\Gcal, I_k)}(\bm Z_3)\setminus\ancestor_{\Aug(\Gcal, I_k)}(\bm C)$ by the closeness property of $Z$. Since $\bm Z_3\subseteq(\bm X,\bm S, I_k)$, there must exist a path $\tilde Z\rightarrow\ldots\rightarrow \hat Z$ for some $\hat Z\in(\bm X,\bm S)$ such that middle nodes on the path belong to $\bm L$. To see this, $\bm C$ can not be the middle node since $\tilde Z\not\in\ancestor_{\Aug(\Gcal, I_k)}(\bm C)$; we can always choose $\hat Z$ to be the first node in $(\bm X,\bm S)$ that appears on the path, so $(\bm X,\bm S)$ are not the middle node; $I_k$ can not be the middle node since $\bm L$ can not be parents of $I_k$ by definition. 

    Then, $\hat Z\in(\bm X,\bm S)$ shares unmeasured confounders with $\bm C$, contradicting Assumption \ref{asm:context}.

    \noindent\textbf{3)} If $Z= I_k$, since $ I_k\in(\bm Z_2,\bm Z_3)$, $I_k$ is either an end point or conditioned, then $Z$ must be a collider. Let $Z'$ be the node between $Z_1,Z$ closest to $Z$, it must happen $Z'\rightarrow Z$, so $Z'\in\ancestor_{\Aug(\Gcal, I_k)}(\bm C)$. If $Z'\ne Z_1$, it contradicts the closeness property of $Z$, if $Z'=Z_1$, it contradicts Assumption \ref{asm:context}.

    Combining pieces concludes the proof.
\end{proof}

\subsection{Proof of Lemma \ref{lem:Markov_Aug}}

\begin{proof}[Proof of Lemma \ref{lem:Markov_Aug}]
    Denote
    \[\widetilde p_{\bm Z}^{(k)}=\prod_{C\in\bm C_2}p_{C\mid\parent_{\Gcal}(C)}^{(k)}p_{X_k\mid\parent_{\Gcal}(X_k)}^{(k)}\prod_{Z\in\bm Z\setminus(\bm C_2, X_k)}p_{Z\mid\parent_{\Gcal}(Z)}^{(0)},\]
    we know 
    \[p_{\bm Z}^{(k)}=\frac{\Prob(\bm I=e_k\mid\bm C_1)}{\Prob(\bm I=e_k)}\widetilde p_{\bm Z}^{(k)},\quad p_{\bm Z\mid\bm C}^{(k)}=\widetilde p_{\bm Z\mid\bm C}^{(k)}.\]
    We introduce an auxiliary random variable $J_k\in\{0,1\}$ with $\Prob(J_k=1)=\frac{1}{2}$ and define a joint distribution $Q$ of $(\bm Z,J_k)$ as 
    \[Q_{\bm Z\mid J_k=1}=\widetilde P_{\bm Z}^{(k)},\quad Q_{\bm Z\mid J_k=0}=P_{\bm Z}^{(0)}.\]
    Let $\Aug(\Gcal,J_k)$ be a graph obtained by relabeling the $I_k$ in $\Aug(\Gcal,I_k)$ as $J_k$. Then it is straightforward that $Q$ is Markov to $\Aug(\Gcal,J_k)$.

    By Proposition \ref{prop:context_conditional}, $(J_k\indep \bm X_1\mid\bm X_2,\bm S)_{\Aug(\Gcal_{\mid\bm C},J_k)}$ if and only if $(J_k\indep \bm X_1\mid\bm C,\bm X_2,\bm S)_{\Aug(\Gcal,J_k)}$. Since $Q$ is Markov to $\Aug(\Gcal,J_k)$, we have
    \begin{align*}
        &(J_k\indep \bm X_1\mid\bm X_2,\bm S)_{\Aug(\Gcal_{\mid\bm C},J_k)}\Leftrightarrow(J_k\indep \bm X_1\mid\bm C,\bm X_2,\bm S)_{\Aug(\Gcal,J_k)}\\
        \Rightarrow&J_k\indep_Q \bm X_1\mid\bm C,\bm X_2,\bm S\Rightarrow\widetilde P_{\bm X_1\mid\bm C,\bm X_2,\bm S}^{(k)}=P_{\bm X_1\mid\bm C,\bm X_2,\bm S}^{(0)}\Rightarrow P_{\bm X_1\mid\bm C,\bm X_2,\bm S}^{(k)}=P_{\bm X_1\mid\bm C,\bm X_2,\bm S}^{(0)}.
    \end{align*}
\end{proof}

\subsection{Proof of Proposition \ref{prop:IMEC}}\label{sec:proof:prop:IMEC}

\begin{lemma}\label{lem:add_I}
    For any disjoint sets $\bm X_1,\bm X_2,\bm X_3\subseteq\bm X$ and $k\in[K]$, $(\bm X_1\indep\bm X_2\mid\bm X_3,\bm S)_{\Gcal_{\mid\bm C}}$ if and only if $(\bm X_1\indep\bm X_2\mid\bm X_3,\bm S)_{\Aug(\Gcal_{\mid\bm C}, I_k)}$ if and only if $(\bm X_1\indep\bm X_2\mid\bm X_3,\bm S, I_k)_{\Aug(\Gcal_{\mid\bm C}, I_k)}$.
\end{lemma}

\begin{proof}[Proof of Lemma \ref{lem:add_I}]
    Since $I_k$ is a source node and has only one adjacent node $X_k$ in $\Aug(\Gcal_{\mid\bm C},I_k)$, we know $I_k$ can not be a middle node on any open path in $\Aug(\Gcal_{\mid\bm C},I_k)$. 

    \noindent\textbf{a)} $(\bm X_1\not\indep\bm X_2\mid\bm X_3,\bm S)_{\Gcal_{\mid\bm C}}\Longrightarrow (\bm X_1\not\indep\bm X_2\mid\bm X_3,\bm S)_{\Aug(\Gcal_{\mid\bm C}, I_k)}$.

    There exists an open path $\pi$ given $(\bm X_3, \bm S)$ connecting $\bm X_1,\bm X_2$ in $\Gcal_{\mid\bm C}$. No non-collider on $\pi$ is in $(\bm X_3,\bm S)$ and all colliders are in $\ancestor_{\Gcal_{\mid\bm C}}(\bm X_3,\bm S)$. Since $\Gcal_{\mid\bm C}$ is a subgraph of $\Aug(\Gcal_{\mid\bm C}, I_k)$, we know $\pi$ also exists in $\Aug(\Gcal_{\mid\bm C}, I_k)$, collider status on $\pi$ remains in $\Aug(\Gcal_{\mid\bm C}, I_k)$ and ancestors in $\Gcal_{\mid\bm C}$ remains in the augmented graph. Therefore, $\pi$ is also open in $\Aug(\Gcal_{\mid\bm C}, I_k)$.

    \noindent\textbf{b)} $(\bm X_1\not\indep\bm X_2\mid\bm X_3,\bm S)_{\Aug(\Gcal_{\mid\bm C}, I_k)}\Longrightarrow (\bm X_1\not\indep\bm X_2\mid\bm X_3,\bm S, I_k)_{\Aug(\Gcal_{\mid\bm C}, I_k)}$.

    There must exist an open path $\pi$ connecting $\bm X_1,\bm X_2$ given $(\bm X_3,\bm S)$ in $\Aug(\Gcal_{\mid\bm C}, I_k)$. Since $I_k$ cannot be a middle node on $\pi$, $\pi$ remains open give $(\bm X_3,\bm S, I_k)$.

    \noindent\textbf{c)} $(\bm X_1\not\indep\bm X_2\mid\bm X_3,\bm S, I_k)_{\Aug(\Gcal_{\mid\bm C}, I_k)}\Longrightarrow (\bm X_1\not\indep\bm X_2\mid\bm X_3,\bm S)_{\Gcal_{\mid\bm C}}$.

    There must exist an open path $\pi$ connecting $\bm X_1,\bm X_2$ given $(\bm X_3,\bm S, I_k)$ in $\Aug(\Gcal_{\mid\bm C}, I_k)$. Non-colliders on $\pi$ are not in $(\bm X_3, \bm S, I_k)$ and all colliders are in $\ancestor_{\Aug(\Gcal_{\mid\bm C}, I_k)}(\bm X_3, \bm S, I_k)$. Since $I_k$ is not on $\pi$ and has no ancestors, we know all colliders are in $\ancestor_{\Aug(\Gcal_{\mid\bm C}, I_k)}(\bm X_3, \bm S)$. For any collider $Z$ on $\pi$, there must exist a path $\pi_Z$ (possibly with length zero) $Z\rightarrow \ldots\rightarrow Z'$ for some $Z'\in(\bm X_3,\bm S)$. Clearly $I_k$ is not on $\pi_Z$ as well. Then the paths $\pi$ and $\pi_Z$ remain in $\Gcal_{\mid\bm C}$ and $\pi$ is still open given $(\bm X_3,\bm S)$ in $\Gcal_{\mid\bm C}$.

\end{proof}

\begin{proof}[Proof of Proposition \ref{prop:IMEC}]

    By Lemma \ref{lem:add_I}, we have part 1) in Definition \ref{def:IMEC} is equivalent to $(\bm X_1\indep \bm X_2\mid \bm X_3,S')_{\Aug(\Gcal_{\mid\bm C}',I_k)}$ if and only if $(\bm X_1\indep \bm X_2\mid \bm X_3,S)_{\Aug(\Gcal_{\mid\bm C},I_k)}$, and it is also equivalent to $(\bm X_1\indep \bm X_2\mid \bm X_3,S',I_k)_{\Aug(\Gcal_{\mid\bm C}',I_k)}$ if and only if $(\bm X_1\indep \bm X_2\mid \bm X_3,S,I_k)_{\Aug(\Gcal_{\mid\bm C},I_k)}$. 
    
    Then, we only need to show that the $d$-separation equivalence between $\Aug(\Gcal_{\mid\bm C}, I_k)$ and $\Aug(\Gcal_{\mid\bm C}', I_k)$ is equivalent to the MAG Markov equivalence between $\MAG(\Aug(\Gcal_{\mid\bm C}, I_k))$ and $\MAG(\Aug(\Gcal_{\mid\bm C}', I_k))$, which follows from Theorem 4.18 in \cite{richardson2002ancestral} and Theorem 1 in \cite{spirtes1997polynomial}. Therefore, we complete the proof.
\end{proof}

\subsection{Proof of Lemma \ref{lem:CM}}\label{sec:proof:lem:CM}

\begin{proof}[Proof of Lemma \ref{lem:CM}]
    By Theorem 4.2 in \cite{richardson2002ancestral}, $X_i, X_j$ are adjacent in $\MAG(\Aug(\Gcal_{\mid\bm C}, I_k))$ if and only if $\big(X_i\not\indep X_j\mid \ancestor_{\Aug(\Gcal_{\mid\bm C}, I_k)}(X_i,X_j,\bm S)\setminus(X_i,X_j,\bm L)\big)_{\Aug(\Gcal_{\mid\bm C},I_k)}$. By Lemma \ref{lem:add_I}, it is also equivalent to $\big(X_i\not\indep X_j\mid \ancestor_{\Gcal_{\mid\bm C}}(X_i,X_j,\bm S)\setminus(X_i,X_j,\bm L)\big)_{\Gcal_{\mid\bm C}}$, which is equivalent to the adjacency of $X_i,X_j$ in $\MAG(\Gcal_{\mid\bm C})$. Therefore, the adjacency among $\bm X$ is the same for all the $K$ MAGs, which equals the adjacency in $\MAG(\Gcal_{\mid\bm C})$. 
    
    Note that $I_k$ is only adjacent to $X_k$ in $\Aug(\Gcal_{\mid\bm C}, I_k)$, therefore it cannot be a middle node on any open path in $\Aug(\Gcal_{\mid\bm C}, I_k)$. Consequently, $\ancestor_{\Aug(\Gcal_{\mid\bm C}, I_k)}(X_j, \bm S)\cap\bm X=\ancestor_{\Gcal_{\mid\bm C}}(X_j, \bm S)\cap\bm X$ and doesn't depend on $k$. By Definition \ref{def:MAG}, the edge types among $\bm X$ are the same for all $K$ MAGs, and thus equal to the edge types in $\MAG(\Gcal_{\mid\bm C})$.
\end{proof}

\subsection{Proof of Theorem \ref{thm:unique}}\label{sec:proof:thm:unique}

\begin{proof}[Proof of Theorem \ref{thm:unique}]
    \noindent\textbf{Sufficiency:} When $K=d_X$, Propositions \ref{prop:adjacency_latent_selection} and \ref{prop:direction_latent_selection} show that the adjacency and edge types can be uniquely determined through CI relations in the augmented graphs. Therefore, $\TCM(\Gcal,\bm I)$ is uniquely identified.

    \noindent\textbf{Necessity:} We only need to find a DAG $\Gcal$, and for each $K<d_X$ intervention targets, construct a DAG $\Gcal'$  such that $\TCM(\Gcal',\bm I)\in\Mcal(\TCM(\Gcal,\bm I))$ and $\TCM(\Gcal',\bm I)\ne \TCM(\Gcal,\bm I)$. To this end, define $\Gcal$ to be the DAG on $(\bm X, S)$ with $d_X$ edges $\{X_i\rightarrow S: i\in[d_X]\}$. Since all the $d_X$ observed nodes $\bm X$ are symmetric in $\Gcal$, without the loss of generality, we assume the $K<d_X$ intervention targets are $X_1,\ldots X_K$. Then $\CM(\Gcal,\bm I)$ has edges $\{I_i-X_i:i\in[K]\}\cup\{X_i-X_j:i<j\in[d_X]\}$ with no unshielded collider and discriminating path. Define another DAG $\Gcal'$ on $(\bm X,S)$ with edges $\{X_i\rightarrow S: i\in[d_X-1]\}\cup\{X_i\rightarrow X_{d_X}:i\in[d_X-1]\}$. Then $\CM(\Gcal',\bm I)$ has edges $\{I_i-X_i:i\in[K]\}\cup\{X_i-X_j:i<j\in[d_X-1]\}\cup\{X_i\rightarrow X_{d_X}:i\in[d_X-1]\}$. See Figure \ref{fig:necessary} for illustration. It is evident that $\MAG(\Gcal,I_j),\MAG(\Gcal',I_j)$ share the same skeleton and have no unshielded collider and discriminating path. Then $\TCM(\Gcal',\bm I)\in\Mcal(\TCM(\Gcal,\bm I))$ but $\TCM(\Gcal',\bm I)\ne \TCM(\Gcal,\bm I)$.
    \begin{figure}
        \centering
        \begin{subfigure}{0.2\textwidth}
            \centering
            \includegraphics[width=0.9\linewidth]{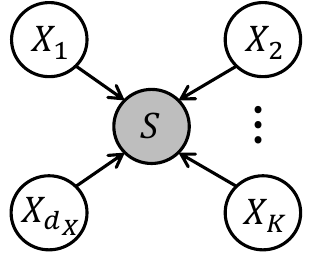}
            \caption{$\Gcal$}
        \end{subfigure}
        \hspace{2pt}
        \begin{subfigure}{0.2\textwidth}
            \centering
            \includegraphics[width=0.9\linewidth]{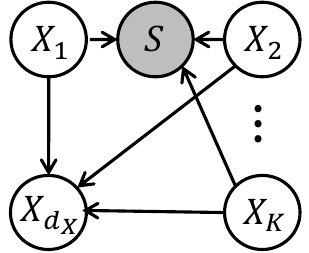}
            \caption{$\Gcal'$}
        \end{subfigure}
        \hspace{3pt}
        \begin{subfigure}{0.25\textwidth}
            \centering
            \includegraphics[width=\linewidth]{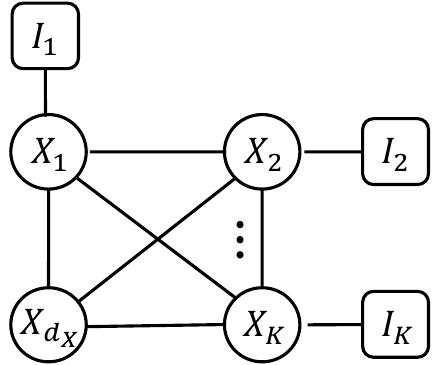}
            \caption{$\CM(\Gcal,\bm I)$}
        \end{subfigure}
        \hspace{14pt}
        \begin{subfigure}{0.25\textwidth}
            \centering
            \includegraphics[width=\linewidth]{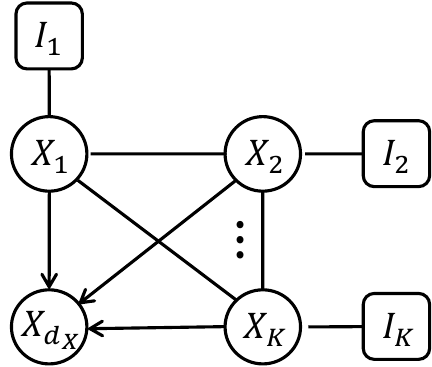}
            \caption{$\CM(\Gcal',\bm I)$}
        \end{subfigure}
        \caption{Illustration for proof of Theorem \ref{thm:unique} with $K=d_X-1$.}
        \label{fig:necessary}
    \end{figure}
\end{proof}

\subsection{Proof of Lemma \ref{lem:identify_GX_selection}}\label{sec:proof:lem:identify_GX_selection}

\begin{lemma}\label{lem:inducing_path_selection}
    If $d_L=0$, for any inducing path $\pi$ in $\Aug(\Gcal_{\mid\bm C},I_k)$, we have 1) $\pi$ contains at most one middle node in $\bm X$, and 2) if $\pi$ has a middle node in $\bm X$, $\pi$ doesn't intersect with $\bm S$ and all the three nodes on $\pi$ are in $\ancestor_{\Aug(\Gcal_{\mid\bm C},I_k)}(\bm S)$.
\end{lemma}
\begin{proof}[Proof of Lemma \ref{lem:inducing_path_selection}]
    \noindent\textbf{1)} We prove by contradiction. Suppose $\pi$ has end nodes $Z_j\in(\bm X, I_k)$ and $X_l\in\bm X$ and at least two observable middle nodes. Denote $X\in\bm X$ to be the one closest to $Z_j$ on $\pi$. By Definition \ref{def:inducing_path}, $X$ is a collider on $\pi$, and thus the edge on the subpath between $X$ and $X_l$ adjacent to $X$ must be into $X$. Therefore, the node $X'$ closest to $X$ on the subpath between $X$ and $X_l$ can not be a selection variable. Since $I_k$ can not be a middle node, $X'$ must be in $\bm X$. Since the edge between $X$ and $X'$ is $X\leftarrow X'$, $X'$ is not a collider on $\pi$, which contradicts the fact that $\pi$ is an inducing path.

    \noindent\textbf{2)} Suppose $\pi$ has end nodes $Z_j\in(\bm X, I_k)$ and $X_l\in\bm X$, and has middle nodes $X\in\bm X$. 
    
    If $\pi$ has middle nodes in $\bm S$, there must be a middle node $S\in\bm S$ adjacent to $X$ on $\pi$. Then the edge between $X$ and $S$ must be $X\rightarrow S$, and $X$ becomes a non-collider on $\pi$, which contradicts the fact that $\pi$ is an inducing path. 
    
    Therefore, $\pi$ equals $Z_j\rightarrow X\leftarrow X_l$. Since $\pi$ is an inducing path, $X$, and thus $Z_j,X_l$, must be in $\ancestor_{\Aug(\Gcal_{\mid\bm C},I_k)}(\bm S)$.
    
\end{proof}

\begin{proof}[Proof of Lemma \ref{lem:identify_GX_selection}]
    For any $X_j,X_k\in\bm X$, it suffices to determine the edge between them in $\Gcal_{\bm X}$ based on $\TCM(\Gcal_{\mid\bm C},\bm I)$.
    
    \noindent\textbf{1)} First, it is clear that if there is no edge between $X_j, X_k$ in $\TCM(\Gcal_{\mid\bm C},\bm I)$, there is no edge between $X_j, X_k$ in $\Gcal_{\bm X}$ as well.

    \noindent\textbf{2)} Next, we show that if $X_j\rightarrow X_k$ is in $\TCM(\Gcal_{\mid\bm C},\bm I)$, then $X_j\rightarrow X_k$ is also in $\Gcal_{\bm X}$.

    Since $X_j\rightarrow X_k$ is in $\TCM(\Gcal_{\mid\bm C},\bm I)$, there is an inducing path $\pi$ between $X_j$ and $X_k$. Since the edge has an arrowhead on $X_k$, we know $X_k$ is not in $\ancestor_{\Gcal_{\mid\bm C}}(\bm S)$. Then Lemma \ref{lem:inducing_path_selection} implies that $\pi$ doesn't contain middle nodes in $\bm X$.

    Now we argue that $\pi$ doesn't contain middle nodes in $\bm S$ as well. Otherwise, if $\pi$ contains middle nodes in $\bm S$, the node adjacent to $X_k$, denoted as $S$, must be in $\bm S$. Then the edge between $X_k$ and $S$ must be $X_k\rightarrow S$. Then $X_k\in \ancestor_{\Gcal_{\mid\bm C}}(\bm S)$, which contradicts the edge $X_j\rightarrow X_k$ in $\TCM(\Gcal_{\mid\bm C},\bm I)$. Therefore, $\pi$ must be $X_j\rightarrow X_k$, which remains in the induced subgraph $\Gcal_{\bm X}$.

    \noindent\textbf{3)} Finally, we show that if $X_j-X_k$ is in $\TCM(\Gcal_{\mid\bm C},\bm I)$, the edge between $X_j, X_k$ in $\Gcal_{\bm X}$ can also be recovered. 

    \textbf{a)} We prove that if $I_k-X_j$ is in $\TCM(\Gcal_{\mid\bm C},\bm I)$, then $X_j\rightarrow X_k$ is in $\Gcal_{\bm X}$.

    There must exist an inducing path $\pi$ between $I_k$ and $X_j$. Since $I_k$ is only adjacent to $X_k$ in $\Gcal_{\mid\bm C}$, $X_k$ is on $\pi$. Then Lemma \ref{lem:inducing_path_selection} implies that $\pi$ contains only $X_j, X_k, I_k$. Then $X_k$ must be a collider. Therefore, $X_j\rightarrow X_k$ exists in $\Gcal_{\mid\bm C}$, and thus $\Gcal_{\bm X}$.

    \textbf{b)} We prove that if $X_j\rightarrow X_k$ is in $\Gcal_{\bm X}$, then $I_k-X_j$ is in $\TCM(\Gcal_{\mid\bm C},\bm I)$.

    Since $X_j\rightarrow X_k$ is in $\Gcal_{\bm X}$, $X_k$ is not an ancestor of $X_j$. Since $X_j-X_k$ is in $\TCM(\Gcal_{\mid\bm C},\bm I)$, $X_k$ must be an ancestor of $\bm S$. Then the path $X_j\rightarrow X_k\leftarrow I_k$ is an inducing path in $\Aug(\Gcal_{\mid\bm C},I_k)$, and thus $I_k-X_j$ is in $\TCM(\Gcal_{\mid\bm C},\bm I)$.
\end{proof}

\subsection{Proof of Lemmas \ref{lem:interpretation_latent} and \ref{lem:interpretation_latent_selection}}\label{sec:proof:lem:interpretation_latent_selection}

Since Lemma \ref{lem:interpretation_latent} is a special case of Lemma \ref{lem:interpretation_latent_selection}, we only prove the latter.

\begin{lemma}\label{lem:inducing_tail_ancestry_latent_selection}
    For $j\ne k\in[d_X]$, the existence of an inducing path in $\Gcal_{\mid\bm C}$ between $X_j,X_k$ out of $X_k$ implies that $X_k\in\ancestor_{\Gcal_{\mid\bm C}}(X_j,\bm S)$.
\end{lemma}

\begin{proof}[Proof of Lemma \ref{lem:inducing_tail_ancestry_latent_selection}]
    Suppose $\pi$ is an inducing path in $\Gcal_{\mid\bm C}$ between $X_j,X_k$ and $\pi$ has a tail at $X_k$. 
    
    If $\pi$ has no collider, then every edge in $\pi$ has arrowhead pointing to $X_j$, i.e., $X_j\leftarrow\ldots\leftarrow X_k$. Therefore, $X_k\in\ancestor_{\Gcal_{\mid\bm C}}(X_j)$.

    If $\pi$ has colliders, denote $Z$ as the one closest to $X_k$, i.e., $X_j\ldots \rightarrow Z\leftarrow \ldots\leftarrow X_k$. Then the subpath of $\pi$ between $Z,X_k$ has all arrowheads pointing to $Z$, and thus $X_k\in\ancestor_{\Gcal_{\mid\bm C}}(Z)\setminus\{Z\}$. Since $Z$ is a collider on an inducing path, we have $Z\in\ancestor_{\Gcal_{\mid\bm C}}(X_j,X_k,\bm S)$. Since $X_k\in\ancestor_{\Gcal_{\mid\bm C}}(Z)\setminus\{Z\}$, we know $Z\not\in\ancestor_{\Gcal_{\mid\bm C}}(X_k)$. Then $Z\in\ancestor_{\Gcal_{\mid\bm C}}(X_j,\bm S)$, and thus $X_k\in\ancestor_{\Gcal_{\mid\bm C}}(X_j,\bm S)$.
\end{proof}

\begin{proof}[Proof of Lemma \ref{lem:interpretation_latent_selection}]
    \noindent\textbf{1)} If $X_j-X_k$ is in $\TCM(\Gcal_{\mid\bm C},\bm I)$, it follows from the proof of Lemma \ref{lem:CM} that $X_j-X_k$ is also in $\MAG(\Gcal_{\mid\bm C})$. Then $X_j\in\ancestor_{\Gcal_{\mid\bm C}}(X_k,\bm S)$ and $X_k\in\ancestor_{\Gcal_{\mid\bm C}}(X_j,\bm S)$. 
    
    If $X_j\in\ancestor_{\Gcal_{\mid\bm C}}(X_k)\setminus\{X_k\}$, $X_k\not\in\ancestor_{\Gcal_{\mid\bm C}}(X_j)$. Therefore, $X_k\in\ancestor_{\Gcal_{\mid\bm C}}(\bm S)$ and $X_j\in\ancestor_{\Gcal_{\mid\bm C}}(X_k)\subseteq\ancestor_{\Gcal_{\mid\bm C}}(\bm S)$. 

    If $X_j\in\ancestor_{\Gcal_{\mid\bm C}}(\bm S)$, then $X_k\in\ancestor_{\Gcal_{\mid\bm C}}(X_j,\bm S)=\ancestor_{\Gcal_{\mid\bm C}}(\bm S)$.

    \noindent\textbf{2)} If $X_j\leftrightarrow X_k$ is in $\TCM(\Gcal_{\mid\bm C},\bm I)$, it follows from the proof of Lemma \ref{lem:CM} that $X_j\leftrightarrow X_k$ is in $\MAG(\Gcal_{\mid\bm C})$. By Definition \ref{def:MAG}, we know $X_j\not\in\ancestor_{\Gcal_{\mid\bm C}}(X_k,\bm S),X_k\not\in\ancestor_{\Gcal_{\mid\bm C}}(X_j,\bm S)$ and there are inducing paths in $\Gcal_{\mid\bm C}$ between $X_j,X_k$. Then Lemma \ref{lem:inducing_tail_ancestry_latent_selection} implies that all such inducing paths are into both $X_j$ and $X_k$. For any such inducing path $\pi$, there must exist middle nodes on $\pi$. Let $Z$ be the one closest to $X_j$. We know $Z$ is a non-collider and thus must be in $\bm L$. Similarly, the middle node closest to $X_k$ is also in $\bm L$. Therefore, all inducing paths between $X_j,X_k$ have the form $X_j\leftarrow L\ldots L'\rightarrow X_k$ for some $L,L'\in\bm L$.

    \noindent\textbf{3)} If $X_j\rightarrow X_k$ is in $\TCM(\Gcal_{\mid\bm C},\bm I)$, it follows from the proof of Lemma \ref{lem:CM} that $X_j\rightarrow X_k$ is in $\MAG(\Gcal_{\mid\bm C})$. Then we have $X_j\in\ancestor_{\Gcal_{\mid\bm C}}(X_k,\bm S),X_k\not\in\ancestor_{\Gcal_{\mid\bm C}}(X_j,\bm S)$ and there exists inducing paths in $\Gcal_{\mid\bm C}$ between $X_j,X_k$. Then Lemma \ref{lem:inducing_tail_ancestry_latent_selection} implies that all such inducing paths are into $X_k$.

    \noindent\textbf{4)}

    \textbf{a)} If $I_j-X_k$ is in $\TCM(\Gcal_{\mid\bm C},\bm I)$, it follows from the proof of Lemma \ref{lem:CM} that $I_j-X_k$ is also in $\MAG(\Aug(\Gcal_{\mid\bm C},I_j))$. Then $X_k\in\ancestor_{\Aug(\Gcal_{\mid\bm C},I_j)}(I_j,\bm S)$ and there is an inducing path $\pi$ connecting $I_j,X_k$ in $\Aug(\Gcal_{\mid\bm C},I_j)$. On $\pi$, every non-collider middle node is latent, and every collider is in $\ancestor_{\Aug(\Gcal_{\mid\bm C},I_j)}(I_j,X_k,\bm S)$. Since $I_j$ is a source node in $\Aug(\Gcal_{\mid\bm C},I_j)$, we have $X_k\in\ancestor_{\Aug(\Gcal_{\mid\bm C},I_j)}(\bm S)$ and every collider on $\pi$ is in $\ancestor_{\Aug(\Gcal_{\mid\bm C},I_j)}(X_k,\bm S)$. For every collider $Z$, there must be a directed path from $Z$ to $(X_k,\bm S)$. Since $I_j$ can not be a middle node on the directed path, the directed path is also in $\Gcal_{\mid\bm C}$. Then all colliders on $\pi$ are in $\ancestor_{\Gcal_{\mid\bm C}}(X_k,\bm S)$. Similarly, $X_k\in\ancestor_{\Gcal_{\mid\bm C}}(\bm S)$. Since $I_j$ is only adjacent to $X_j$ in $\Aug(\Gcal_{\mid\bm C},I_j)$, $\pi$ must start with $I_j\rightarrow X_j$. Since $X_j$ is not latent, $X_j$ must be a collider and $X_j\in\ancestor_{\Gcal_{\mid\bm C}}(X_k,\bm S)$. Therefore, $X_j,X_k\in\ancestor_{\Gcal_{\mid\bm C}}(\bm S)$.

    It is easy to verify that the subpath of $\pi$ between $X_j,X_k$ is an inducing path. Since $X_j$ is a collider on $\pi$, the subpath must be into $X_j$. If there is no middle node on the subpath, the subpath equals $X_j\leftarrow X_k$. If there are middle nodes on the subpath, the one closest to $X_j$ must be a noncollider, and thus must be in $\bm L$. Then the subpath equals $X_j\leftarrow L\ldots,X_k$ for some $L\in\bm L$.

    \textbf{b)} Suppose $X_j,X_k\in\ancestor_{\Gcal_{\mid\bm C}}(\bm S)$ and $\Gcal_{\mid\bm C}$ contains an inducing path $\pi$ with either $X_j\leftarrow X_k$ or $X_j\leftarrow L\ldots X_k$ for some $L\in\bm L$. We know $\pi$ is also in $\Aug(\Gcal_{\mid\bm C},I_j)$. Then the concatenation of $I_j\rightarrow X_j$ and $\pi$ must be an inducing path in $\Aug(\Gcal_{\mid\bm C},I_j)$. Since $I_j\in\ancestor_{\Aug(\Gcal_{\mid\bm C},I_j)}(X_j)\subseteq\ancestor_{\Aug(\Gcal_{\mid\bm C},I_j)}(\bm S)$ and $X_k\in\ancestor_{\Aug(\Gcal_{\mid\bm C},I_j)}(\bm S)$, we know $I_j-X_k$ is in $\MAG(\Aug(\Gcal_{\mid\bm C},I_j))$ and thus is in $\TCM(\Gcal_{\mid\bm C},\bm I)$.

    \noindent\textbf{5)}
    
    \textbf{a)} Suppose $I_j\rightarrow X_k$ is in $\TCM(\Gcal_{\mid\bm C},\bm I)$. By Lemma \ref{lem:CM}, we know $I_j\rightarrow X_k$ is in $\MAG(\Aug(\Gcal_{\mid\bm C},I_j))$. Then $X_k\not\in\ancestor_{\Aug(\Gcal_{\mid\bm C},I_j)}(I_j,\bm S)$, and thus $X_k\not\in\ancestor_{\Gcal_{\mid\bm C}}(\bm S)$. 
    
    There is an inducing path $\pi$ connecting $I_j,X_k$ in $\Aug(\Gcal_{\mid\bm C},I_j)$. On $\pi$, every non-collider is latent, and every collider is in $\ancestor_{\Aug(\Gcal_{\mid\bm C},I_j)}(I_j,X_k,\bm S)$. Since $I_j$ is a source node in $\Aug(\Gcal_{\mid\bm C},I_j)$, every collider on $\pi$ is in $\ancestor_{\Aug(\Gcal_{\mid\bm C},I_j)}(X_k,\bm S)$. For any collider $Z$, there is a directed path in $\Aug(\Gcal_{\mid\bm C},I_j)$ from $Z$ to $(X_k,\bm S)$. Then $I_j$ is not on the path, and the path also belongs to $\Gcal_{\mid\bm C}$. Therefore, every collider on $\pi$ is in $\ancestor_{\Gcal_{\mid\bm C}}(X_k,\bm S)$. Since $I_j$ is only adjacent to $X_j$ in $\Aug(\Gcal_{\mid\bm C},I_k)$, $X_j$ must be on $\pi$ and $\pi$ must start with $I_j\rightarrow X_j$. Then $X_j$ must be a collider and $X_j\in\ancestor_{\Gcal_{\mid\bm C}}(X_k,\bm S)$. Since $X_k\not\in\ancestor_{\Gcal_{\mid\bm C}}(\bm S)$, it is easy to see $X_k\not\in\ancestor_{\Gcal_{\mid\bm C}}(X_j,\bm S)$.
    
    Moreover, the subpath $\pi_j$ of $\pi$ between $X_j,X_k$ is an inducing path in $\Gcal_{\mid\bm C}$ that is into $X_j$. Lemma \ref{lem:inducing_tail_ancestry_latent_selection} implies that $\pi_j$ is also into $X_k$. Then $\pi_j$ must contain middle nodes. Let $Z$ be the one closest to $X_j$, $Z$ is a non-collider on $\pi_j$, thus $Z\in\bm L$. Similarly, the node on $\pi_j$ closest to $X_k$ is also in $\bm L$. Therefore, $\pi_j$ has the form $X_j\leftarrow L\ldots L'\rightarrow X_k$ for some $L,L'\in\bm L$.

    \textbf{b)} Suppose $X_j\in\ancestor_{\Gcal_{\mid\bm C}}(X_k,\bm S),X_k\not\in\ancestor_{\Gcal_{\mid\bm C}}(X_j,\bm S)$ and $\Gcal_{\mid\bm C}$ contains an inducing path $\pi$ $X_j\leftarrow L\ldots L'\rightarrow X_k$ for some $L,L'\in\bm L$. Then every collider on $\pi$ is in $\ancestor_{\Gcal_{\mid\bm C}}(X_j,X_k,\bm S)=\ancestor_{\Gcal_{\mid\bm C}}(X_k,\bm S)\subseteq\ancestor_{\Aug(\Gcal_{\mid\bm C},I_j)}(X_k,\bm S)$. Let $\pi'$ be the concatenation of $I_j\rightarrow X_j$ and $\pi$ in $\Aug(\Gcal_{\mid\bm C},I_j)$, we have $X_j$ is a collider on $\pi'$ and is in $\ancestor_{\Gcal_{\mid\bm C}}(X_k,\bm S)\subseteq\ancestor_{\Aug(\Gcal_{\mid\bm C},I_j)}(X_k,\bm S)$. Then $\pi'$ is an inducing path in $\Aug(\Gcal_{\mid\bm C},I_j)$ between $I_j,X_k$ and $I_j\in\ancestor_{\Aug(\Gcal_{\mid\bm C},I_j)}(X_j)\subseteq\ancestor_{\Aug(\Gcal_{\mid\bm C},I_j)}(X_k,\bm S)$. Therefore, $I_j\rightarrow X_k$ is in $\MAG(\Aug(\Gcal_{\mid\bm C},I_j))$ and thus is also in $\TCM(\Gcal_{\mid\bm C},\bm I)$.
\end{proof}

\begin{lemma}\label{lem:confounding}
    If a walk $\pi=\langle Z_{j_0},\ldots,Z_{j_{m+1}}\rangle$ between $Z_{j_0},Z_{j_{m+1}}$ has no collider and is into both $Z_{j_0},Z_{j_{m+1}}$, then $\pi$ has the form $Z_{j_0}\leftarrow \ldots\leftarrow Z_{j_l}\rightarrow\ldots\rightarrow Z_{j_{m+1}}$, where the subwalk between $Z_{j_0},Z_{j_l}$ has all arrowheads pointing to $Z_{j_0}$ and the subwalk between $Z_{j_l},Z_{j_{m+1}}$ has all arrowheads pointing to $Z_{j_{m+1}}$.
\end{lemma}

\begin{proof}[Proof of Lemma \ref{lem:confounding}]
    Suppose $Z_{j_l}\rightarrow Z_{j_{l+1}}$ is the edge on $\pi$ closest to $Z_{j_0}$, with arrowhead pointing to $Z_{j_{m+1}}$. We know such an edge exists, because $Z_{j_m}\rightarrow Z_{j_{m+1}}$ is pointing to $Z_{j_{m+1}}$. Then the subwalk of $\pi$ between $Z_{j_0},Z_{j_l}$ equals $Z_{j_0}\leftarrow\ldots\leftarrow Z_{j_l}$. Moreover, the subwalk of $\pi$ between $Z_{j_l},Z_{j_{m+1}}$ must be $Z_{j_l}\rightarrow\ldots\rightarrow Z_{j_{m+1}}$. Otherwise, suppose $k$ is the smallest index $l<k<m$ with edge $Z_{j_k}\leftarrow Z_{j_{k+1}}$ pointing to $Z_{j_0}$. Then $Z_{j_k}$ becomes a collider, contradicting the assumption that there are no colliders. Then we conclude the lemma.
\end{proof}

\begin{lemma}\label{lem:confounding_seq}
    If there is an inducing path $\pi$ of the form $X_j\leftarrow L\ldots L'\rightarrow X_k$ for some $L,L'\in\bm L$, then there exist a sequence of nodes $Z_{j_0}\overset{\triangle}{=}X_j,Z_{j_1},\ldots,Z_{j_m},Z_{j_{m+1}}\overset{\triangle}{=}X_k$ in $(\bm X,\bm S)$, such that $Z_{j_{l-1}},Z_{j_l}$ are latently confounded for $l\in[m+1]$. 
\end{lemma}

\begin{proof}[Proof of Lemma \ref{lem:confounding_seq}]
    If there is no collider on $\pi$, all middle nodes are in $\bm L$. Then Lemma \ref{lem:confounding} implies that $X_j,X_k$ are latently confounded.

    If there are colliders on $\pi$, denote $\tilde Z_1, \ldots,\tilde Z_m$ as all the colliders in order. We know all other middle nodes are noncolliders and are in $\bm L$. Since $\pi$ is an inducing path, each $\tilde Z_l$ is in $\ancestor_{\Gcal_{\mid\bm C}}(X_j,X_k,\bm S)$. Then there is a node $Z_{j_l}$ (equal to $\tilde Z_l$ if $\tilde Z_l\not\in\bm L$) in $\bm X\cup\bm S$ such that there exists a directed path $\tilde Z_l\rightarrow\ldots\rightarrow Z_{j_l}$ with all middle nodes in $\bm L$. For $\tilde Z_l,\tilde Z_{l+1}$, if we concatenate $Z_{j_l}\leftarrow\ldots\leftarrow\tilde Z_l$, the subpath of $\pi$ between $\tilde Z_l,\tilde Z_{l+1}$, and $\tilde Z_{l+1}\rightarrow\ldots\rightarrow Z_{j_{l+1}}$, we know the concatenated walk has no collider with all middle nodes in $\bm L$ and is into both $Z_{j_l},Z_{j_{l+1}}$. Then Lemma \ref{lem:confounding} implies that $Z_{j_l},Z_{j_{l+1}}$ are latently confounded. Similarly, $X_j,Z_{j_1}$ are latently confounded and $Z_{j_m},X_k$ are latently confounded.
\end{proof}

\subsection{Proof of Proposition \ref{prop:adjacency_latent_selection}}\label{sec:proof:prop:adjacency_latent_selection}

\begin{proof}[Proof of Proposition \ref{prop:adjacency_latent_selection}]
    We prove the three parts separately.
    
    \noindent\textbf{1)} $\ancestor_{\Gcal_{\mid\bm C}}(X_j)\setminus (X_j,\bm L)\subseteq\bm A_j=\anterior_{\MAG(\Gcal_{\mid\bm C})}(X_j)\setminus\{X_j\}\subseteq\ancestor_{\Gcal_{\mid\bm C}}(X_j,\bm S)\setminus(X_j,\bm L,\bm S)$.

    \textbf{a)} $\bm A_j\subseteq\ancestor_{\Gcal_{\mid\bm C}}(X_j,\bm S)\setminus(X_j,\bm L,\bm S)$.
    
    Firstly, we show $(I_i\not\indep X_j\mid\bm S)_{\Aug(\Gcal_{\mid\bm C}, I_i)}\Longrightarrow X_i\in\ancestor_{\Gcal_{\mid\bm C}}(X_j,\bm S)$. To see this, for $i\in[d_X]\setminus\{j\}$, when $(I_i\not\indep X_j\mid\bm S)_{\Aug(\Gcal_{\mid\bm C}, I_i)}$ holds, there exists an open path $\pi$ connecting $I_i, X_j$ given $\bm S$ in $\Aug(\Gcal_{\mid\bm C}, I_i)$. Since $I_i$ only has one adjacent node $X_i$ in $\Aug(\Gcal_{\mid\bm C}, I_i)$, $\pi$ must start with $I_i\rightarrow X_i$. If there is no collider on $\pi$, all the arrowheads on $\pi$ must point to $X_j$, and thus $X_i\in\ancestor_{\Gcal_{\mid\bm C}}(X_j)$. If there exist colliders on $\pi$, denote $Z$ to be the one closest to $I_i$. Since there is no collider on the subpath between $I_i, Z$, and the subpath starts from $I_i\rightarrow X_i$, we know $X_i\in\ancestor_{\Gcal_{\mid\bm C}}(Z)$. Moreover, since $\pi$ is open, $Z\in\ancestor_{\Gcal_{\mid\bm C}}(\bm S)$. Therefore, we have $X_i\in\ancestor_{\Gcal_{\mid\bm C}}(X_j,\bm S)\setminus\{X_j\}$.

    Then, we know $\bm A_j\subseteq\ancestor_{\Gcal_{\mid\bm C}}(X_j,\bm S)\setminus\{X_j\}$. Since $\bm X$, $\bm L$ and $\bm S$ are disjoint, we have $\bm A_j\subseteq\ancestor_{\Gcal_{\mid\bm C}}(X_j,\bm S)\setminus (X_j,\bm L,\bm S)$.

    \textbf{b)} $\ancestor_{\Gcal_{\mid\bm C}}(X_j)\setminus (X_j,\bm L)\subseteq\bm A_j$.

    If $X_i\in\ancestor_{\Gcal_{\mid\bm C}}(X_j)\setminus (X_j,\bm L)$, there exists a path $\pi$ $X_i\rightarrow\ldots\rightarrow X_j$ in $\Gcal_{\mid\bm C}$ on which all arrowheads point to $X_j$. Since $\bm S$ has no descendant, we know $\pi$ is disjoint from $\bm S$. Then the concatenated path $I_i\rightarrow X_i\rightarrow\ldots\rightarrow X_j$ is disjoint from $\bm S$ and has no collider, and thus is open given $\bm S$ in $\Aug(\Gcal_{\mid\bm C}, I_i)$. Consequently, $(I_i\not\indep X_j\mid\bm S)_{\Aug(\Gcal_{\mid\bm C}, I_i)}$ and $X_i\in\bm A_j$.

    \textbf{c)} $\bm A_j\subseteq\anterior_{\MAG(\Gcal_{\mid\bm C})}(X_j)\setminus\{X_j\}$.

    If $X_k\in\bm A_j$, we know $(I_k\not\indep X_j\mid\bm S)_{\Aug(\Gcal_{\mid\bm C},I_k)}$. It follows from Theorem 4.18 in \cite{richardson2002ancestral} that $(I_k\not\indep X_j)_{\MAG(\Aug(\Gcal_{\mid\bm C},I_k))}$. Then there is an $m$-open path $\pi$ between $I_k,X_j$ in $\MAG(\Aug(\Gcal_{\mid\bm C},I_k))$. Denote the node on $\pi$ adjacent to $I_k$ as $X_l$. 
    
    First, we show that the edge in $\MAG(\Aug(\Gcal_{\mid\bm C},I_k))$ between $I_k,X_l$ has a tail at $I_k$. If $l=k$, since $I_k\in\ancestor_{\Aug(\Gcal_{\mid\bm C},I_k)}(X_k)$, it follows from Definition \ref{def:MAG} that the edge in $\MAG(\Aug(\Gcal_{\mid\bm C},I_k))$ between $I_k,X_k$ has a tail at $I_k$. If $l\ne k$, there is an inducing path $\pi_l$ in $\Aug(\Gcal_{\mid\bm C},I_k)$ connecting $I_k,X_l$. $\pi_l$ must start with $I_k\rightarrow X_k$, and $X_k$ must be a collider. Then, $X_k\in\ancestor_{\Aug(\Gcal_{\mid\bm C},I_k)}(X_l,\bm S)$, and thus $I_k\in\ancestor_{\Aug(\Gcal_{\mid\bm C},I_k)}(X_l,\bm S)$. Then it follows from Definition \ref{def:MAG} that the edge in $\MAG(\Aug(\Gcal_{\mid\bm C},I_k))$ between $I_k,X_l$ must have a tail at $I_k$.

    Next, we show that the open path $\pi$ in $\MAG(\Aug(\Gcal_{\mid\bm C},I_k))$ between $I_k,X_j$ has all arrowheads pointing to $X_j$. We prove by induction. Suppose $\pi=\langle I_k,X_l,X_{l_1},\ldots,X_{l_m},X_j\rangle$. We have shown that the edge between $I_k,X_l$ has a tail at $I_k$. Suppose the edge between $X_{l_{i-1}},X_{l_i}$ has a tail at $X_{l_{i-1}}$. If the edge is $X_{l_{i-1}}\rightarrow X_{l_i}$, the edge between $X_{l_i},X_{l_{i+1}}$ must have a tail at $X_{l_i}$. Otherwise, $X_{l_i}$ is a collider which blocks the path $\pi$. If the edge is $X_{l_{i-1}}-X_{l_i}$, it follows from Lemma 3.2(b) in \cite{richardson2002ancestral} that the edge between $X_{l_i},X_{l_{i+1}}$ must also have a tail at $X_{l_i}$. By induction, we conclude that $\pi$ has all arrowheads pointing to $X_j$.

    Then the subpath of $\pi$ between $X_l$ and $X_j$ also has all arrowheads pointing to $X_j$. Since the proof of Lemma \ref{lem:CM} implies that this subpath is also in $\MAG(\Gcal_{\mid\bm C})$, we have $X_l\in\anterior_{\MAG(\Gcal_{\mid\bm C})}(X_j)$.
    
    Next, we show that $X_k\in\anterior_{\MAG(\Gcal_{\mid\bm C})}(X_l)$. We only need to consider the case with $k\ne l$. Recall $\pi_l$ is an inducing path in $\Aug(\Gcal_{\mid\bm C}, I_k)$ between $I_k,X_l$ and $X_k$ is on $\pi_l$. Then, among the middle nodes $\pi_l$, all observed nodes are colliders, non-colliders are latent, and all colliders are in $\ancestor_{\Aug(\Gcal_{\mid\bm C},I_k)}(I_k,X_l,\bm S)$. Since $I_k$ is a source node in $\Aug(\Gcal_{\mid\bm C},I_k)$, all colliders are in $\ancestor_{\Aug(\Gcal_{\mid\bm C},I_k)}(X_l,\bm S)$. Particularly, $X_k\in\ancestor_{\Aug(\Gcal_{\mid\bm C},I_k)}(X_l,\bm S)$. Denote $\pi_l'$ as the subpath between $X_k$ and $X_l$. We know that, among the middle nodes of $\pi_l'$, all observed nodes are colliders, non-colliders are latent, and all colliders are in $\ancestor_{\Aug(\Gcal_{\mid\bm C},I_k)}(X_l,\bm S)$. Therefore, $\pi_l'$ is an inducing path and $X_k,X_l$ are adjacent in $\MAG(\Aug(\Gcal_{\mid\bm C},I_k))$. Since $X_k\in\ancestor_{\Aug(\Gcal_{\mid\bm C},I_k)}(X_l,\bm S)$, the edge between $X_k,X_l$ has a tail at $X_k$. By the proof of Lemma \ref{lem:CM}, $X_k,X_l$ are adjacent in $\MAG(\Gcal_{\mid\bm C})$ and the edge between them has a tail at $X_k$. Then $X_k\in\anterior_{\MAG(\Gcal_{\mid\bm C})}(X_l)$.

    Since $X_l\in\anterior_{\MAG(\Gcal_{\mid\bm C})}(X_j)$ and $X_k\in\anterior_{\MAG(\Gcal_{\mid\bm C})}(X_l)$, Proposition 2.1 in \cite{richardson2002ancestral} implies that $X_k\in\anterior_{\MAG(\Gcal_{\mid\bm C})}(X_j)$. In conclusion, we prove that $\bm A_j\subseteq\anterior_{\MAG(\Gcal_{\mid\bm C})}(X_j)\setminus\{X_j\}$.

    \textbf{d)} $\bm A_j\supseteq\anterior_{\MAG(\Gcal_{\mid\bm C})}(X_j)\setminus\{X_j\}$.

    Suppose $X_k\in\anterior_{\MAG(\Gcal_{\mid\bm C})}(X_j)\setminus\{X_j\}$, there is a path $\pi$ in $\MAG(\Gcal_{\mid\bm C})$ between $X_k,X_j$ with all arrowheads pointing to $X_j$. The proof of Lemma \ref{lem:CM} implies that $\pi$ is also in $\MAG(\Aug(\Gcal_{\mid\bm C},I_k))$. Since $I_k\in\ancestor_{\Aug(\Gcal_{\mid\bm C},I_k)}(X_k)$, we know the edge in $\MAG(\Aug(\Gcal_{\mid\bm C},I_k))$ between $I_k,X_k$ has a tail at $I_k$. Then the concatenated path of $\langle I_k,X_k\rangle$ and $\pi$ in $\MAG(\Aug(\Gcal_{\mid\bm C},I_k))$ has all arrowheads pointing to $X_j$, and thus is open. Then $(I_k\not\indep X_j)_{\MAG(\Aug(\Gcal_{\mid\bm C},I_k))}$. By Theorem 4.18 in \cite{richardson2002ancestral}, we have $(I_k\not\indep X_j\mid\bm S)_{\Aug(\Gcal_{\mid\bm C},I_k)}$ and $X_k\in\bm A_j$.
    
    \noindent\textbf{2)} $\forall\bm A$ s.t. $\bm A_j\cup\bm A_k\subseteq\bm A\subseteq\ancestor_{\Gcal_{\mid\bm C}}(X_j,X_k,\bm S)\setminus\bm L$, $X_j, X_k$ adjacent in $\CM(\Gcal_{\mid\bm C}, \bm I)\Longleftrightarrow\big(X_j\not\indep X_k\mid (\bm A,\bm S)\setminus(X_j,X_k)\big)_{\Gcal_{\mid\bm C}}$.
    
    \textbf{a)} $X_j, X_k$ adjacent in $\CM(\Gcal_{\mid\bm C}, \bm I)\Longrightarrow\big(X_j\not\indep X_k\mid (\bm A,\bm S)\setminus(X_j,X_k)\big)_{\Gcal_{\mid\bm C}}$.

    By the proof of Lemma \ref{lem:CM}, we know there exists an inducing path $\pi$ between $X_j,X_k$ in $\Gcal_{\mid\bm C}$ on which every middle node is either in $\bm L$ or a collider, and all colliders are in $\ancestor_{\Gcal_{\mid\bm C}}(X_j, X_k,\bm S)$. Since $\bm L$ is disjoint from $(\bm A,\bm S)\setminus(X_j,X_k)$ and $\ancestor_{\Gcal_{\mid\bm C}}(X_j,X_k,\bm S)\subseteq\ancestor_{\Gcal_{\mid\bm C}}(\bm A\cup\{X_j,X_k\}\cup\bm S)$, it follows from Lemma 3.14 in \cite{richardson2002ancestral} that $\big(X_j\not\indep X_k\mid (\bm A,\bm S)\setminus(X_j,X_k)\big)_{\Gcal_{\mid\bm C}}$.

    \textbf{b)} $\big(X_j\not\indep X_k\mid (\bm A,\bm S)\setminus(X_j,X_k)\big)_{\Gcal_{\mid\bm C}}\Longrightarrow$ $X_j, X_k$ adjacent in $\CM(\Gcal_{\mid\bm C}, \bm I)$.

    There exists an open path $\pi$ connecting $X_j,X_k$ given $(\bm A, \bm S)\setminus(X_j,X_k)$ in $\Gcal_{\mid\bm C}$. By Lemma 3.13 in \cite{richardson2002ancestral}, every middle node on $\pi$ is in $\ancestor_{\Gcal_{\mid\bm C}}(\bm A,X_j,X_k,\bm S)$, which by definition of $\bm A$, is also in $\ancestor_{\Gcal_{\mid\bm C}}(X_j,X_k,\bm S)$. Since non-colliders middle nodes on $\pi$ are not in $(\bm A, X_j, X_k, \bm S)$, they must be in $\ancestor_{\Gcal_{\mid\bm C}}(X_j,X_k,\bm S)\setminus(\bm A, X_j, X_k, \bm S)$. Since by \textbf{1) b)}, $\ancestor_{\Gcal_{\mid\bm C}}(X_j,X_k)\setminus(X_j,X_k,\bm L)\subseteq\bm A_j\cup\bm A_k\subseteq\bm A$, we know every non-collider is either in $\bm L$ or in $\ancestor_{\Gcal_{\mid\bm C}}(\bm S)\setminus(\bm A_j\cup\bm A_k\cup\bm L\cup\bm S)$. 
    
    Now we argue that non-colliders can not be in $\ancestor_{\Gcal_{\mid\bm C}}(\bm S)\setminus(\bm A_j\cup\bm A_k\cup\bm L\cup\bm S)$. To see this, if there are non-colliders in $\ancestor_{\Gcal_{\mid\bm C}}(\bm S)\setminus(\bm A_j\cup\bm A_k\cup\bm L\cup\bm S)\subseteq\bm X$, denote $X_i$ to be such a middle node. Since $X_i$ is a non-collider on $\pi$, at least one of the two edges adjacent to $X_i$ is out of $X_i$. Without the loss of generality, we assume the edge adjacent to $X_i$ on the subpath $\pi_j$ of $\pi$ between $X_j,X_i$ is out of $X_i$. 
    
    On the subpath $\pi_j$, if there is no collider, then all the arrowheads must point to $X_j$ and the concatenated path $I_i\rightarrow X_i\rightarrow\ldots\rightarrow X_j$ has no collider and every middle point is not in $\bm S$ since $\pi$ is open given $(\bm A,\bm S)\setminus(X_j,X_k)$. Therefore, we have $(I_i\not\indep X_j\mid\bm S)_{\Aug(\Gcal_{\mid\bm C}, I_i)}$ and $X_i\in\bm A_j$, which contradicts the definition of $X_i$.

    If there are colliders on $\pi_j$, denote $Z$ to be the one closest to $X_i$. Then the subpath of $\pi$ between $X_i,Z$ has all arrowheads point to $Z$, i.e., $X_i\rightarrow\ldots\rightarrow Z$, and no middle node is in $\bm S$. Since $\pi$ is open given $(\bm A,\bm S)\setminus(X_j,X_k)$ in $\Gcal_{\mid\bm C}$, we know $Z\in\ancestor_{\Gcal_{\mid\bm C}}\big((\bm A,\bm S)\setminus(X_j,X_k)\big)\subseteq\ancestor_{\Gcal_{\mid\bm C}}(X_j,X_k,\bm S)$. 
    
    If $Z$ is in $\ancestor_{\Gcal_{\mid\bm C}}(X_j)\setminus\ancestor_{\Gcal_{\mid\bm C}}(\bm S)$, then there exists a path $Z\rightarrow\ldots\rightarrow X_j$ between $Z,X_j$ with all arrowheads point to $X_j$ and no middle node in $\bm S$. Therefore, the concatenated path $I_i\rightarrow X_i\rightarrow\ldots\rightarrow Z\rightarrow\ldots\rightarrow X_j$ is open given $\bm S$ in $\Aug(\Gcal_{\mid\bm C}, I_i)$, and $X_i\in\bm A_j$ which contradicts the definition of $X_i$. 
    
    Similarly, $Z$ is not in $\ancestor_{\Gcal_{\mid\bm C}}(X_k)\setminus\ancestor_{\Gcal_{\mid\bm C}}(\bm S)$. Then, $Z$ must be in $\ancestor_{\Gcal_{\mid\bm C}}(\bm S)\subseteq\ancestor_{\Aug(\Gcal_{\mid\bm C},I_i)}(\bm S)$. 
    
    If there is no collider on $\pi$ between $Z,X_j$, the path $\langle I_i, X_i, \ldots, Z,\ldots, X_j\rangle$ only contains one collider $Z\in\ancestor_{\Gcal_{\mid\bm C}}(\bm S)\subseteq\ancestor_{\Aug(\Gcal_{\mid\bm C},I_i)}(\bm S)$ and no non-collider is in $\bm S$. Then the path is open given $\bm S$ in $\Aug(\Gcal_{\mid\bm C}, I_i)$ and $X_i\in\bm A_j$, which contradicts the definition of $X_i$.

    If there are colliders on $\pi$ between $Z,X_j$, denote $Z'$ to be the one closest to $Z$. Then $Z'\in\ancestor_{\Gcal_{\mid\bm C}}\big((\bm A,\bm S)\setminus(X_j,X_k)\big)\subseteq\ancestor_{\Gcal_{\mid\bm C}}(X_j,X_k,\bm S)$. If $Z'\in\ancestor_{\Gcal_{\mid\bm C}}(X_j)\setminus\ancestor_{\Gcal_{\mid\bm C}}(\bm S)$, there is a path $Z'\rightarrow\ldots\rightarrow X_j$ with no middle node in $\bm S$. Then the concatenated walk $\langle I_i,X_i,\ldots,Z,\ldots,Z',\ldots, X_j\rangle$ has only one collider $Z\in\ancestor_{\Aug(\Gcal_{\mid\bm C},I_i)}(\bm S)$ and no non-collider is in $\bm S$. Therefore, $X_i\in\bm A_j$ contradicts the definition of $X_i$. Similarly, $Z'$ is not in $\ancestor_{\Gcal_{\mid\bm C}}(X_k)\setminus\ancestor_{\Gcal_{\mid\bm C}}(\bm S)$ and thus must be in $\ancestor_{\Gcal_{\mid\bm C}}(\bm S)\subseteq\ancestor_{\Aug(\Gcal_{\mid\bm C},I_i)}(\bm S)$. 
    
    Similarly, by induction, all colliders on $\pi_j$ are in $\ancestor_{\Aug(\Gcal_{\mid\bm C},I_i)}(\bm S)$ and no non-colliders are in $\bm S$. So the concatenation of $I_i\rightarrow X_i$ and $\pi_j$ is open given $\bm S$ in $\Aug(\Gcal_{\mid\bm C}, I_i)$ and $X_i\in\bm A_j$ contradicts the definition of $X_i$.

    Therefore, every non-collider on $\pi$ is in $\bm L$, and every collider is in $\ancestor_{\Gcal_{\mid\bm C}}(\bm A,X_j,X_k,\bm S)\subseteq\ancestor_{\Gcal_{\mid\bm C}}(X_j,X_k,\bm S)$. Therefore, $\pi$ is an inducing path in $\Gcal_{\mid\bm C}$. By the proof of Lemma \ref{lem:CM}, $X_j,X_k$ are adjacent in $\CM(\Gcal_{\mid\bm C},\bm I)$.

    \noindent\textbf{3)} $\forall j\ne k$, $\forall\bm A_j'$ s.t. $\bm A_j\subseteq \bm A_j'\subseteq\ancestor_{\Gcal_{\mid\bm C}}(X_j,\bm S)\setminus\bm L$, $X_j,I_k$ adjacent in $\CM(\Gcal_{\mid\bm C},\bm I)\Longleftrightarrow \big(I_k\not\indep X_j\mid(\bm A_j',\bm S)\setminus\{X_j\}\big)_{\Aug(\Gcal_{\mid\bm C}, I_k)}$.

    The proof for part 3) is similar to that of part 2), so we omit the details. Below, we sketch the proof of the reverse direction.

    \textbf{a)} $\big(I_k\not\indep X_j\mid(\bm A_j',\bm S)\setminus\{X_j\}\big)_{\Aug(\Gcal_{\mid\bm C}, I_k)}\Longrightarrow X_j,I_k$ adjacent in $\CM(\Gcal_{\mid\bm C},\bm I)$.

    There must exists an open path $\pi$ connecting $I_k, X_j$ given $(A_j',\bm S)$ in $\Aug(\Gcal_{\mid\bm C}, I_k)$. Similar to \textbf{2) b)}, every non-collider is either in $\bm L$ or in $\ancestor_{\Gcal_{\mid\bm C}}(\bm S)\setminus(\bm A_j\cup\bm L\cup\bm S)$.

    Then we show non-colliders can not be in $\ancestor_{\Gcal_{\mid\bm C}}(\bm S)\setminus(\bm A_j\cup\bm L\cup\bm S)$. To see this, if there is such a non-collider $X_i$ in $\ancestor_{\Gcal_{\mid\bm C}}(\bm S)\setminus(\bm A_j\cup\bm L\cup\bm S)\subseteq\bm X$, we let $\pi_j$ to be the subpath of $\pi$ between $X_i, X_j$. We know $\pi_j$ also exists in $\Aug(\Gcal_{\mid\bm C}, I_i)$. Since $X_i$ is in $\ancestor_{\Aug(\Gcal_{\mid\bm C}, I_i)}(\bm S)$ but is not in $\bm S$, the concatenation of $I_i\rightarrow X_i$ and $\pi_j$ always has the two edges adjacent to $X_i$ being open given $\bm S$, regardless of whether $X_i$ is a collider.

    If there is no collider on $\pi_j$, we have the concatenation of $I_i\rightarrow X_i$ and $\pi_j$ being open given $\bm S$ and $X_i\in\bm A_j$, which contradicts the definition of $X_i$.

    If there are colliders on $\pi_j$, similar to \textbf{2) b)}, all colliders are in $\ancestor_{\Aug(\Gcal_{\mid\bm C}, I_i)}(\bm S)$ and no non-collider is in $\bm S$. Then the concatenation of $I_i\rightarrow X_i$ and $\pi_j$ being open given $\bm S$ and $X_i\in\bm A_j$, which contradicts the definition of $X_i$.

    Therefore, non-colliders on $\pi$ are all in $\bm L$, and every collider is in $\ancestor_{\Aug(\Gcal_{\mid\bm C}, I_k)}(X_j,\bm S)$. Consequently, $\pi$ is an inducing path in $\Aug(\Gcal_{\mid\bm C}, I_k)$ and $I_k, X_j$ are adjacent in $\CM(\Gcal_{\mid\bm C}, \bm I)$.
    
\end{proof}

\subsection{Proof of Proposition \ref{prop:direction_latent_selection}}\label{sec:proof:prop:direction_latent_selection}

\begin{proof}[Proof of Proposition \ref{prop:direction_latent_selection}]

    \noindent\textbf{1)} By Definition \ref{def:MAG}, it suffices to show the equivalence between $X_j\in\ancestor_{\Gcal_{\mid\bm C}}(X_k,\bm S)$ and $X_j\in \bm A_k'$ when $X_j,X_k$ are adjacent in $\CM(\Gcal_{\mid\bm C},\bm I)$.
    
    \textbf{i)} $X_j\in\ancestor_{\Gcal_{\mid\bm C}}(X_k,\bm S)\Longrightarrow X_j\in \bm A_k'$.

    Since $X_j,X_k$ are adjacent in $\CM(\Gcal_{\mid\bm C},\bm I)$, by proof of Lemma \ref{lem:CM}, there exists an inducing path $\pi$ between $X_j,X_k$ in $\Gcal_{\mid\bm C}$. Then, every middle node on $\pi$ is either in $\bm L$ or is a collider, and every collider is in $\ancestor_{\Gcal_{\mid\bm C}}(X_j,X_k,\bm S)$. Since $X_j\in\ancestor_{\Gcal_{\mid\bm C}}(X_k,\bm S)$, we know $\ancestor_{\Gcal_{\mid\bm C}}(X_j,X_k,\bm S)\subseteq\ancestor_{\Gcal_{\mid\bm C}}(X_k,\bm S)$, thus all colliders together with $X_j$ are in $\ancestor_{\Gcal_{\mid\bm C}}(X_k,\bm S)$. 
    
    If among all colliders and $X_j$, there are nodes in $\ancestor_{\Gcal_{\mid\bm C}}(X_k)\setminus\ancestor_{\Gcal_{\mid\bm C}}(\bm S)$, denote $Z$ to be the one closest to $X_j$ (possibly $Z=X_j$). Then, there exists a path $Z\rightarrow \ldots\rightarrow X_k$ with no middle point in $\bm S$. If $Z\ne X_j$, for the subpath of $\pi$ between $X_j,Z$, all colliders and $X_j$ are in $\ancestor_{\Gcal_{\mid\bm C}}(\bm S)$ by the definition of $Z$. Then the concatenated walk $\langle I_j, X_j,\ldots,Z,\ldots,X_k\rangle$ in $\Aug(\Gcal_{\mid\bm C}, I_j)$ has all colliders in $\ancestor_{\Aug(\Gcal_{\mid\bm C},I_j)}(\bm S)$ and no non-collider in $\bm S$, regardless whether $Z=X_j$ or whether $X_j$ is a collider. Therefore, the walk is open given $\bm S$ in $\Aug(\Gcal_{\mid\bm C}, I_j)$, and thus $X_j\in\bm A_k\subseteq\bm A_k'$.

    If all colliders and $X_j$ are in $\ancestor_{\Gcal_{\mid\bm C}}(\bm S)$, the concatenation of $I_j\rightarrow X_j$ and $\pi$ is open given $\bm S$ in $\Aug(\Gcal_{\mid\bm C}, I_j)$, and thus $X_j\in\bm A_k\subseteq\bm A_k'$.

    \textbf{ii)} $X_j\in\bm A_k'\Longrightarrow X_j\in\ancestor_{\Gcal_{\mid\bm C}}(X_k,\bm S)$.

    By Proposition \ref{prop:adjacency_latent_selection} and the definition of $\bm A_k'$, we have $\bm A_k'\subseteq\ancestor_{\Gcal_{\mid\bm C}}(X_k,\bm S)$.

    Combining pieces, part 1) follows from Definition \ref{def:MAG}.

    \noindent\textbf{2)} By Definition \ref{def:MAG}, we only need to show that if $I_k, X_j$ are adjacent in $\CM(\Gcal_{\mid\bm C}, \bm I)$, then $I_k\in\ancestor_{\Aug(\Gcal_{\mid\bm C}, I_k)}(X_j,\bm S)$, and we have $X_j\in\bm A_k'\Longleftrightarrow X_j\in\ancestor_{\Gcal_{\mid\bm C}}(\bm S)$.

    \textbf{i)} $I_k\in\ancestor_{\Aug(\Gcal_{\mid\bm C}, I_k)}(X_j,\bm S)$.

    There must exist an inducing path $\pi$ between $I_k, X_j$ in $\Aug(\Gcal_{\mid\bm C}, I_k)$. Then every collider on $\pi$ is in $\ancestor_{\Aug(\Gcal_{\mid\bm C}, I_k)}(I_k,X_j,\bm S)=\ancestor_{\Aug(\Gcal_{\mid\bm C}, I_k)}(X_j,\bm S)\cup\{I_k\}$ and every non-collider on $\pi$ is in $\bm L$. Since $\pi$ must start with $I_k\rightarrow X_k$ and $X_k$ is not in $\bm L$, $X_k$ must be a collider on $\pi$. Since $X_k\in\ancestor_{\Aug(\Gcal_{\mid\bm C}, I_k)}(X_j,\bm S)$, we know $I_k\in\ancestor_{\Aug(\Gcal_{\mid\bm C}, I_k)}(X_j,\bm S)$.

    \textbf{ii)} $X_j\in\bm A_k'\Longrightarrow X_j\in\ancestor_{\Gcal_{\mid\bm C}}(\bm S)$.

    By Proposition \ref{prop:adjacency_latent_selection} and the definition of $\bm A_k'$, we have $\bm A_k'\subseteq\ancestor_{\Gcal_{\mid\bm C}}(X_k,\bm S)$. Then $X_j\in\ancestor_{\Gcal_{\mid\bm C}}(X_k,\bm S)$. By part \textbf{2) i)}, $X_k\in\ancestor_{\Gcal_{\mid\bm C}}(X_j,\bm S)\setminus\{X_j\}$. If $X_k\in\ancestor_{\Gcal_{\mid\bm C}}(X_j)\setminus\{X_j\}$, then $X_j\not\in\ancestor_{\Gcal_{\mid\bm C}}(X_k)$ and thus $X_j\in\ancestor_{\Gcal_{\mid\bm C}}(\bm S)$. If $X_k\in\ancestor_{\Gcal_{\mid\bm C}}(\bm S)$, then $X_j\in\ancestor_{\Gcal_{\mid\bm C}}(X_k,\bm S)\subseteq\ancestor_{\Gcal_{\mid\bm C}}(\bm S)$.

    \textbf{iii)} $X_j\in \ancestor_{\Gcal_{\mid\bm C}}(\bm S)\Longrightarrow X_j\in\bm A_k'$.

    There must be an inducing path $\pi$ between $I_k,X_j$ in $\Aug(\Gcal_{\mid\bm C}, I_k)$ and $\pi$ starts with $I_k\rightarrow X_k$. On the subpath $\pi_k$ of $\pi$ between $X_k,X_j$, all non-colliders are in $\bm L$ and all colliders are in $\ancestor_{\Gcal_{\mid\bm C}}(X_j,\bm S)$. 
    
    If there is no collider on $\pi_k$, we know the concatenated path of $\pi_k$ and $I_j\rightarrow X_j$ in $\Aug(\Gcal_{\mid\bm C}, I_j)$ has at most one collider $X_j$. Since $X_j$ is in $\ancestor_{\Aug(\Gcal_{\mid\bm C},I_j)}(\bm S)$ but not in $\bm S$, we know the concatenated path of $\pi_k$ and $I_j\rightarrow X_j$ in $\Aug(\Gcal_{\mid\bm C}, I_j)$ is open given $\bm S$ regardless of whether $X_j$ is a collider. Then $X_j\in\bm A_k\subseteq\bm A_k'$.

    If there are colliders on $\pi_k$, denote $Z$ to be the one closest to $X_k$, then there is no collider between $X_k$ and $Z$. Since $\pi$ is an inducing path, $Z\in\ancestor_{\Gcal_{\mid\bm C}}(X_j,\bm S)$.

    If $Z\in\ancestor_{\Gcal_{\mid\bm C}}(X_j)\setminus\ancestor_{\Gcal_{\mid\bm C}}(\bm S)$, there exists a path $Z\rightarrow\ldots\rightarrow X_j$ with all arrowheads pointing to $X_j$ and no middle point in $\bm S$. Then the concatenated walk of $\langle X_k,\ldots,Z\rangle$ and $Z\rightarrow\ldots\rightarrow X_j\leftarrow I_j$ has a single collider $X_j\in\ancestor_{\Aug(\Gcal_{\mid\bm C}, I_j)}(\bm S)$, and no non-collider in $\bm S$. Therefore the walk is open given $\bm S$ in $\Aug(\Gcal_{\mid\bm C},I_j)$ and $X_j\in\bm A_k\subseteq\bm A_k'$.

    If $Z\in\ancestor_{\Gcal_{\mid\bm C}}(\bm S)$, we apply induction on the subpath of $\pi$ between $Z,X_j$. By the same argument as above, if there is no collider on the subpath, or the first collider is in $\ancestor_{\Gcal_{\mid\bm C}}(X_j)\setminus\ancestor_{\Gcal_{\mid\bm C}}(\bm S)$, we have $X_j\in\bm A_k'$. Finally, if all colliders on $\pi_k$ are in $\ancestor_{\Gcal_{\mid\bm C}}(\bm S)$, the concatenated path of $\pi_k$ and $I_j\rightarrow X_j$ is open given $\bm S$ in $\Aug(\Gcal_{\mid\bm C}, I_j)$. Therefore, $X_j\in\bm A_k\subseteq\bm A_k'$.
\end{proof}

\subsection{Proof of Lemma \ref{lem:S}}\label{sec:proof:lem:S}

\begin{proof}[Proof of Lemma \ref{lem:S}]
    \noindent\textbf{1)} $X_j\in\ancestor_{\Gcal_{\mid\bm C}}(\bm S\setminus\bm S_{\rm is})\Longrightarrow \exists X_k\in\bm X\setminus\{X_j\}$ s.t. $(I_k\not\indep X_j\mid\bm S)_{\Aug(\Gcal_{\mid\bm C}, I_k)}$ and $(I_j\not\indep X_k\mid\bm S)_{\Aug(\Gcal_{\mid\bm C}, I_j)}$.

    By Definition \ref{def:isolated_selection}, there exists $X_k\in\ancestor_{\Gcal_{\mid\bm C}}(\bm S)\setminus\{X_j\}$ such that $(X_j\not\indep X_k\mid\bm S)_{\Gcal_{\mid\bm C}}$. Therefore, there is an open path $\pi$ between $X_j,X_k$ given $\bm S$ in $\Gcal_{\mid\bm C}$, on which all colliders are in $\ancestor_{\Gcal_{\mid\bm C}}(\bm S)$ and no non-collider is in $\bm S$. Consider the concatenation of $I_j\rightarrow X_j$ and $\pi$ in $\Aug(\Gcal_{\mid\bm C}, I_j)$, no matter whether $X_j$ is a collider, we have all colliders on the path are in $\ancestor_{\Aug(\Gcal_{\mid\bm C}, I_j)}(\bm S)$ and no non-collider is in $\bm S$. Therefore, the path is open in $\Aug(\Gcal_{\mid\bm C}, I_j)$ given $\bm S$ and $(I_j\not\indep X_k\mid\bm S)_{\Aug(\Gcal_{\mid\bm C}, I_j)}$. Similarly, we have $(I_k\not\indep X_j\mid\bm S)_{\Aug(\Gcal_{\mid\bm C}, I_k)}$.

    \noindent\textbf{2)} $\exists X_k\in\bm X\setminus\{X_j\}$ s.t. $(I_k\not\indep X_j\mid\bm S)_{\Aug(\Gcal_{\mid\bm C}, I_k)}$ and $(I_j\not\indep X_k\mid\bm S)_{\Aug(\Gcal_{\mid\bm C}, I_j)}\Longrightarrow X_j\in\ancestor_{\Gcal_{\mid\bm C}}(\bm S\setminus\bm S_{\rm is})$.

    Suppose $X_j\not\in\ancestor_{\Gcal_{\mid\bm C}}(\bm S)$. Since $X_j\in \bm A_k\subseteq\ancestor_{\Gcal_{\mid\bm C}}(X_k,\bm S)$, we know $X_j\in\ancestor_{\Gcal_{\mid\bm C}}(X_k)\setminus\ancestor_{\Gcal_{\mid\bm C}}(\bm S)$. Since $(I_k\not\indep X_j\mid\bm S)_{\Aug(\Gcal_{\mid\bm C}, I_k)}$, there exists an open path $\pi$ between $I_k,X_j$ given $\bm S$.

    If there is no collider on $\pi$, since $\pi$ starts with $I_k\rightarrow X_k$, we have $\pi$ equals $I_k\rightarrow X_k\rightarrow \ldots\rightarrow X_j$ with all arrowheads point to $X_j$. Then, $X_k\in\ancestor_{\Gcal_{\mid\bm C}}(X_j)\setminus\{X_j\}$ contradicts $X_j\in\ancestor_{\Gcal_{\mid\bm C}}(X_k)$. Consequently, there must be colliders on $\pi$. Denote $Z$ to be the one closest to $I_k$, we know $Z\in\ancestor_{\Gcal_{\mid\bm C}}(\bm S)$. Since there is no collider on the subpath of $\pi$ between $I_k,Z$, we know $X_k\in\ancestor_{\Gcal_{\mid\bm C}}(Z)\subseteq\ancestor_{\Gcal_{\mid\bm C}}(\bm S)$. Then, $X_j\in\ancestor_{\Gcal_{\mid\bm C}}(X_k)\subseteq\ancestor_{\Gcal_{\mid\bm C}}(\bm S)$, which contradicts $X_j\not\in \ancestor_{\Gcal_{\mid\bm C}}(\bm S)$. Therefore, we know $X_j\in \ancestor_{\Gcal_{\mid\bm C}}(\bm S)$.

    By Proposition \ref{prop:adjacency_latent_selection}, $X_k\in\bm A_j\subseteq\ancestor_{\Gcal_{\mid\bm C}}(X_j,\bm S)=\ancestor_{\Gcal_{\mid\bm C}}(\bm S)$. Since the subpath of $\pi$ between $X_j,X_k$ is open given $\bm S$ in $\Gcal_{\mid\bm C}$, it follows from Definition \ref{def:isolated_selection} that $X_j\in\ancestor_{\Gcal_{\mid\bm C}}(\bm S\setminus\bm S_{\rm is})$.
\end{proof}

\subsection{Proof of Lemma \ref{lem:pool}}\label{sec:proof:lem:pool}

\begin{proof}[Proof of Lemma \ref{lem:pool}]
    We prove by contradiction. To this end, we assume $\big(I_l\not\indep (X_j,X_k)\mid\bm A\cup\bm S\setminus\{X_j,X_k\}\big)_{\Aug(\Gcal_{\mid\bm C},I_l)}$.

    There exists an open path $\pi$ between $I_l,X_j$ or $I_l,X_k$ given $\bm A\cup\bm S\setminus\{X_j,X_k\}$ in $\Aug(\Gcal_{\mid\bm C},I_l)$. Without the loss of generality, we assume $\pi$ connects $I_l,X_j$. Then no non-collider on $\pi$ is in $\bm S$.
    
    Since $\pi$ must start with $I_l\rightarrow X_l$, if there is no collider on $\pi$, $\pi$ must be a directed path $I_l\rightarrow X_l\rightarrow\ldots\rightarrow X_j$ with all arrowheads pointing to $X_j$. Then $X_l\in\ancestor_{\Gcal_{\mid\bm C}}(X_j)\setminus(X_j,\bm L)$. It follows from Proposition \ref{prop:adjacency_latent_selection} that $X_l\in\bm A_j$ which contradicts the assumption $X_l\not\in\bm A_j\cup\bm A_k\cup\{X_j,X_k\}$.

    If there are colliders on $\pi$, since $\pi$ is open given $\bm A\cup\bm S\setminus\{X_j,X_k\}$, the colliders must be in $\ancestor_{\Gcal_{\mid\bm C}}(\bm A\cup\bm S\setminus\{X_j,X_k\})\subseteq\ancestor_{\Gcal_{\mid\bm C}}(X_j,X_k,\bm S)$. 
    
    If there are colliders in $\ancestor_{\Gcal_{\mid\bm C}}(X_j,X_k)\setminus\ancestor_{\Gcal_{\mid\bm C}}(\bm S)$, denote $Z$ to be the one on $\pi$ closest to $I_l$. Denote $\pi_Z$ to be the subpath of $\pi$ between $I_l,Z$. Then all colliders on $\pi_Z$ must be in $\ancestor_{\Gcal_{\mid\bm C}}(\bm S)$. Since $Z\in\ancestor_{\Gcal_{\mid\bm C}}(X_j,X_k)\setminus\ancestor_{\Gcal_{\mid\bm C}}(\bm S)$, there exists a directed path $\pi'$ of the form $Z\rightarrow\ldots\rightarrow X_j$ or $Z\rightarrow\ldots\rightarrow X_k$, and $\bm S$ doesn't intersect $\pi'$. Then the concatenation of $\pi_Z$ and $\pi'$ has no non-collider in $\bm S$ and all colliders in $\ancestor_{\Gcal_{\mid\bm C}}(\bm S)$. Thus, the concatenated walk is open given $\bm S$ in $\Aug(\Gcal_{\mid\bm C},I_l)$, and then $X_l\in\bm A_j\cup\bm A_k$ which contradicts the assumption.

    Then we know all colliders on $\pi$ are in $\ancestor_{\Gcal_{\mid\bm C}}(\bm S)$. Thus, $\pi$ is open given $\bm S$ in $\Aug(\Gcal_{\mid\bm C},I_l)$, and then $X_l\in\bm A_j$ which contradicts the assumption.

    Combining pieces concludes the proof.
\end{proof}

\subsection{Proof of Proposition \ref{prop:consistent_anterior_latent_selection}}\label{sec:proof:prop:consistent_anterior_latent_selection}

\begin{lemma}\label{lem:direct_ancestor}
    For any $X_j\in\bm X$ and $X_k\in\ancestor_{\Gcal_{\mid\bm C}}(X_j)\setminus\{X_j\}\cap\bm X$, there exists a sequence $(X_{j_1},\ldots,X_{j_m})\subseteq\bm X\setminus\{X_j,X_k\}$ with $X_{j_0}\overset{\triangle}{=}X_k$ and $X_{j_{m+1}}\overset{\triangle}{=}X_j$, such that for any $l=0,\ldots, m$, we have $X_{j_l}\in\ancestor_{\Gcal_{\mid\bm C}}(X_{j_{l+1}})\setminus\{X_{j_{l+1}}\}$ and there is no directed path $X_{j_l}\rightarrow\ldots\rightarrow X_{j_{l+1}}$ intersecting with $\bm X\setminus\{X_{j_l},X_{j_{l+1}}\}$.
\end{lemma}

\begin{proof}[Proof of Lemma \ref{lem:direct_ancestor}]
    Since $X_k\in\ancestor_{\Gcal_{\mid\bm C}}(X_j)\setminus\{X_j\}$, there exist directed paths between $X_j,X_k$ with all arrowheads pointing to $X_j$. Let $\pi$ denote such a path containing the largest number of nodes in $\bm X$. We show that all the observed variables $(X_{j_1},\ldots,X_{j_m})\subseteq\bm X\setminus\{X_j,X_k\}$ on $\pi:X_{j_0}\overset{\triangle}{=}X_k\rightarrow\ldots\rightarrow X_{j_1}\rightarrow\ldots\rightarrow X_{j_2}\rightarrow\ldots\rightarrow X_{j_{m-1}}\rightarrow\ldots\rightarrow X_{j_m}\rightarrow\ldots\rightarrow X_{j_{m+1}}\overset{\triangle}{=}X_j$ form such a sequence.

    It is evident that $X_{j_l}\in\ancestor_{\Gcal_{\mid\bm C}}(X_{j_{l+1}})\setminus\{X_{j_{l+1}}\}$ for $l=0,\ldots,m$. If there exists $l\in\{0,\ldots,m\}$ and directed path $\pi_l$ $X_{j_l}\rightarrow \ldots\rightarrow X'\rightarrow \ldots\rightarrow X_{j_{l+1}}$ such that $X'\in\bm X\setminus\{X_{j_l},X_{j_{l+1}}\}$. We first show that $X'\in\bm X\setminus\{X_{j_l}:l=0,\ldots,m+1\}$. To see this, if $X'=X_{j_i}$ for $i<l$, we have $X'=X_{j_i}\in\ancestor_{\Gcal_{\mid\bm C}}(X_{j_l})\setminus\{X_{j_l}\}$ and $X_{j_l}\in\ancestor_{\Gcal_{\mid\bm C}}(X')\setminus\{X'\}$, which contradicts the acyclic assumption of $\Gcal_{\mid\bm C}$. If $X'=X_{j_i}$ for $i>l+1$, we have $X'\in\ancestor_{\Gcal_{\mid\bm C}}(X_{j_{l+1}})\setminus\{X_{j_{l+1}}\}$ and $X_{j_{l+1}}\in\ancestor_{\Gcal_{\mid\bm C}}(X_{j_i})\setminus\{X_{j_i}\}=\ancestor_{\Gcal_{\mid\bm C}}(X')\setminus\{X'\}$, which again contradicts the acyclic assumption of $\Gcal_{\mid\bm C}$.

    Since $X'\in\bm X\setminus\{X_{j_l}:l=0,\ldots,m+1\}$, we can replace the subpath of $\pi$ between $X_{j_l},X_{j_{l+1}}$ by $\pi_l$ to get another directed path between $X_j,X_k$ with all arrowheads pointing to $X_j$. The new path contains at least $m+1$ middle nodes in $\bm X$, which contradicts the definition of $\pi$. Therefore, for any $l=0,\ldots, m$, there is no directed path $X_{j_l}\rightarrow\ldots\rightarrow X_{j_{l+1}}$ intersecting with $\bm X\setminus\{X_{j_l},X_{j_{l+1}}\}$.
\end{proof}

\begin{lemma}\label{lem:direct_ancestor_S}
    For any $X_j\in\ancestor_{\Gcal_{\mid\bm C}}(\bm S\setminus\bm S_{\rm is})\cap\bm X$, there exists a sequence $(X_{j_1},\ldots,X_{j_m})\subseteq\bm X\setminus\{X_j\}$ with $X_{j_0}\overset{\triangle}{=}X_j$ such that 1) for any $l=0,\ldots, m-1$, we have $X_{j_l}\in\ancestor_{\Gcal_{\mid\bm C}}(X_{j_{l+1}})\setminus\{X_{j_{l+1}}\}$ and there is no directed path $X_{j_l}\rightarrow\ldots\rightarrow X_{j_{l+1}}$ intersecting with $\bm X\setminus\{X_{j_l},X_{j_{l+1}}\}$, and 2) $X_{j_m}\in\ancestor_{\Gcal_{\mid\bm C}}(\bm S\setminus\bm S_{\rm is})$, but there is no $S\in\bm S\setminus\bm S_{\rm is}$ and directed path $X_{j_m}\rightarrow\ldots\rightarrow S$ intersecting with $\bm X\setminus\{X_{j_m}\}$.
\end{lemma}

\begin{proof}[Proof of Lemma \ref{lem:direct_ancestor_S}]
    Since $X_j\in\ancestor_{\Gcal_{\mid\bm C}}(\bm S\setminus\bm S_{\rm is})$, there exist directed paths between $X_j$ and some $S\in\bm S\setminus\bm S_{\rm is}$ with all arrowheads pointing to $S$. Let $\pi$ denote such a path containing the largest number of nodes in $\bm X$. We show that all the observed variables $(X_{j_1},\ldots,X_{j_m})\subseteq\bm X\setminus\{X_j\}$ on $\pi:X_{j_0}\overset{\triangle}{=}X_j\rightarrow\ldots\rightarrow X_{j_1}\rightarrow\ldots\rightarrow X_{j_2}\rightarrow\ldots\rightarrow X_{j_{m-1}}\rightarrow\ldots\rightarrow X_{j_m}\rightarrow\ldots\rightarrow S$ form such a sequence.

    Similar to the proof of Lemma \ref{lem:direct_ancestor}, for any $l=0,\ldots,m-1$, we have there is no directed path $X_{j_l}\rightarrow \ldots\rightarrow X_{j_{l+1}}$ intersecting with $\bm X\setminus\{X_{j_l},X_{j_{l+1}}\}$. Then we show there is no $S'\in\bm S\setminus\bm S_{\rm is}$ and directed path $X_{j_m}\rightarrow\ldots\rightarrow S'$ intersecting with $\bm X\setminus\{X_{j_m}\}$. If there exists such a path $\pi_m$ $X_{j_m}\rightarrow\ldots\rightarrow X'\rightarrow\ldots\rightarrow S'$ such that $X'\in\bm X\setminus\{X_{j_m}\}$, similar to the proof of Lemma \ref{lem:direct_ancestor}, we know $X'\in\bm X\setminus\{X_{j_l}:l=0,\ldots,m\}$. Then we can replace the subpath of $\pi$ between $X_{j_m},S$ by $\pi_m$ to get a new path between $X_j,S'$ with all arrowheads pointing to $S'$. The new path contains at least $m+1$ middle nodes in $\bm X$, which contradicts the definition of $\pi$. Therefore, $(X_{j_1},\ldots,X_{j_m})$ satisfies the statement of Lemma \ref{lem:direct_ancestor_S}.

\end{proof}

\begin{lemma}\label{lem:hat_A_S}
    Under Assumption \ref{asm:faithful_descendant_latent_selection}, if all anterior tests are correct, we have $\bm A_{\bm S}\subseteq\widehat{\bm A}_{\bm S}$.
\end{lemma}

\begin{proof}[Proof of Lemma \ref{lem:hat_A_S}]
    For any $X_j\in\ancestor_{\Gcal_{\mid\bm C}}(\bm S\setminus\bm S_{\rm is})\cap\bm X$, it follows from Lemma \ref{lem:direct_ancestor_S} that there exists a sequence $(X_{j_1},\ldots,X_{j_m})$ satisfying the statement in Lemma \ref{lem:direct_ancestor_S}. 
    
    If $m>0$, it follows from Assumption \ref{asm:faithful_descendant_latent_selection} that $\cap_{l=0}^{m-1}H_{j_{l+1},1}^{\anterior(j_l)}\cap H_{j_{m-1},1}^{\anterior(j_m)}$ holds. Therefore, $X_{j_m}\in\widehat{\bm A}_{\bm S}$, $X_j\in\widehat{\bm A}_{j_m}$, and then $X_j\in\widehat{\bm A}_{\bm S}$.

    If $m=0$, since $X_j\in\ancestor_{\Gcal_{\mid\bm C}}(\bm S\setminus\bm S_{\rm is})$, there exists $X_k\in\ancestor_{\Gcal_{\mid\bm C}}(\bm S)\cap\bm X\setminus\{X_j\}$ such that $(X_j\not\indep X_k\mid\bm S)_{\Gcal_{\mid\bm S}}$. Denote $\pi$ to be an open path between $X_j,X_k$ given $\bm S$. It follows from Lemma 3.13 in \cite{richardson2002ancestral} that all middle notes are in $\ancestor_{\Gcal_{\mid\bm C}}(\bm S)$. Denote $X_l$ to be the observed node on $\pi\setminus\{X_j\}$ closest to $X_j$, it follows from Assumption \ref{asm:faithful_descendant_latent_selection} that $H_{l,1}^{\anterior(j)}$ holds, and thus $X_j\in\widehat{\bm A}_l$. Since $X_l\in\ancestor_{\Gcal_{\mid\bm C}}(\bm S)$, it follows from Lemma \ref{lem:direct_ancestor_S} that there exists a sequence $(X_{j_1'},\ldots,X_{j_{m'}'})$ with $X_{j_0'}\overset{\triangle}{=}X_l$ satisfying the statement in Lemma \ref{lem:direct_ancestor_S}. 
    
    If $m'>0$, it follows from Assumption \ref{asm:faithful_descendant_latent_selection} that $\cap_{i=0}^{m'-1}H_{j_{i+1}',1}^{\anterior(j_i')}\cap H_{j_{m'-1}',1}^{\anterior(j_{m'}')}$ holds. Therefore, $X_{j_{m'}'}\in\widehat{\bm A}_{\bm S}$, $X_l\in\widehat{\bm A}_{j_{m'}'}$, and then $X_l\in\widehat{\bm A}_{\bm S}$. Since $X_j\in\widehat{\bm A}_l$, we know $X_j\in\widehat{\bm A}_{\bm S}$.

    If $m'=0$, it follows from Assumption \ref{asm:faithful_descendant_latent_selection} that $H_{j,1}^{\anterior(l)}$ holds. Therefore, $X_l\in\widehat{\bm A}_j$, $X_j\in\widehat{\bm A}_l$, and thus $X_j\in\widehat{\bm A}_{\bm S}$.
    
\end{proof}

\begin{lemma}\label{lem:hat_A_j}
    Under Assumption \ref{asm:faithful_descendant_latent_selection}, if all anterior tests are correct, then $\bm A_j\subseteq\widehat{\bm A}_j$ for all $j\in[d_X]$.
\end{lemma}

\begin{proof}[Proof of Lemma \ref{lem:hat_A_j}]

    \noindent\textbf{1)} First, we show $(\ancestor_{\Gcal_{\mid\bm C}}(X_j)\cap\bm X)\setminus\{X_j\}\subseteq\widehat{\bm A}_j$ for all $j\in[d_X]$.

    For any $X_k\in(\ancestor_{\Gcal_{\mid\bm C}}(X_j)\cap\bm X)\setminus\{X_j\}$, Lemma \ref{lem:direct_ancestor} provides a sequence $X_{j_0}=X_k,\ldots,X_{j_m}=X_j\in\bm X$ such that for $l\in[m]$, $X_{j_{l-1}}\in\ancestor_{\Gcal_{\mid\bm C}}(X_{j_l})\setminus\{X_{j_l}\}$ and no directed path from $X_{j_{l-1}}$ to $X_{j_l}$ contains another observed node. Assumption \ref{asm:faithful_descendant_latent_selection} part 1) and correctness of anterior tests imply that $\psi_{j_l}^{\anterior(j_{l-1})}=1$ for $l\in[m]$. Step 2 in Algorithm \ref{alg:conf_MAG_latent_selection} then gives that $X_k\in\widehat{\bm A}_j$. Therefore, $(\ancestor_{\Gcal_{\mid\bm C}}(X_j)\cap\bm X)\setminus\{X_j\}\subseteq\widehat{\bm A}_j$.

    \noindent\textbf{2)} Then we prove $\bm A_j\subseteq\widehat{\bm A}_j$ by contradiction.

    Suppose $X_k\in\bm A_j\setminus\widehat{\bm A}_j$. By the definition of $\bm A_j$, there exists a path $\pi$ between $I_k,X_j$ that is open given $\bm S$ in $\Aug(\Gcal_{\mid\bm C},\bm I_k)$. This path must start with $I_k\rightarrow X_k$. By Lemma 3.13 in \cite{richardson2002ancestral}, all middle nodes on $\pi$ are in $\ancestor_{\Aug(\Gcal_{\mid\bm C},I_k)}(X_j,\bm S)$. Since $I_k$ cannot be a middle node of any directed path in $\Aug(\Gcal_{\mid\bm C},I_k)$, all middle nodes on $\pi$ are in $\ancestor_{\Gcal_{\mid\bm C}}(X_j,\bm S)$. 

    Denote $X_l$ on $\pi$ to be the observed node in $\widehat{\bm A}_j\cup\{X_j\}$ closest to $X_k$ and $X_{l'}$ to be the observed node on the subpath of $\pi$ between $I_k,X_l$ closest to $X_l$. Then we know $X_{l'}\not\in\widehat{\bm A}_j\cup\{X_j\}$ and $X_l,X_{l'}$ are well defined since $X_j\in\widehat{\bm A}_j\cup\{X_j\},X_k\not\in \widehat{\bm A}_j\cup\{X_j\}$. The subpath of $\pi$ between $X_l,X_{l'}$ is open given $\bm S$, and contains no other observed nodes.

    Since $X_{l'}\not\in\widehat{\bm A}_j$, we know from \textbf{1)} that $X_{l'}\not\in\ancestor_{\Gcal_{\mid\bm C}}(X_j)$. Then $X_{l'}\in\ancestor_{\Gcal_{\mid\bm C}}(\bm S)$.

    If $X_l\not\in\ancestor_{\Gcal_{\mid\bm C}}(X_{l'})$, Assumption \ref{asm:faithful_descendant_latent_selection} part 2) and the correctness of anterior tests imply $\psi_l^{\anterior(l')}=1$. Then $X_{l'}\in\widehat{\bm A}_l$. Since $X_l\in\widehat{\bm A}_j\cup\{X_j\}$, Step 3 in Algorithm \ref{alg:conf_MAG_latent_selection} implies $X_{l'}\in\widehat{\bm A}_j$, which contradicts the definition of $X_{l'}$.

    If $X_l\in\ancestor_{\Gcal_{\mid\bm C}}(X_{l'})$, part \textbf{1)} gives $X_l\in\widehat{\bm A}_{l'}$. Since $X_{l'}\in\ancestor_{\Gcal_{\mid\bm C}}(\bm S)$, choose $S\in\bm S$ such that $X_{l'}\in\ancestor_{\Gcal_{\mid\bm C}}(S)$, then $X_l,X_{l'}\in\ancestor_{\Gcal_{\mid\bm C}}(S)$ and $X_{l'}\in\bm A_{\bm S}$. By Lemma \ref{lem:hat_A_S}, $X_{l'}\in\widehat{\bm A}_{\bm S}$. $X_l\in\widehat{\bm A}_{l'}$ and Step 3 of Algorithm \ref{alg:conf_MAG_latent_selection} imply that $X_{l'}\in\widehat{\bm A}_l$. Since $X_l\in\widehat{\bm A}_j\cup\{X_j\}$, Step 3 in Algorithm \ref{alg:conf_MAG_latent_selection} implies $X_{l'}\in\widehat{\bm A}_j$, which contradicts the definition of $X_{l'}$.

    In conclusion, we have $\bm A_j\subseteq\widehat{\bm A}_j$.
\end{proof}

\begin{lemma}\label{lem:A_closure}
    \begin{enumerate}[1)]
        \item For any $X_k\in\bm A_{\bm S}$ and $X_j\in\bm A_k$, we have $X_j\in\bm A_{\bm S}$ and $X_k\in\bm A_j$.
        \item For $j\ne l$, if $X_j\in\bm A_k$ and $X_k\in\bm A_l$, then $X_j\in\bm A_l$.
    \end{enumerate}
\end{lemma}

\begin{proof}[Proof of Lemma \ref{lem:A_closure}]
    \noindent\textbf{1)} Since $X_k\in\bm A_{\bm S}=\ancestor_{\Gcal_{\mid\bm C}}(\bm S\setminus\bm S_{\rm is})$, we have $X_j\in\bm A_k\subseteq\ancestor_{\Gcal_{\mid\bm C}}(X_k,\bm S)\subseteq\ancestor_{\Gcal_{\mid\bm C}}(\bm S)$. Since $X_j$ is $d$-connected to $X_k\in\ancestor_{\Gcal_{\mid\bm C}}(\bm S\setminus\bm S_{\rm is})$, by Definition \ref{def:isolated_selection}, we have $X_j\in\ancestor_{\Gcal_{\mid\bm C}}(\bm S\setminus\bm S_{\rm is})=\bm A_{\bm S}$. 
    
    Since $X_j\in\bm A_k$, there exists an open path $\pi$ between $I_j,X_k$ given $\bm S$ in $\Aug(\Gcal_{\mid\bm C},I_j)$. Since $I_j$ is only adjacent to $X_j$ in $\Aug(\Gcal_{\mid\bm C},I_j)$, we know $X_j$ is on $\pi$. Denote $\pi_j$ to be the subpath of $\pi$ between $X_j,X_k$. Since $X_k\in\bm A_{\bm S}$, we know the concatenated path of $I_k\rightarrow X_k$ and $\pi_j$ is open given $\bm S$ in $\Aug(\Gcal_{\mid\bm C},I_k)$, regardless of whether $X_k$ is a collider. Therefore, $X_k\in\bm A_j$.

    \noindent\textbf{2)} By Proposition \ref{prop:adjacency_latent_selection}, we know $\bm A_j=\anterior_{\MAG(\Gcal_{\mid\bm C})}(X_j)\setminus\{X_j\}$. Then Proposition 2.1 in \cite{richardson2002ancestral} implies the desired result.
\end{proof}

\begin{proof}[Proof of Proposition \ref{prop:consistent_anterior_latent_selection}]
    Under Assumptions \ref{asm:context}, suppose there is no false positive among $\big\{\psi_j^{\anterior(k)}:k\in[d_X],j\in[d_X]\setminus\{k\}\big\}$. Then we have $\psi_j^{\anterior(k)}=1$ only if $H_{j,1}^{\anterior(k)}$ holds. Since the data-generating distributions are Markov to $\Gcal$, $H_{j,1}^{\anterior(k)}$ holds only if $(I_k\not\indep X_j\mid\bm S)_{\Aug(\Gcal_{\mid\bm C},I_k)}$. Then $\psi_j^{\anterior(k)}=1$ implies $X_k\in\bm A_j$. By Lemma \ref{lem:A_closure}, we have that $\psi_j^{\anterior(j_1)}\prod_{l\in[m-1]}\psi_{j_l}^{\anterior(j_{l+1})}\psi_{j_m}^{\anterior(k)}=1$ implies $X_k\in\bm A_j$. Moreover, $\psi_j^{\anterior(k)}\psi_k^{\anterior(j)}=1$ implies $X_j\in\bm A_k$, $X_k\in\bm A_j$, which then follows from Lemma \ref{lem:S} that $X_j,X_k\in\bm A_{\bm S}$. Consequently, the initial sets in Step 2 of Algorithm \ref{alg:conf_MAG_latent_selection} satisfy $\widehat{\bm A}_j\subseteq\bm A_j$ and $\widehat{\bm A}_{\bm S}\subseteq\bm A_{\bm S}$. 
    
    By Lemma \ref{lem:A_closure}, the updated $\widehat{\bm A}_j$ and $\widehat{\bm A}_{\bm S}$ in Step 3 of Algorithm \ref{alg:conf_MAG_latent_selection} still satisfy $\widehat{\bm A}_j\subseteq\bm A_j$ and $\widehat{\bm A}_{\bm S}\subseteq\bm A_{\bm S}$. Therefore, we have
    \begin{align*}
        \Prob\big(\exists j\in[d_X]{\rm ~s.t.~}\widehat{\bm A}_j\not\subseteq\bm A_j{\rm ~or~}\widehat{\bm A}_{\bm S}\not\subseteq\bm A_{\bm S}\big)\le\Prob\big(\text{exists false positive anterior test}\big)\le\frac{\alpha}{2}.
    \end{align*}

    Under Assumptions \ref{asm:context} and \ref{asm:faithful_descendant_latent_selection}, if all anterior tests are correct, Lemmas \ref{lem:hat_A_j} and \ref{lem:hat_A_S} imply that $\bm A_j=\widehat{\bm A}_j$ and $\bm A_{\bm S}=\widehat{\bm A}_{\bm S}$. Then,
    \begin{align*}
        &\Prob\big(\exists j\in[d_X]{\rm ~s.t.~}\widehat{\bm A}_j\ne\bm A_j{\rm ~or~}\widehat{\bm A}_{\bm S}\ne\bm A_{\bm S}\big)\\
        \le&\Prob\big(\text{exists false positive anterior test}\big)+\Prob\big(\text{exists false negative anterior test}\big)\\
        \le&\frac{\alpha}{2}+{\rm FWER}_2^\anterior.
    \end{align*}
\end{proof}

\subsection{Proof of Theorem \ref{thm:coverage_latent_selection}}\label{sec:proof:thm:coverage_latent_selection}

\begin{proof}[Proof of Theorem \ref{thm:coverage_latent_selection}]
    If all the anterior tests are correct, it follows from Proposition \ref{prop:consistent_anterior_latent_selection} that $\bm A_j=\widehat{\bm A}_j$ and $\bm A_{\bm S}=\widehat{\bm A}_{\bm S}$. If there is no false positive among the adjacency tests, $\psi_{j,k}^{\rm Adj}(\bm A_{\bm S},\bm A_{[d_X]})=1$ implies $H_{j,k,1}^{\rm Adj}(\bm A_{\bm S},\bm A_{[d_X]})$ holds and $\psi_j^{{\rm Adj}(k)}(\bm A_{\bm S},\bm A_{[d_X]})=1$ implies $H_{j,1}^{{\rm Adj}(k)}(\bm A_{\bm S},\bm A_{[d_X]})$ holds. Since the data-generating distributions are Markov to $\Gcal$, $H_{j,k,1}^{\rm Adj}(\bm A_{\bm S},\bm A_{[d_X]})$ implies $\big(X_j\not\indep X_k\mid\bm A_j\cup\bm A_k\cup\bm A_{\bm S}\cup\bm S\setminus\{X_j,X_k\}\big)_{\Gcal_{\mid\bm C}}$ and $H_{j,1}^{{\rm Adj}(k)}(\bm A_{\bm S},\bm A_{[d_X]})$ implies $\big(I_k\not\indep X_j\mid\bm A_j\cup\bm A_{\bm S}\cup\bm S\setminus\{X_j\}\big)_{\Aug(\Gcal_{\mid\bm C},I_k)}$. It then follows from Proposition \ref{prop:adjacency_latent_selection} and \ref{prop:direction_latent_selection} that all the detected edges in $\Ccal_\alpha$ exist in $\Ecal\big(\TCM(\Gcal_{\mid\bm C},\bm I)\big)$.
    Consequently, we have
    \begin{align*}
        &\Prob\big(\Ccal_\alpha\not\subseteq \Ecal(\TCM(\Gcal_{\mid\bm C},\bm I))\big)\\
        \le&\Prob\big(\Ccal_\alpha\not\subseteq \Ecal(\TCM(\Gcal_{\mid\bm C},\bm I)), \text{ all anterior tests are correct, and no false positive adjacency test}\big)\\
        &+\Prob\big(\text{all anterior tests are correct, but exist false positive adjacency tests}\big)\\
        &+\Prob\big(\text{exist false positive anterior test}\big)+\Prob\big(\text{exist false negative anterior tests}\big)\\
        \le&0+\frac{\alpha}{2}+\frac{\alpha}{2}+\text{FWER}_2^\anterior\\
        =&\alpha+\text{FWER}_2^\anterior.
    \end{align*}
    For the second inequality, we use the fact that tests $\psi_{j,k}^{\rm Adj}(\widehat{\bm A}_{\bm S},\widehat{\bm A}_{[d_X]})$ and $\psi_{j}^{{\rm Adj}(k)}(\widehat{\bm A}_{\bm S},\widehat{\bm A}_{[d_X]})$ are conducted pretending $\widehat{\bm A}_{\bm S},\widehat{\bm A}_{[d_X]}$ are fixed, and thus, 
    \[\psi_{j,k}^{\rm Adj}(\widehat{\bm A}_{\bm S},\widehat{\bm A}_{[d_X]})=\psi_{j,k}^{\rm Adj}(\bm A_{\bm S},\bm A_{[d_X]}), \quad \psi_{j}^{{\rm Adj}(k)}(\widehat{\bm A}_{\bm S},\widehat{\bm A}_{[d_X]})=\psi_{j}^{{\rm Adj}(k)}(\bm A_{\bm S},\bm A_{[d_X]}),\]
    if $\widehat{\bm A}_{\bm S}=\bm A_{\bm S}$ and $\widehat{\bm A}_{[d_X]}=\bm A_{[d_X]}$. Then
    \begin{align*}
        &\Prob\big(\text{all anterior tests are correct, but exist false positive adjacency tests}\big)\\
        \le&\Prob\big(\widehat{\bm A}_{\bm S}=\bm A_{\bm S},\widehat{\bm A}_{[d_X]}=\bm A_{[d_X]}, \text{ exist false positive in }\\
        &\quad\big\{\psi_{j,k}^{\rm Adj}\big(\widehat{\bm A}_{\bm S},\widehat{\bm A}_{[d_X]}\big):j,k\in[d_X],j>k\big\}\cup\big\{\psi_j^{{\rm Adj}(k)}\big(\widehat{\bm A}_{\bm S},\widehat{\bm A}_{[d_X]}\big):j\in[d_X],k\in\widehat{\bm A}_j\big\}\big)\\
        =&\Prob\big(\widehat{\bm A}_{\bm S}=\bm A_{\bm S},\widehat{\bm A}_{[d_X]}=\bm A_{[d_X]}, \text{ exist false positive in }\\
        &\quad\big\{\psi_{j,k}^{\rm Adj}\big(\bm A_{\bm S},\bm A_{[d_X]}\big):j,k\in[d_X],j>k\big\}\cup\big\{\psi_j^{{\rm Adj}(k)}\big(\bm A_{\bm S},\bm A_{[d_X]}\big):j\in[d_X],k\in\bm A_j\big\}\big)\\
        \le&\Prob\big(\text{exist false positive in }\\
        &\quad\big\{\psi_{j,k}^{\rm Adj}\big(\bm A_{\bm S},\bm A_{[d_X]}\big):j,k\in[d_X],j>k\big\}\cup\big\{\psi_j^{{\rm Adj}(k)}\big(\bm A_{\bm S},\bm A_{[d_X]}\big):j\in[d_X],k\in\bm A_j\big\}\big)\\
        \le&\frac{\alpha}{2}.
    \end{align*}
\end{proof}

\subsection{Proof of Theorem \ref{thm:power_latent_selection}}\label{sec:proof:thm:power_latent_selection}

\begin{proof}[Proof of Theorem \ref{thm:power_latent_selection}]
    Similar to the proof of Theorem \ref{thm:coverage_latent_selection}, under Assumptions \ref{asm:context} and \ref{asm:faithful_descendant_latent_selection}, if all anterior tests are correct and there is no false positive in the adjacency tests, $\Ccal_\alpha\subseteq\Ecal(\TCM(\Gcal_{\mid\bm C},\bm I))$. If Assumption \ref{asm:faithful_adjacency_latent_selection} also holds, $d$-connection implies conditional dependence. Then each edge in $\TCM(\Gcal_{\mid\bm C},\bm I)$ corresponds to an alternative adjacency hypothesis. If there is no false negative in the adjacency tests, we detect all the correct alternative hypotheses and thus all the edges in $\TCM(\Gcal_{\mid\bm C},\bm I)$. Consequently,
    \begin{align*}
        &\Prob\big(\Ccal_\alpha\ne \Ecal(\TCM(\Gcal_{\mid\bm C},\bm I))\big)\\
        \le&\Prob\big(\Ccal_\alpha\ne \Ecal(\TCM(\Gcal_{\mid\bm C},\bm I)), \text{ all anterior tests are correct, and all adjacency tests are correct}\big)\\
        &+\Prob\big(\text{all anterior tests are correct, but exist false positive adjacency tests}\big)\\
        &+\Prob\big(\text{all anterior tests are correct, but exist false negative adjacency tests}\big)\\
        &+\Prob\big(\text{exist false positive anterior tests}\big)+\Prob\big(\text{exist false negative anterior tests}\big)\\
        \le&0+\frac{\alpha}{2}+\text{FWER}_2^{\rm Adj}+\frac{\alpha}{2}+\text{FWER}_2^{\anterior}.
    \end{align*}
    For the last inequality, we use the fact that tests $\psi_{j,k}^{\rm Adj}(\widehat{\bm A}_{\bm S},\widehat{\bm A}_{[d_X]})$ and $\psi_{j}^{{\rm Adj}(k)}(\widehat{\bm A}_{\bm S},\widehat{\bm A}_{[d_X]})$ are conducted pretending $\widehat{\bm A}_{\bm S},\widehat{\bm A}_{[d_X]}$ are fixed, and thus, 
    \[\psi_{j,k}^{\rm Adj}(\widehat{\bm A}_{\bm S},\widehat{\bm A}_{[d_X]})=\psi_{j,k}^{\rm Adj}(\bm A_{\bm S},\bm A_{[d_X]}), \quad \psi_{j}^{{\rm Adj}(k)}(\widehat{\bm A}_{\bm S},\widehat{\bm A}_{[d_X]})=\psi_{j}^{{\rm Adj}(k)}(\bm A_{\bm S},\bm A_{[d_X]}),\]
    if $\widehat{\bm A}_{\bm S}=\bm A_{\bm S}$ and $\widehat{\bm A}_{[d_X]}=\bm A_{[d_X]}$. Then
    \begin{align*}
        &\Prob\big(\text{all anterior tests are correct, but exist false positive adjacency tests}\big)\\
        \le&\Prob\big(\widehat{\bm A}_{\bm S}=\bm A_{\bm S},\widehat{\bm A}_{[d_X]}=\bm A_{[d_X]}, \text{ exist false positive in }\\
        &\quad\big\{\psi_{j,k}^{\rm Adj}\big(\widehat{\bm A}_{\bm S},\widehat{\bm A}_{[d_X]}\big):j,k\in[d_X],j>k\big\}\cup\big\{\psi_j^{{\rm Adj}(k)}\big(\widehat{\bm A}_{\bm S},\widehat{\bm A}_{[d_X]}\big):j\in[d_X],k\in\widehat{\bm A}_j\big\}\big)\\
        =&\Prob\big(\widehat{\bm A}_{\bm S}=\bm A_{\bm S},\widehat{\bm A}_{[d_X]}=\bm A_{[d_X]}, \text{ exist false positive in }\\
        &\quad\big\{\psi_{j,k}^{\rm Adj}\big(\bm A_{\bm S},\bm A_{[d_X]}\big):j,k\in[d_X],j>k\big\}\cup\big\{\psi_j^{{\rm Adj}(k)}\big(\bm A_{\bm S},\bm A_{[d_X]}\big):j\in[d_X],k\in\bm A_j\big\}\big)\\
        \le&\Prob\big(\text{exist false positive in }\\
        &\quad\big\{\psi_{j,k}^{\rm Adj}\big(\bm A_{\bm S},\bm A_{[d_X]}\big):j,k\in[d_X],j>k\big\}\cup\big\{\psi_j^{{\rm Adj}(k)}\big(\bm A_{\bm S},\bm A_{[d_X]}\big):j\in[d_X],k\in\bm A_j\big\}\big)\\
        \le&\frac{\alpha}{2}.
    \end{align*}
    Likewise, we have
    \begin{align*}
        &\Prob\big(\text{all anterior tests are correct, but exist false negative adjacency tests}\big)\\
        \le&\Prob\big(\widehat{\bm A}_{\bm S}=\bm A_{\bm S},\widehat{\bm A}_{[d_X]}=\bm A_{[d_X]}, \text{ exist false negative in }\\
        &\quad\big\{\psi_{j,k}^{\rm Adj}\big(\widehat{\bm A}_{\bm S},\widehat{\bm A}_{[d_X]}\big):j,k\in[d_X],j>k\big\}\cup\big\{\psi_j^{{\rm Adj}(k)}\big(\widehat{\bm A}_{\bm S},\widehat{\bm A}_{[d_X]}\big):j\in[d_X],k\in\widehat{\bm A}_j\big\}\big)\\
        =&\Prob\big(\widehat{\bm A}_{\bm S}=\bm A_{\bm S},\widehat{\bm A}_{[d_X]}=\bm A_{[d_X]}, \text{ exist false negative in }\\
        &\quad\big\{\psi_{j,k}^{\rm Adj}\big(\bm A_{\bm S},\bm A_{[d_X]}\big):j,k\in[d_X],j>k\big\}\cup\big\{\psi_j^{{\rm Adj}(k)}\big(\bm A_{\bm S},\bm A_{[d_X]}\big):j\in[d_X],k\in\bm A_j\big\}\big)\\
        \le&\Prob\big(\text{exist false negative in }\\
        &\quad\big\{\psi_{j,k}^{\rm Adj}\big(\bm A_{\bm S},\bm A_{[d_X]}\big):j,k\in[d_X],j>k\big\}\cup\big\{\psi_j^{{\rm Adj}(k)}\big(\bm A_{\bm S},\bm A_{[d_X]}\big):j\in[d_X],k\in\bm A_j\big\}\big)\\
        \le&{\rm FWER}_2^{\rm Adj}.
    \end{align*}
\end{proof}

\subsection{Proof of Theorem \ref{thm:lower_bound}}\label{sec:proof:thm:lower_bound}

\begin{proof}[Proof of Theorem \ref{thm:lower_bound}]
    Consider a class of DAGs on $\bm X$ with causal ordering $(X_1,\ldots,X_{d_X})$. For any $j<k\in[d_X]$, the edge $X_j\rightarrow X_k$ may be either present or absent in a DAG from the class. Then, the class contains $2^{\frac{1}{2}d_X(d_X-1)}$ DAGs.

    For any $m$-constraint-based algorithm $\Acal$ and any distributions class $\big\{\Pcal_\Gcal:\text{DAG } \Gcal\big\}$, we look at the response vectors $\Psi_{\Gcal}=(\psi_{j,\Gcal}:j\in[m])\in\{0,1\}^m$ of $\Acal(\Pcal_\Gcal)$. By the pigeonhole principle, if $m<\frac{1}{2}d_X(d_X-1)$, there exists $\Gcal\ne\Gcal'$ such that $\Psi_{\Gcal}=\Psi_{\Gcal'}$. Since $\Acal$ is deterministic, it follows from Definition \ref{def:algorithm} that equal response vectors $\Psi_{\Gcal}=\Psi_{\Gcal'}$ implies that the constraint vectors are also equal $(c_{j,\Gcal}:j\in[m])=(c_{j,\Gcal'}:j\in[m])$, and thus the output graphs are identical, $\Acal(\Pcal_\Gcal)=\Acal(\Pcal_{\Gcal'})$. Then we have either $\Acal(\Pcal_\Gcal)\ne\Gcal$ or $\Acal(\Pcal_{\Gcal'})\ne\Gcal'$.
\end{proof}

\subsection{Proof of Lemma \ref{lem:parent_intervention}}\label{sec:proof:lem:parent_intervention}

\begin{proof}[Proof of Lemma \ref{lem:parent_intervention}]
    \noindent\textbf{1)} $X_j\in\parent_{\Gcal_{\mid\bm C}}(X_k)\Longrightarrow \big(I_k\not\indep X_j\mid\bm X_{-j},\bm S\big)_{\Aug(\Gcal_{\mid\bm C},I_k)}$.

    In $\Aug(\Gcal_{\mid\bm C},I_k)$, we know $I_k\rightarrow X_k \leftarrow X_j$ forms a collider, which is open given $(\bm X_{-j},\bm S)$. Therefore, $\big(I_k\not\indep X_j\mid\bm X_{-j},\bm S\big)_{\Aug(\Gcal_{\mid\bm C},I_k)}$.

    \noindent\textbf{2)} $\big(I_k\not\indep X_j\mid\bm X_{-j},\bm S\big)_{\Aug(\Gcal_{\mid\bm C},I_k)}\Longrightarrow X_j\in\parent_{\Gcal_{\mid\bm C}}(X_k)$.

    There exists an open path $\pi$ between $I_k,X_j$ given $(\bm X_{-j},\bm S)$ in $\Aug(\Gcal_{\mid\bm C},I_k)$. Since the only adjacent node to $I_k$ is $X_k$, we know $\pi$ starts with $I_k\rightarrow X_k$. Since $X_k$ is in $(\bm X_{-j},\bm S)$, $X_k$ must be a collider on $\pi$. Then $\pi$ must start with $I_k\rightarrow X_k\leftarrow X_l$. Since $X_l$ is not a collider on $\pi$, $X_l$ is not in $(\bm X_{-j},\bm S)$. Consequently, $X_l=X_j$ and $X_j\in\parent_{\Gcal_{\mid\bm C}}(X_k)$.
\end{proof}

\subsection{Proof of Theorem \ref{thm:coverage_selection}}\label{sec:proof:thm:coverage_selection}

\begin{proof}[Proof of Theorem \ref{thm:coverage_selection}]
    Since $P_{\bm Z}$ is Markov to $\Gcal$, under Assumption \ref{asm:context}, conditional dependence implies $d$-connection in $\Gcal_{\mid\bm C}$. If there is no type I error in the parent tests, the rejected null hypotheses in these tests imply $d$-connections in $\Gcal_{\mid\bm C}$. Then it follows from Lemma \ref{lem:parent_intervention} that the detected edges in $\Ccal_\alpha$ exist in $\Gcal_{\mid\bm C}$. Since the edges are among $\bm X$, they are also in $\Gcal_{\bm X}$. Therefore,
    \begin{align*}
        &\Prob\big(\Ccal_\alpha\not\subseteq\Ecal(\Gcal_{\bm X})\big)\\
        \le&\Prob\big(\Ccal_\alpha\not\subseteq\Ecal(\Gcal_{\bm X}), \text{no type I error in the parent tests}\big)+\Prob\big(\text{exist type I error in the parent tests}\big)\\
        \le&\alpha.
    \end{align*}
\end{proof}

\subsection{Proof of Theorem \ref{thm:power_selection}}\label{sec:proof:thm:power_selection}

\begin{proof}[Proof of Theorem \ref{thm:power_selection}]
    Similar to the proof of Theorem \ref{thm:coverage_selection}, under Assumption \ref{asm:context}, if there is no type I error in the parent tests, $\Ccal_\alpha\subseteq\{X_j\rightarrow X_k:k\in[K], X_j\in\parent_{\Gcal_{\mid\bm C}}(X_k)\}$. If Assumption \ref{asm:faithful_parent_selection} also holds, $d$-connection implies conditional distribution shift. If there is no type II error in the parent tests, we detect all the parent information of the intervention targets. Therefore,
    \begin{align*}
        &\Prob\big(\Ccal_\alpha\ne\{X_j\rightarrow X_k:k\in[K], X_j\in\parent_{\Gcal_{\mid\bm C}}(X_k)\}\big)\\
        \le&\Prob\big(\Ccal_\alpha\ne\{X_j\rightarrow X_k:k\in[K], X_j\in\parent_{\Gcal_{\mid\bm C}}(X_k)\}, \text{ all parent tests are correct}\big)\\
        &+\Prob\big(\text{exist type I error in the parent tests}\big)+\Prob\big(\text{exist type II error in the parent tests}\big)\\
        \le&\alpha+{\rm FWER}_2^P.
    \end{align*}
    If $K=d_X$, it is evident that
    \[\Ecal(\Gcal_{\bm X})=\{X_j\rightarrow X_k:k\in[K], X_j\in\parent_{\Gcal_{\mid\bm C}}(X_k)\},\quad\Prob\big(\Ccal_\alpha\ne\Ecal(\Gcal_{\bm X})\big)\le\alpha+{\rm FWER}_2^P.\]
\end{proof}

\subsection{Proof of Proposition \ref{prop:GCM}}\label{sec:proof:prop:GCM}

The proof is adapted from that of Theorem 6 in \cite{shah2020hardness}.

\begin{proof}[Proof of Proposition \ref{prop:GCM}]
    \noindent\textbf{1)} First, we show $\hat\E_{\Dcal}R=\hat\E_{\Dcal}\big\{I-f_P(\bm X)\big\}\big\{Y-g_P(\bm X)\big\}+o_{\Pcal}\big(\frac{1}{\sqrt{N}}\big)$.
    
    Denote $\xi_I=I-f_P(\bm X)$ and $\xi_Y=Y-g_P(\bm X)$, we can decompose $\hat\E_{\Dcal}R$ into
    \begin{align*}
        \hat\E_{\Dcal}R=&\hat\E_{\Dcal}\big\{\xi_I+f_P(\bm X)-\hat f(\bm X)\big\}\big\{\xi_Y+g_P(\bm X)-\hat g(\bm X)\big\}\\
        =&\underbrace{\hat\E_{\Dcal}\xi_I\xi_Y}_{T_1}+\underbrace{\hat\E_{\Dcal}\xi_I\big\{g_P(\bm X)-\hat g(\bm X)\big\}}_{T_2}+\underbrace{\hat\E_{\Dcal}\xi_Y\big\{f_P(\bm X)-\hat f(\bm X)\big\}}_{T_3}\\
        &+\underbrace{\hat\E_{\Dcal}\big\{f_P(\bm X)-\hat f(\bm X)\big\}\big\{g_P(\bm X)-\hat g(\bm X)\big\}}_{T_4}.
    \end{align*}
    
    Under $H_0$, we have 
    \begin{align*}
        \E_{P_{Y\mid\bm X}^{(0)}}(Y\mid\bm X)=\E_{P_{Y\mid\bm X}^{(1)}}(Y\mid\bm X)=g_P(\bm X)\qquad\text{$\frac{1}{2}P_{\bm X}^{(0)}+\frac{1}{2}P_{\bm X}^{(1)}$-almost surely}.
    \end{align*}
    Denote $N=n_0+n_1$, then for $T_1$, we have 
    \begin{align*}
        \E T_1
        =&\frac{n_0}{N}\E_{P_{Y,\bm X}^{(0)}}-f_P(\bm X)\bigg\{\E_{P_{Y\mid\bm X}^{(0)}}(Y\mid\bm X)-g_P(\bm X)\bigg\}\\
        &+\frac{n_1}{N}\E_{P_{Y,\bm X}^{(1)}}\{1-f_P(\bm X)\}\bigg\{\E_{P_{Y\mid\bm X}^{(1)}}(Y\mid\bm X)-g_P(\bm X)\}\bigg\}\\
        =&0.
    \end{align*}

    For $T_2$, note that 
    \[f_P(\bm X)=\E_{P_{Y,\bm X,I}}(I\mid\bm X)=\frac{n_1p_{\bm X}^{(1)}(\bm X)}{n_0p_{\bm X}^{(0)}(\bm X)+n_1p_{\bm X}^{(1)}(\bm X)}.\]
    Given $\hat g$, we have
    \begin{align*}
        &\frac{n_0}{N}\E_{P_{\bm X}^{(0)}}\big\{-f_P(\bm X)\big\}\big\{g_P(\bm X)-\hat g(\bm X)\big\}+\frac{n_1}{N}\E_{P_{\bm X}^{(1)}}\big\{1-f_P(\bm X)\big\}\big\{g_P(\bm X)-\hat g(\bm X)\big\}\\
        =&\frac{n_0}{N}\E_{P_{\bm X}^{(0)}}\frac{-n_1p_{\bm X}^{(1)}(\bm X)}{n_0p_{\bm X}^{(0)}(\bm X)+n_1p_{\bm X}^{(1)}(\bm X)}\big\{g_P(\bm X)-\hat g(\bm X)\big\}+\frac{n_1}{N}\E_{P_{\bm X}^{(1)}}\frac{n_0p_{\bm X}^{(0)}(\bm X)}{n_0p_{\bm X}^{(0)}(\bm X)+n_1p_{\bm X}^{(1)}(\bm X)}\big\{g_P(\bm X)-\hat g(\bm X)\big\}\\
        =&0.
    \end{align*}
    Therefore, $T_2$ satisfies
    \[T_2=\frac{n_0}{N}(\hat\E_{\Dcal_0}-\E_{P_{\bm X}^{(0)}})f_P(\bm X)\big\{\hat g(\bm X)-g_P(\bm X)\big\}+\frac{n_1}{N}(\hat\E_{\Dcal_1}-\E_{P_{\bm X}^{(1)}})\big\{1-f_P(\bm X)\big\}\big\{g_P(\bm X)-\hat g(\bm X)\big\}=o_{\Pcal}\bigg(\frac{1}{\sqrt{N}}\bigg).\]

    Since $\hat f$ is trained on $\Dcal_{\bm X,I}\overset{\triangle}{=}\{(\bm X_i,I_i):i\in[N]\}$, $\big\{Y_i-g_P(\bm X_i)\big\}\big\{f_P(\bm X_i)-\hat f(\bm X_i)\big\}$'s are independent for different $i$ given $\Dcal_{\bm X,I}$, and
    \[\E_{P_{Y\mid\bm X}^{(0)}}\big[\big\{Y_i-g_P(\bm X_i)\big\}\big\{f_P(\bm X_i)-\hat f(\bm X_i)\big\}\mid\Dcal_{\bm X,I}\big]=0.\]
    Therefore, given $\Dcal_{\bm X,I}$, $T_3$ is a summation of independent mean-zero terms.
    \begin{align*}
        \E [T_3^2\mid\Dcal_{\bm X,I}]
        =&\E\big[\big(\hat\E_{\Dcal}\xi_Y\{f_P(\bm X)-\hat f(\bm X)\}\big)^2\mid\Dcal_{\bm X,I}\big]
        =\frac{1}{N}\hat\E_{\Dcal}v_Y(\bm X)\big\{f_P(\bm X)-\hat f(\bm X)\big\}^2=o_{\Pcal}\bigg(\frac{1}{N}\bigg),
    \end{align*}
    where $v_Y(\bm X)=\Var_{P_{Y\mid\bm X}^{(0)}}(Y\mid\bm X)$. Then given $\epsilon>0$,
    \begin{equation}\label{eq:T_3}
        \begin{aligned}
        \sup_{(P^{(0)},P^{(1)})\in\Pcal}\Prob(NT_3^2\ge\epsilon)=&\sup_{(P_{Y,\bm X}^{(0)},P_{Y,\bm X}^{(1)})\in\Pcal}\Prob(NT_3^2\wedge\epsilon\ge\epsilon)\\
        \le&\sup_{(P^{(0)},P^{(1)})\in\Pcal}\epsilon^{-1}\E\E[NT_3^2\wedge\epsilon\mid\Dcal_{\bm X,I}]\\
        \le&\sup_{(P^{(0)},P^{(1)})\in\Pcal}\epsilon^{-1}\E\big[\hat\E_{\Dcal}v_Y(\bm X)\big\{f_P(\bm X)-\hat f(\bm X)\big\}^2\wedge\epsilon\big]\\
        =&o(1),
        \end{aligned}
    \end{equation}
    where the last control follows from Lemma 25 in \cite{shah2020hardness}. Therefore $T_3=o_{\Pcal}\big(\frac{1}{\sqrt{N}}\big)$.

    By Cauchy-Schwarz inequality, we have 
    \[|T_4|\le\bigg(\hat\E_{\Dcal}\big\{f_P(\bm X)-\hat f(\bm X)\big\}^2\hat\E_{\Dcal}\big\{g_P(\bm X)-\hat g(\bm X)\big\}^2\bigg)^{\frac{1}{2}}=o_{\Pcal}\bigg(\frac{1}{\sqrt{N}}\bigg).\]
    Therefore, $\hat\E_{\Dcal}R=\hat\E_{\Dcal}\xi_I\xi_Y+o_{\Pcal}\big(\frac{1}{\sqrt{N}}\big)$.

    \noindent\textbf{2)} Then we show $\widehat{\Var}_{\Dcal}(R)=\frac{n_0}{N}\Var_{P^{(0)}}\big(f_P(\bm X)\xi_Y\big)+\frac{n_1}{N}\Var_{P^{(1)}}\big(\{1-f_P(\bm X)\}\xi_Y\big)+o_{\Pcal}(1)$.

    Similar to Lemma 19 in \cite{shah2020hardness}, we have 
    \[\hat\E_{\Dcal}R=o_{\Pcal}(1),\quad\hat\E_{\Dcal}\xi_I^2\xi_Y^2=\E\hat\E_{\Dcal}\xi_I^2\xi_Y^2+o_{\Pcal}(1).\]
    Then it suffices to show $\hat\E_{\Dcal}R^2=\hat\E_{\Dcal}\xi_I^2\xi_Y^2+o_{\Pcal}(1)$.
    \begin{align*}
        &\hat\E_{\Dcal}R^2-\hat\E_{\Dcal}\xi_I^2\xi_Y^2\\
        =&\underbrace{\hat\E_{\Dcal}\big\{\big(f_P(\bm X)-\hat f(\bm X)\big)^2+2\xi_I\big(f_P(\bm X)-\hat f(\bm X)\big)\big\}\big\{\big(g_P(\bm X)-\hat g(\bm X)\big)^2+2\xi_Y\big(g_P(\bm X)-\hat g(\bm X)\big)\big\}}_{T_5}\\
        &+\underbrace{\hat\E_{\Dcal}\xi_I^2\big\{\big(g_P(\bm X)-\hat g(\bm X)\big)^2+2\xi_Y\big(g_P(\bm X)-\hat g(\bm X)\big)\big\}}_{T_6}\\
        &+\underbrace{\hat\E_{\Dcal}\xi_Y^2\big\{\big(f_P(\bm X)-\hat f(\bm X)\big)^2+2\xi_I\big(f_P(\bm X)-\hat f(\bm X)\big)\big\}}_{T_7}\\
    \end{align*}
    
    For term $T_5$,
    \begin{align*}
        |T_5|\le&3\hat\E_{\Dcal}\big\{f_P(\bm X)-\hat f(\bm X)\big\}^2\big\{g_P(\bm X)-\hat g(\bm X)\big\}^2+\hat\E_{\Dcal}\xi_I^2\big\{g_P(\bm X)-\hat g(\bm X)\big\}^2\\
        &+\hat\E_{\Dcal}\xi_Y^2\big\{f_P(\bm X)-\hat f(\bm X)\big\}^2+4\big|\hat\E_{\Dcal}\xi_I\xi_Y\big\{f_P(\bm X)-\hat f(\bm X)\big\}\big\{g_P(\bm X)-\hat g(\bm X)\big\}\big|.
    \end{align*}
    Then we control each terms respectively.
    \[\hat\E_{\Dcal}\big\{f_P(\bm X)-\hat f(\bm X)\big\}^2\big\{g_P(\bm X)-\hat g(\bm X)\big\}^2\le N\hat\E_{\Dcal}\big\{f_P(\bm X)-\hat f(\bm X)\big\}^2\hat\E_{\Dcal}\big\{g_P(\bm X)-\hat g(\bm X)\big\}^2=o_{\Pcal}(1),\]
    \[\hat\E_{\Dcal}\xi_I^2\big\{g_P(\bm X)-\hat g(\bm X)\big\}^2=o_{\Pcal}(1),\]
    \[\E\big[\hat\E_{\Dcal}\xi_Y^2\big\{f_P(\bm X)-\hat f(\bm X)\big\}^2\mid\Dcal_{\bm X,I}\big]=\hat\E_{\Dcal}v_Y(\bm X)\big\{f_P(\bm X)-\hat f(\bm X)\big\}^2=o_{\Pcal}(1).\]
    Similar to \eqref{eq:T_3}, we have
    \[\hat\E_{\Dcal}\xi_Y^2\big\{f_P(\bm X)-\hat f(\bm X)\big\}^2=o_{\Pcal}(1).\]
    For the last term,
    \begin{align*}
        &\big|\hat\E_{\Dcal}\xi_I\xi_Y\big\{f_P(\bm X)-\hat f(\bm X)\big\}\big\{g_P(\bm X)-\hat g(\bm X)\big\}\big|\\
        \le&\big[\hat\E_{\Dcal}\xi_I^2\xi_Y^2\big]^{\frac{1}{2}}\big[\hat\E_{\Dcal}\big\{f_P(\bm X)-\hat f(\bm X)\big\}^2\big\{g_P(\bm X)-\hat g(\bm X)\big\}^2\big]^{\frac{1}{2}}\\
        =&o_{\Pcal}(1).
    \end{align*}
    Therefore $T_5=o_{\Pcal}(1)$.

    For term $T_6$,
    \[\hat\E_{\Dcal}\xi_I^2\xi_Y\big\{g_P(\bm X)-\hat g(\bm X)\big\}\le\big[\hat\E_{\Dcal}\xi_I^2\xi_Y^2\big]^{\frac{1}{2}}\big[\hat\E\xi_I^2\big\{g_P(\bm X)-\hat g(\bm X)\big\}^2\big]^{\frac{1}{2}}=o_{\Pcal}(1),\]
    and thus $T_6=o_{\Pcal}(1)$. Similarly, $T_7=o_{\Pcal}(1)$, and we complete the proof of part \textbf{2)}.
    
    \noindent\textbf{3)} Finally, we prove Proposition \ref{prop:GCM}.
    
    Denote 
    \[\tilde T_N=\frac{\sqrt{N}\hat\E_{\Dcal}\xi_I\xi_Y}{\{\frac{n_0}{N}\Var_{P^{(0)}}(f_P(\bm X)\xi_Y)+\frac{n_1}{N}\Var_{P^{(1)}}(\{1-f_P(\bm X)\}\xi_Y)\}^{\frac{1}{2}}}.\] 
    Similar to Lemma 18 in \cite{shah2020hardness}, we have 
    \[\lim_{n_0\rightarrow\infty,n_1\rightarrow\infty}\sup_{(P^{(0)},P^{(1)})\in\Pcal}\sup_{t\in\R}\big|\Prob_{P^{(0)},P^{(1)}}(\tilde T_N\le t)-\Phi(t)\big|= 0.\]
    Given part \textbf{1)} and part \textbf{2)}, Lemma 20 in \cite{shah2020hardness} implies
    \[\lim_{n_0\rightarrow\infty,n_1\rightarrow\infty}\sup_{(P^{(0)},P^{(1)})\in\Pcal}\sup_{t\in\R}\big|\Prob_{P^{(0)},P^{(1)}}(T_N\le t)-\Phi(t)\big|= 0.\]
    Therefore
    \[\lim_{n_0\rightarrow\infty,n_1\rightarrow\infty}\sup_{(P^{(0)},P^{(1)})\in\Pcal}\big|\E_{P^{(0)},P^{(1)}}\psi_{N,\alpha}-\alpha\big|=0.\]
\end{proof}

\end{document}